\documentclass[reqno]{amsart}
\usepackage{tkz-euclide}
\usepackage{hyperref}
\usepackage{amsmath}
\usepackage{amssymb}
\usepackage{graphicx}
\usepackage{tikz}
\usetikzlibrary{calc}
\usepackage{xcolor}
\usetikzlibrary{snakes}
\usepackage{amsthm}
\usepackage{float}
\usepackage{mathrsfs}  

\newcommand{\EEE}{\color{black}} 
\newcommand{\eps}{\varepsilon} 
\newcommand{\dx}{{\rm d}x}
\theoremstyle{plain}
\begingroup
\newtheorem{theorem}{Theorem}[section]
\newtheorem{lemma}[theorem]{Lemma}
\newtheorem{proposition}[theorem]{Proposition}

\endgroup

\numberwithin{equation}{section}
\newcommand{\N}{\Bbb N}
\newcommand{\Z}{\Bbb Z}
\newcommand{\R}{\Bbb R}
\newcommand{\C}{\Bbb C}

\theoremstyle{definition}
\begingroup

\newtheorem{remark}[theorem]{Remark}

\endgroup

\usepackage[english]{babel}

\begin{document}

\title[Emergence of rigid polycrystals from atomistic systems]{Emergence of rigid polycrystals from atomistic systems with Heitmann-Radin sticky   disk  energy}


\author{Manuel Friedrich}
\address[Manuel Friedrich]{Applied Mathematics M\"unster, University of M\"unster\\
Einsteinstrasse 62, 48149 M\"unster, Germany.}
\email{manuel.friedrich@uni-muenster.de}

\author{Leonard Kreutz}
\address[Leonard Kreutz]{Applied Mathematics M\"unster, University of M\"unster\\
Einsteinstrasse 62, 48149 M\"unster, Germany.}
\email{lkreutz@uni-muenster.de}

\author{Bernd Schmidt}
\address[Bernd Schmidt]{Institut f{\"u}r Mathematik, Universit{\"a}t Augsburg, 
Universit{\"a}tsstr.\ 14, 86159 Augsburg, Germany.}
\email{bernd.schmidt@math.uni-augsburg.de}

\date{\today}

\begin{abstract}
We investigate the emergence of rigid polycrystalline structures  from atomistic  particle systems. The atomic interaction is governed by  a suitably normalized pair interaction energy, where the `sticky disk'  interaction potential models the atoms as hard spheres that interact when  they  are tangential.  The  discrete energy is frame invariant and no underlying reference  lattice  on the atomistic configurations is assumed. By means of $\Gamma$-convergence, we characterize the asymptotic behavior of  configurations with finite surface energy scaling in the infinite particle limit. The effective continuum theory is described in terms of  a piecewise constant  field  delineating  the local orientation and micro-translation of the configuration. The limiting energy is local and concentrated on the grain boundaries, i.e., on the boundaries of the zones where the underlying microscopic configuration has constant  parameters. The corresponding surface energy density depends on the relative orientation of the two grains, their microscopic translation misfit, and the normal to the interface. We further provide a fine analysis of the surface energies at grain boundaries both for vacuum-solid and solid-solid phase transitions. The latter relies fundamentally on a structure result for grain boundaries showing that due to the extremely brittle setup interpolating boundary layers near cracks are  energetically not favorable. 
\end{abstract}

\maketitle

\section{Introduction}\label{section:introduction}

Most inorganic solids in nature are polycrystals. They are composed of microscopic crystallites (grains) of varying size and orientation in which the atoms are arranged in a periodic, crystalline pattern. In spite of their ubiquity, it remains poorly understood why in these materials such highly regular structures develop at the microscale. The core challenge is to investigate the phenomenon of \emph{crystallization}, i.e., the tendency of atoms to self-assemble into a crystal structure. An ultimate solution would be to understand this as a consequence of the interatomic interactions, where such interactions are determined by the laws of quantum mechanics. 

In view of the current state of research, however, the crystallization question seems out of reach in this generality. It is thus necessary to consider reduced models and to study simplified theories which, however, retain essential features of the interatomic interactions. We follow this route by restricting to zero temperature and by describing our system in the frame of Molecular Mechanics \cite{Molecular, Friesecke-Theil15,Lewars}  as a classical system of particles, whose interaction is given in terms of an empirical pair interaction potential. Moreover, we consider planar rather than three-dimensional models. Given a configuration $X =\{x_1,\ldots,x_N\} \subset \mathbb{R}^2$ consisting of a finite number of particles, their configurational energy $\mathcal{E}(X)$ takes the form 
\begin{align*}
\mathcal{E}(X) = \frac{1}{2}\sum\nolimits_{i \neq j}  V_{\rm pair}  \big(|x_i-x_j|\big),
\end{align*}
where $ V_{\rm pair} \colon [0,+\infty) \to \overline{\mathbb{R}}$ denotes the pair potential. (The factor $1/2$ accounts for double counting.) Such potentials typically are repulsive for close-by atoms while two atoms at larger distances (yet still in their interaction range) exert attractive forces on each other. The latter favors the formation of clusters, whereas the short-range repulsion guarantees that the atoms keep a minimal distance.

Notably, even for commonly used models such as the Lennard-Jones potential, the crystallization problem is still open beyond the one-dimensional setting. (In one dimension, the situation is considerably easier: crystallization at zero temperature for Lennard-Jones interactions is shown in \cite{GardnerRadin:79}. Recent results for positive temperature including an analysis of boundary layers are obtained in \cite{JansenKoenigSchmidtTheil:19a,JansenKoenigSchmidtTheil:19b}. For results on dimers we refer to \cite{Betermin, FriedrichStefanelli}.) First rigorous results for a two-dimensional system have been achieved in \cite{Harborth:74,HeitmannRadin:80,Radin:81}, see also the recent  paper \cite{Lucia}. For the very special choice of the `Heitmann-Radin sticky disk' interaction potential
\begin{align}\label{def:potential}
 V_{\rm sticky} (r) = \begin{cases} +\infty &\text{if } r<1,\\
-1 &\text{if }r=1,\\
0&\text{if } r>1,
\end{cases}
\end{align} 
it has been shown in \cite{HeitmannRadin:80} that ground states, i.e., minimizers under the cardinality constraint $\#X= N$, crystallize: they are subsets of the triangular lattice. The potential $ V_{\rm sticky}$ is pictured schematically in Figure \ref{fig:interactionpotential}.
\begin{figure}[H]
 \includegraphics{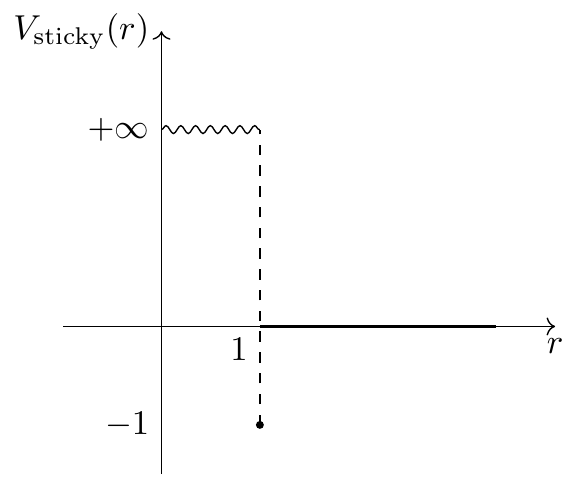}
\caption{The interaction potential $ V_{\rm sticky}$.}
\label{fig:interactionpotential}
\end{figure}
\noindent  On the one hand, it draws its motivation from being the most basic choice of a potential featuring the properties discussed above. On the other hand, it models extremely brittle materials and  might be viewed as an `infinitely brittle' limiting model for more generic interaction potentials, in which the hard core radius, the equlibrium distance, and the interaction range coincide. Slightly more general potentials are discussed in \cite{Radin:81} which, however, do not allow for soft elastic interactions  either. Still only partial results are available for more general potentials or higher dimensions,  see \cite{BlancLewin:15} for a recent survey.  Most noteworthy, \cite{Theil:06,ELi:09} in two and \cite{FlatleyTheil:15} in three dimensions prove that crystalline structures have optimal bulk energy scaling and crystals are ground states subject to their own boundary conditions. Such conditions, however, are insufficient, respectively, prohibitive in view of our goal to investigate the emergence of polycrystals. For this task, it is indispensable to both work at the surface energy scale, which is much finer than the bulk scaling, and to allow for free boundary conditions.

The ground states of sticky disk potentials in two dimensions are by now very well understood, not only on the atomic microscale. In \cite{AuYeungFrieseckeSchmidt:12} their macroscopic shape was identified as being the Wulff shape of an associated crystalline perimeter functional. Fine properties and surface fluctuations were investigated in \cite{Schmidt:13} and quantified in terms of an $N^{3/4}$ law (see the comment below \eqref{eq:energymin}). Sharp constants for this law were then established in \cite{DavoliPiovanoStefanelli:17} and the uniqueness of ground states was characterized in \cite{DeLucaFriesecke:17}. We also mention extensions to other crystals \cite{Mainini,MaininiPiovanoStefanelli:14,DavoliPiovanoStefanelli:16} and  dimers \cite{FriedrichKreutz:19,FriedrichKreutz:20}. By way of contrast, in dimension three or higher the recent results \cite{MaininiPiovanoSchmidtStefanelli:19,CicaleseLeonardi:19,MaininiSchmidt:20} characterize optimal energy configurations within classes of lattices and are in this sense conditional to crystallization. 

The main objective of our contribution is to advance our understanding of (microscopic) crystallization and formation of macroscopic clusters beyond ground states and single crystals. Indeed, all of the aforementioned results ultimately rely on the emergence of a single crystal which is supported on a unique periodic structure. Restricting our analysis to the basic Heitmann-Radin sticky disk potential \eqref{def:potential}, we succeed in deriving a rather complete picture on the formation of general polycrystals by considering the $\Gamma$-limit for the interaction energy in the surface energy regime in the infinite particle limit. (We refer to \cite{Braides:02, DalMaso:93} for an exhaustive treatment of $\Gamma$-convergence.) First  relevant  steps in this direction have been obtained in \cite{DeLucaNovagaPonsiglione:19}, where the authors prove a compactness result for polycrystals and identify the $\Gamma$-limit in the case of a single crystal limiting configuration. In the present work, we prove a full $\Gamma$-convergence result and provide a limiting continuum model consisting of grains that are characterized by a rotation and, in addition, a micro-translation. We also analyze in depth the surface energy of grain boundaries both for vacuum-solid and solid-solid phase transitions.

We proceed to describe our particle model in more detail. The minimal energy of a configuration $X_N = \{x_1,\ldots,x_N\} \subset \mathbb{R}^2$ of $N$ particles has been determined already in \cite{Harborth:74}:  
\begin{align}\label{eq:energymin}
\min\{\mathcal{E}(X_N) \colon \#X_N = N\} = -\lfloor 3N - \sqrt{12N-3}\rfloor \approx -3N +  {\rm O}(\sqrt{N}). 
\end{align}
The leading order term $-3N$ comes from $N - {\rm O}(\sqrt{N})$ atoms in the bulk, each having six neighbors. The lower order term $\sim\sqrt{N}$ is due to missing neighbors of a number ${\rm O}(\sqrt{N})$ of atoms at the boundary and is thus a \emph{surface energy}. (The aforementioned $N^{3/4}$ law quantifies the surprisingly large possible deviations of ground states from the macroscopic Wulff shape which involve a number of $\sim N^{3/4} \gg \sqrt{N}$ particles.) 

As polycrystals  will not be ground states in general, but rather metastable states with surface energy contributions from atoms at individual grain boundaries, we proceed to address the class of all configurations at the finite surface energy scaling, i.e., we consider $X_N \subset \mathbb{R}^2$, $\# X_N = N$, with bounded \emph{normalized energy} \begin{align*}
\frac{\mathcal{E}(X_N)+3N}{\sqrt{N}} = \frac{1}{2\sqrt{N}}\sum\nolimits_{x \in X_N} \Big( 6 +  \sum\nolimits_{y \in X_N \setminus \lbrace x \rbrace}  V_{\rm sticky} \big(|x-y|\big)\Big) 
\end{align*}
as $N \to \infty$. Here, we have subtracted the minimal energy $-3$ per particle times the number of particles and rescaled with $\sqrt{N}$.

The diameter of an $N$-particle configuration $X_N$ with energy given in \eqref{eq:energymin} is $\sim\sqrt{N}$.    To obtain configurations which are contained in a bounded domain, we therefore rescale the configuration by a factor $\eps:=1/\sqrt{N}$, i.e., $X_\varepsilon := \eps X_N$. We then study the asymptotics of the energy $E_\varepsilon(X_\varepsilon)$ where the energy functional $E_\varepsilon$ is defined on finite point sets $X \subset \R^2$ by 
\begin{align}\label{def:energy}
E_\varepsilon(X) = \frac{1}{2}\sum\nolimits_{x \in X}\varepsilon\Big(6+ \sum\nolimits_{y \in X\setminus \lbrace x\rbrace} V_{\rm sticky}\Big( \frac{|x-y|}{\eps} \Big)\Big).
\end{align}
This will allow us to pass to a macroscopic description as $\eps \to 0$. In the following, we consider the energy $E_\varepsilon$ in \eqref{def:energy} \emph{without cardinality constraint} since the energy has already been normalized with respect to the minimal energy per particle.

Our main results are a full $\Gamma$-convergence proof for the functionals $E_\varepsilon$ towards a surface energy functional (Theorem \ref{th: Gamma}) and a detailed analysis of the limiting continuum surface energy density (Proposition \ref{proposition:existence-original} and Theorem \ref{prop: properties of varphi}). We also prove a corresponding compactness result for bounded energy sequences (Theorem \ref{proposition:compactness}), which turns out to be comparatively straightforward. The proofs in fact also provide a rather complete picture of the structure of grain boundaries. We collect these findings of independent interest in Theorem \ref{theorem:grain-boundary}. Our continuum description keeps track not only of the \emph{orientation angles} of various grains but depends additionally on a \emph{micro-translation vector} which in particular measures the translational offset of two lattices with the same orientation. Indeed, the introduction of such an augmented field does not only provide a finer characterization of the continuum limit, but turns out to be crucial when polycrystals with multiple solid-solid grain boundaries are considered. 

The limiting surface energy $\varphi$ is a function of the relative orientation of the two grains, their microscopic translation misfit, and the normal to the interface. For solid-vacuum surfaces it had been identified in \cite{AuYeungFrieseckeSchmidt:12, DeLucaNovagaPonsiglione:19} as the Finsler norm whose unit ball is shaped like a Voronoi cell of the lattice in the solid part. In other words, this is just the surface energy density of the crystal perimeter. For solid-solid interfaces, however, the problem is considerably more subtle as there are atomic interactions across the interface. In softer materials, one expects dislocations to accumulate and elastic strain to concentrate near such grain boundaries. We refer to \cite{ponsiglione, LauteriLuckhaus:16} for recent mathematical developments on substantiating the Read-Shockley formula, see \cite{read1950dislocation}, in such a regime. By way of contrast, within our extremely brittle set-up,   \emph{generically} $\varphi$ turns out to be given by the sum of the solid-vacuum surface energies of the two grains. Here, the term generic refers to the fact that the surface energy may be smaller only for a countable number of mismatch angles between the two lattices, and  corresponding  micro-translations contained in a finite number of spheres. 

We proceed with some comments on the general proof strategy. As it is customary for variational limits with interfacial energies, the density $\varphi$ is expressed in terms of a cell formula minimizing the asymptotic surface energy between two grains separated by a flat grain boundary. In such cell problems, it is instrumental to pass from a mere $L^1$-convergence to fixed boundary values in order to match the $\Gamma$-$\liminf$ and $\Gamma$-$\limsup$ inequalities. Motivated by \cite{barroso.fonseca, fonseca.tartar, sternberg} for vectorial problems in liquid-liquid phase transitions and \cite{conti.schweizer, davoli, kytavsev-ruland-luckhaus2} in solid-solid phase transitions, we use a cut-off construction, the so-called \emph{fundamental estimate}, to replace an asymptotic realization by the exact attainnment of \emph{converging boundary values} in a first step. Here, our extremely brittle set-up on the one hand renders geometric rigidity estimates easier as compared to, e.g., \cite{conti.schweizer, davoli}. On the other hand, it calls for carefully refined cut-off constructions since very small modifications in the configurations may induce a lot of energy. Yet, in contrast to  \cite{conti.schweizer, davoli}, a cell problem with converging boundary data turns out to be insufficient in the presence of \emph{multiple grain boundaries}. Thus, a further step is needed to show that they can be replaced by \emph{fixed boundary values}. Also this passage is subtle due to our rigid set-up which requires a thorough analysis of possible \emph{touching points} of two lattices (points with distance $\varepsilon$).  Finally,   let us also mention  that  related, very general $\Gamma$-convergence results for elastic materials exhibiting discontinuities along surfaces, see e.g.\ \cite{BachBraidesCicalese:20, Caterina, FriedrichSolombrino},  do not apply to our situation.  Most notably, in \cite{FriedrichSolombrino} a model similar to ours featuring rigid grains is considered.  Unfortunately, these results cannot be used in our setting as they  fundamentally rely on continuous surface interactions.

At the core of our proofs, there are two key steps to which we devote Sections \ref{sec:Reduction-two-lattices} and \ref{sec: Part II}, respectively. Firstly, Lemma \ref{lemma:reduction} allows to reduce the cell formula to two lattices only. An expanded version of this observation is detailed in Theorem \ref{theorem:grain-boundary}. It shows that in our brittle set-up there are no interpolating boundary layers at interfaces.  This is done by employing techniques from  graph theory  in order to exclude inclusions of grains of different orientation as the prescribed boundary datum. The basic idea behind its proof is that to each admissible configuration one can associate its bond-graph and for this graph such inclusions induce non-triangular faces which in turn lead to fewer bonds than a competitor without such inclusions. This can be quantified via the \emph{face defect}, see definition \eqref{def:face defect}. Once established, this in particular results in a largely simplified analysis of the interaction energy with vacuum as compared to \cite{DeLucaNovagaPonsiglione:19}, see Lemma~\ref{lemma:vacuumirrational}. More importantly, it is crucial for the second main ingredient of the proof: the quantification of solid-solid interactions with the help of Lemma \ref{lemma: touching},  which clarifies when the surface energy can be smaller than twice the interaction energy with vacuum and plays a pivotal role in order to show that converging boundary values can be replaced by fixed ones. This can be understood as a rigidity theorem for the mismatch-angle between two grains: the generically expected interaction energy can  exceed the grain boundary energy only for finitely many mismatch angles depending on the excess. Its proof relies on the fact that such an energy gap can only occur if the two lattices have many touching points (points with distance $\varepsilon$). This entails that the touching points of the two lattices have to be rather equi-distributed along the interface. This, however, can only happen in a periodic landscape, which reduces the possible mismatch-angle to a finite set. Many  further ingredients of our proofs \EEE are more standard (blow-up, density arguments, fundamental estimate, \ldots), but technically challenging in our case since the energy is very rigid and thus very sensitive to small changes of the configuration. 	

The paper is organized as follows. In Section \ref{section:setting} we introduce the model and present the main results. Section \ref{sec: proof of main results} is devoted to the proofs of compactness and $\Gamma$-convergence. They fundamentally rely on a fine characterization of the surface energy density whose proof is postponed to Sections~\ref{section:surfacetension}--\ref{section:surfacetension2}. In Section \ref{section:surfacetension} we address the fundamental estimate and in Section \ref{section:surfacetension2} we show that converging boundary values can be replaced by fixed ones. Sections \ref{sec:Reduction-two-lattices} and \ref{sec: Part II} are devoted to the reduction of the cell formula to two lattices only and to the characterization of solid-vacuum/solid-solid interactions at grain boundaries, respectively.

\section{Setting of the problem and main results}\label{section:setting}

In this section we introduce our model, give basic definitions, and present our main results.

\subsection{Configurations and  atomistic energy}\label{subsection:energy}

In the following we always assume that $X$  is  a finite subset of $\mathbb{R}^2$. We  denote by $ V_{\rm sticky}  \colon [0,+\infty) \to \overline{\mathbb{R}}$ the Heitmann-Radin potential defined in \eqref{def:potential}, see Figure \ref{fig:interactionpotential}.  By $\eps>0$ we denote the \emph{atomic spacing}.   The  \emph{normalized atomistic energy} $E_\varepsilon$ of a given configuration  $X$ is given by \eqref{def:energy}. The notion \emph{normalized} has been explained in the  introduction and is chosen in such a way that an infinite triangular lattice with spacing $\eps$ has energy zero. Equivalently, the energy can be expressed in terms of the neighborhoods of the atoms. To this end, we introduce  the \textit{neighborhood} of $x \in X$ by
\begin{align}\label{def:neighbourhood}
\mathcal{N}_\varepsilon(x) =\{y\in X : |x-y| = \varepsilon\}.
\end{align}
If $\varepsilon=1$, we omit the subscript $\varepsilon$ and just write $\mathcal{N}(x)$ for simplicity. In view of $ V_{\rm sticky} (r) = \infty$ for $r \in (0,1)$, an elementary geometric argument shows that for configurations $X$ with $E_\varepsilon(X) < +\infty$ there holds 
\begin{align}\label{eq: neighborhood bound}
\# \mathcal{N}_\eps(x) \le 6  \ \ \ \text{for all $x \in X$.}
\end{align}  
In particular, if $\# \mathcal{N}_\eps(x) = 6$, the neighbors form a regular hexagon with center $x$ and diameter $2\eps$. By (\ref{def:potential}) and (\ref{def:energy}) we can now rewrite the energy as 
\begin{align*}
E_\varepsilon(X)= \frac{1}{2}\sum\nolimits_{x\in X} \varepsilon\big(6-\#\mathcal{N}_\varepsilon(x)\big).
\end{align*}
Additionally, for  $X \subset \mathbb{R}^2$ and Borel sets $B \subset \mathbb{R}^2$,  we define a localized version of the energy by 
\begin{align}\label{def:energyneighbourhood}
E_\varepsilon(X,B)= \frac{1}{2}\sum\nolimits_{x\in X\cap B} \varepsilon\big(6-\#\mathcal{N}_\varepsilon(x)\big).
\end{align}

\subsection{Basic definitions}\label{subsection:definitions}
This subsection is devoted to basic notions which we will use throughout the paper.

\noindent \textbf{Notation.} We let $\mathbb{S}^1 = \lbrace x\in\mathbb{R}^2 \colon |x| = 1 \rbrace$.  Given $\nu \in \mathbb{S}^1$, we denote by  $\nu^\bot \in \mathbb{S}^1$  the unit vector obtained by rotating $\nu$ by $\pi/2$ in a clockwise sense.    The scalar product between two vectors $x,y \in \R^2$ is denoted by $\langle x,y \rangle$. Without further notice, we sometimes identify vectors $x \in \R^2$ with elements of $\C$. In particular, we identify rotations  in the plane with  a multiplication with a unit vector in $\mathbb{C}$: namely, the rotation of  $x \in \R^2$ by an angle $\theta \in [0,2\pi)$ is indicated by $e^{i\theta} x$. For $ t\in \R$, we write $\lfloor t \rfloor = \max \lbrace k\in \Z \colon k \le t\rbrace$ and $\lceil t \rceil = \min \lbrace k\in \Z \colon k \ge t\rbrace$.

We denote by $\mathcal{L}^2$ and $\mathcal{H}^1$  the  two-dimensional Lebesgue measure and the one-dimensional Hausdorff measure, respectively.  We write $\chi_E$ for the characteristic function of any $E\subset \R^2$, which is 1 on $E$ and 0 otherwise.  If $E$ is a set of finite perimeter, we denote its \emph{essential boundary} by $\partial^* E$,  see \cite[Definition 3.60]{AFP}.  For $r>0$ and $x \in \mathbb{R}^2$, we denote by $B_r(x)$ the open ball of radius $r$ centered in $x$. For simplicity, we write $B_r$ if $x=0$. Given $A \subset \mathbb{R}^2$, $\tau \in \mathbb{R}^2$, and $\lambda \in \mathbb{R}$, we define
\begin{align}\label{eq: basic set def}
A+ \tau = \{x+\tau : x \in A \},\quad \lambda A = \{\lambda x : x \in A\} \text{ and } (A)_\eps =\{x+y \colon \, x \in A, y \in B_\varepsilon\}.
\end{align}
For $x_1,x_2 \in \mathbb{R}^2$, we define the \emph{line segment between $x_1$ and $x_2$}  by
\begin{align}\label{def:line segment}
[x_1;x_2] =\big\{\lambda x_1 + (1-\lambda) x_2 : \lambda \in [0,1]\big\}.
\end{align}
By $Q^\nu = \lbrace y \in \R^2\colon   -\frac{1}{2} \le \langle y,\nu \rangle < \frac{1}{2}, -\frac{1}{2} \le \langle y,\nu^\bot \rangle < \frac{1}{2} \rbrace$ we denote the half-open unit cube in $\mathbb{R}^2$ with center zero and two sides parallel to $\nu \in \mathbb{S}^1$.  Moreover, we define  the half-cubes
\begin{align}\label{eq: plus-minus}
Q^{\nu,\pm} = \lbrace y \in  Q^{\nu}  \colon \pm  \langle\nu, y\rangle \geq 0\rbrace.
\end{align}
 Here and in the following, we will frequently use the notation $\pm$ to indicate that a property holds for both signs $+$ and $-$.  In a similar fashion, for $x \in \mathbb{R}^2$ and $\rho >0$ we define $Q^\nu_\rho(x) := x + \rho Q^\nu$ and $Q^{\nu,\pm}_\rho(x) := x + \rho Q^{\nu,\pm}$.  For $\rho= 1$, we write $Q^\nu(x)$ instead of $Q^\nu_1(x)$ for simplicity.   For $\varepsilon >0$ and $Q^\nu_\rho(x)$ we introduce the notation of \emph{boundary regions}
\begin{align}\label{eq: eps-rand}
\partial^\pm_\varepsilon  Q^\nu_\rho(x)  = x + \left\{y \in \overline{Q^\nu_{\rho + 10\eps} \setminus Q^\nu_{\rho - 10\eps}} \colon \pm \langle \nu,  y \rangle  \geq 5 \varepsilon\right\},
\end{align}
see also Figure \ref{fig:cellformula} below for an illustration. For $\rho= 1$, we write $\partial^\pm_\varepsilon Q^\nu(x)$ instead of $\partial^\pm_\varepsilon Q^\nu_\rho(x)$.

\noindent\textbf{The triangular lattice.} We define the \emph{triangular lattice} as the set of points given by
\begin{align*}
\mathscr{L}:=\left\{ p+q\omega : p,q\in \mathbb{Z}  \right\},
\end{align*}
where $\omega := \frac{1}{2}+\frac{i}{2}\sqrt{3} \in \C$.

\noindent\textbf{The set of lattice isometries.} We denote by $\mathbb{A}$  the  \emph{set of rotations}  by angles  in $[0,\frac{\pi}{3})$  equipped with the metric of the $1$-dimensional torus, i.e., $\mathbb{A} = \mathbb{R} / \frac{\pi}{3} \mathbb{Z}$. In a similar fashion, we introduce the  \emph{set of translations} $\mathbb{T}= \mathbb{R}^2 / \mathscr{L} = \mathbb{C} / \mathscr{L}$.  We observe that each translation $\tau \in \mathbb{T}$ can be represented by a vector in    
\begin{align}\label{eq: tautautau}
\{\lambda_1 + \lambda_2 \omega \colon 0 \leq \lambda_1 <1,0 \leq \lambda_2<1 \}.
\end{align}
We introduce  the \emph{set of lattice isometries}  by 
\begin{align}\label{eq: state space}
\mathcal{Z}:= \big(\mathbb{A}\times \mathbb{T} \times \{1\}\big) \cup  \{\mathbf{0}\}, 
\end{align}
where for each $\theta \in \mathbb{A}$ and $\tau \in \mathbb{T}$ the triple  $z = (\theta,\tau,1) \in \mathcal{Z}$ represents the  \emph{rotated and translated lattice} 
\begin{align*}
\mathscr{L}(z) =  \mathscr{L}(\theta,\tau,1) := e^{i\theta} (\mathscr{L} + \tau).
\end{align*}
Here, the entry $1$ encodes that a lattice is present. On the contrary,  $\mathbf{0}=(0,0,0) \in \mathbb{A}\times \mathbb{T} \times \{0\} $ represents the empty set, also referred to as  \emph{vacuum} in the following. We set
\begin{align*}
\mathscr{L}(\mathbf{0}) = \emptyset.
\end{align*}
 Note that $\mathbb{A} \simeq \mathbb{S}^1 $ and $\mathbb{T} \simeq \mathbb{S}^1\times \mathbb{S}^1$. Therefore,  the three-dimensional set  $\mathcal{Z}$ can naturally be embedded into $\mathbb{R}^7$. We endow $\mathcal{Z}$ with the product topology, i.e., $z_j=(\theta_j,\tau_j,1) \to z=(\theta,\tau,1)$ if and only if $\theta_j \to \theta$ in $\mathbb{A}$ and $\tau_j\to \tau$ in $\mathbb{T}$.  Moreover, $z_j \to \mathbf{0}$ if and only if $z_j = \mathbf{0}$ for all $j$ large enough.    For a set $A \subset \mathbb{R}^2$, $z \in \mathcal{Z}$, and a configuration $X$ with $E_\varepsilon(X) <+\infty$, we say that $X$ \emph{coincides with the lattice} $ \eps \mathscr{L}(z)$ on $A$, written  $X= \eps  \mathscr{L}(z)$ on $A$, if
\begin{align}\label{eq: coincidence with lattice}
X \cap A =  (\varepsilon\mathscr{L}(z)) \cap A. 
\end{align}

\noindent\textbf{The state space.}  For $A \subset \mathbb{R}^2$, we introduce the space of \emph{piecewise constant functions} $PC(A;\mathcal{Z})$ with values in $\mathcal{Z}$ as functions of the form
\begin{align}\label{eq: PC def}
u = \sum\nolimits_{j=1}^\infty  \chi_{G_j} z_j,
\end{align}
where $\lbrace z_j \rbrace_j \subset \mathcal{Z} \setminus \lbrace \mathbf{0} \rbrace$  are pairwise distinct and $ G_j  \subset A$ are pairwise disjoint sets satisfying  $\mathcal{L}^2\big(\bigcup\nolimits_{j=1}^\infty G_j\big) < \infty$ and 
\begin{align}\label{eq: boundedness cond for jump}
\sum\nolimits_{j=1}^\infty \mathcal{H}^1(\partial^* G_j) < + \infty.
\end{align}
Here, $\lbrace G_j\rbrace_j$ represent the \emph{grains} of the polycrystal and $\lbrace z_j\rbrace_j$ the corresponding \emph{orientation and translation} of the lattice. 
We remark that this space can be identified with
\begin{align}\label{eq: PC-def}
PC(A;\mathcal{Z})=\big\{u \in SBV(A;\mathcal{Z})\colon \,  \nabla u =0,  \, \mathcal{L}^2(\lbrace u \neq \mathbf{0} \rbrace) < + \infty, \,   \mathcal{H}^1(J_u) <+\infty   \big\}. 
\end{align}
 Here, $u$ is a function in $  SBV(A;\mathcal{Z})$ in the sense that $u \in SBV(A;\mathbb{R}^7)$ and $u$ takes values in $\mathcal{Z}$. The jump set of $u$ is denoted by $J_u$.  The one-sided limits of $u$ at a jump point will be indicated by $u^+$ and $u^-$ in the following, and the normal will be denoted by $\nu_u$.  We refer to \cite[Definition 4.21]{AFP} for details on this space.  In a similar fashion, we say $u \in PC_{\rm loc}(\mathbb{R}^2;\mathcal{Z})$ if $u|_A \in PC(A;\mathcal{Z})$ for all compact sets $A \subset \mathbb{R}^2$.

\noindent \textbf{Identification of configurations with piecewise constant functions.}  We now relate atomistic configurations $X$ to the state space defined above.  Consider   $x \in X\cap \mathscr{L}$ such that $\mathcal{N}(x) \subset \mathscr{L}$. Then,  we define the open \textit{lattice  Voronoi cell} of $x$ by
\begin{align}\label{eq: Voronoi}
V(x)= x        +  \frac{1}{\sqrt{3}}  e^{i\pi/6} \, {\rm int}  \big(\mathrm{conv} \{\pm 1,\pm \omega,\pm \omega^2\}\big),
\end{align}
where $\mathrm{conv}\lbrace \cdot \rbrace$ denotes the convex hull of a point set, and int the interior. In a similar fashion, if $x$ and  the points in its neighborhood  $\mathcal{N}_\varepsilon(x)$ lie in a scaled rotated and translated lattice  $\eps \mathscr{L}(z)$, for  $\eps>0$ and $z = (\theta,\tau,1) \in \mathcal{Z}$,  we define $V^z_\varepsilon(x) =  x +    e^{i\theta} \eps V(0)$.  We also point out the implicit dependence on $\tau$ here, since $x = e^{i\theta}(v +\tau)$ for some $v\in \mathscr{L}$.

Given a configuration $X$ with $E_\varepsilon(X) < +\infty$, we now identify $X$ with a suitable function $u\in PC(\mathbb{R}^2;\mathcal{Z})$. Since $E(X) < +\infty$, we have $\#\mathcal{N}_\varepsilon(x) \leq 6$ for all $x \in X$ with equality only if $\lbrace x\rbrace \cup \mathcal{N}_\varepsilon(x) \subset  e^{i \theta(x)}\eps(\mathscr{L} + \tau(x))$ for a unique pair  $(\theta(x),\tau(x))\in \mathbb{A} \times \mathbb{T}$.  We set
\begin{align*}
z(x) = { \big(\theta(x),\tau(x),1\big)  } \in \mathcal{Z} \ \ \  \text{ for all $x \in X$ with $\# \mathcal{N}_\eps(x) = 6$}
\end{align*}
and define $u_\varepsilon^X\colon \mathbb{R}^2\to \mathcal{Z}$ by
\begin{align}\label{def:u}
u_\varepsilon^X(x) := \begin{cases}
z(x) \text{ on $V^{z(x)}_\varepsilon(x)$} & \text{ if } x \in X\text{ with } \#\mathcal{N}_\varepsilon(x)=6 ,\\
 \mathbf{0}  & \text{ else.}
\end{cases} 
\end{align} 
In the following, if no confusion may arise, we write $u_\varepsilon$ instead of $u_\varepsilon^X$.  We note  that this definition is well  posed  in the sense that $V^{z(x_1)}_\varepsilon(x_1) \cap  V^{z(x_2)}_\varepsilon(x_2) = \emptyset$ for all $x_1,x_2 \in X$, $x_1 \neq x_2$, with $\#\mathcal{N}_\varepsilon(x_1) = \#\mathcal{N}_\varepsilon(x_2) = 6$.   In fact, if this were not the case,  one of the six atoms in $\mathcal{N}_\varepsilon(x_1)$  (forming a regular hexagon on $\partial B_\varepsilon(x_1)$)  would have  distance smaller than $1$ to $x_2$. This contradicts $E_\varepsilon(X) < + \infty$.  Clearly,  $u_\eps$  as defined in \eqref{def:u} lies in $PC(\mathbb{R}^2; \mathcal{Z})$.  

The  function  $u_\varepsilon$ for some finite energy configuration   $X$  is illustrated in Figure \ref{fig:interpolation}.  We point out that the translation $\tau(x)$ induces a shift of the Voronoi cells by the vector $\eps e^{i\theta(x)} \tau(x)$. This is the reason why we call the variable $\tau$ a \emph{micro-translation}. 
\begin{figure}[H]
 \includegraphics{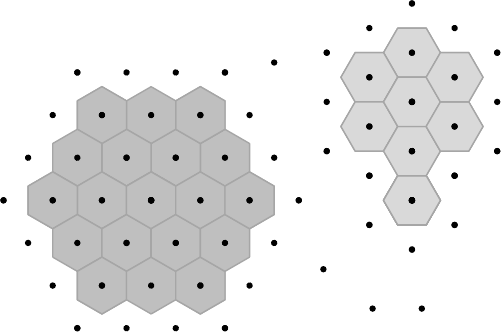}

\caption{A function $u_\eps$ defined in \eqref{def:u}: the different regions $\{u=z\}$ with $z \neq \mathbf{0}$ (here illustrated in different shades of gray) are made of unions of regular hexagons. The complement of those regions is the set $\{u=\mathbf{0}\}$.}
\label{fig:interpolation}

\end{figure}

\noindent\textbf{Convergence:} Let $\{X_\varepsilon\}_\varepsilon$ be a sequence of configurations. We say that $X_\varepsilon \to u$ in $L^1_{\mathrm{loc}}(\mathbb{R}^2)$ if $u_\varepsilon \to u$ in $L^1_{\mathrm{loc}}(\mathbb{R}^2;\mathcal{Z})$, where $u_\varepsilon$ is given by (\ref{def:u}) for $X_\varepsilon$.

\subsection{Main results} \label{subsection:limitfunctional}

We now formulate our main results. We start with a compactness result for sequences of configurations with bounded energy. Recall the definition for convergence of configurations in Subsection \ref{subsection:definitions}.

\begin{theorem}[Compactness]\label{proposition:compactness} Let  $\{X_\varepsilon\}_\varepsilon$ be a sequence of configurations with 
$$\sup\nolimits_{\varepsilon >0}  E_\varepsilon(X_\varepsilon) <+\infty.$$
 Then, there exists a subsequence $ \{\varepsilon_k\}_{k\in \mathbb{N}}$ with  $\varepsilon_k \to 0$ and a function $u \in PC(\mathbb{R}^2;\mathcal{Z})$ such that     $X_{\varepsilon_k} \to u$ in $L^1_{\rm loc}(\mathbb{R}^2)$ as $k \to +\infty$.
\end{theorem}

For  $\varepsilon >0$ and $\nu \in \mathbb{S}^1$, recall the definition of $\partial^\pm_\varepsilon Q^\nu_\rho$ in \eqref{eq: eps-rand}. Recall also the coincidence with a lattice in \eqref{eq: coincidence with lattice}. The following proposition introduces the density $\varphi \colon \mathcal{Z} \times \mathcal{Z} \times\mathbb{S}^1\to [0,+\infty)$  which appears in our continuum limiting functional, see Figure \ref{fig:cellformula} for an illustration.

\begin{proposition}[Density]\label{proposition:existence-original} 
For every $z^+,z^- \in \mathcal{Z}$,   $\nu \in \mathbb{S}^1$,  $x_0 \in \R^2$, and $\rho>0$ there exists
\begin{align}\label{def:varphi}
\varphi(z^+,z^-,\nu)= \lim_{\varepsilon \to 0} \frac{1}{\rho}\min\Big\{E_\varepsilon \big(X,Q^\nu_\rho(x_0)\big)\colon \,    X =  \varepsilon\mathscr{L}(z^\pm) \text{ \rm on } \partial_\varepsilon^\pm Q^\nu_\rho(x_0) \Big\},
\end{align}
and is independent of $x_0$ and $\rho$.
\end{proposition}

\begin{figure}[H]
\centering
 \includegraphics{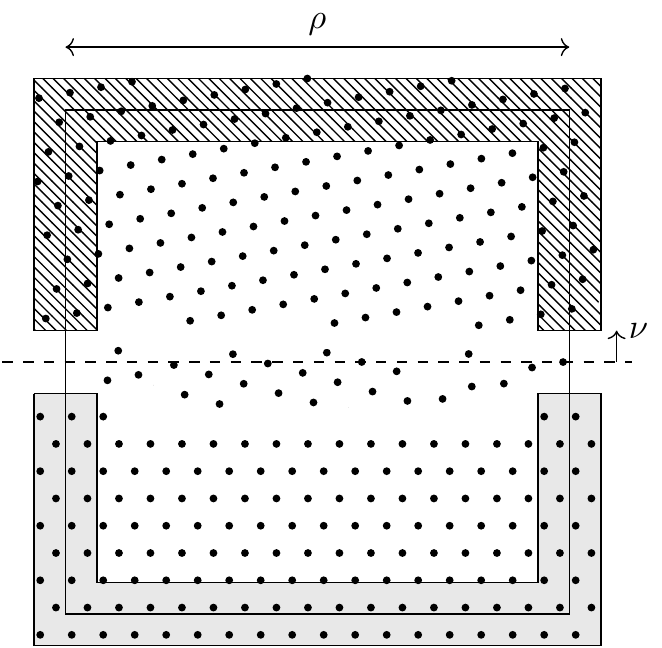}

\caption{Illustration of a competitor for the cell-problem on $Q^\nu_\rho$ in the definition of $\varphi$. On the  light gray hatched and dark gray regions  we have $X=\varepsilon\mathscr{L}(z^\pm)$, respectively.  We point out that the competitor is prescribed in a small neighborhood $ \partial_\varepsilon^-  Q^\nu_\rho \cup  \partial_\varepsilon^+  Q^\nu_\rho$ \emph{both} inside \emph{and} outside of the cube.  (The thickness of the neighborhood is larger than the lattice spacing, see \eqref{eq: eps-rand}. Here, for illustration purposes, it is drawn with thickness $2\varepsilon$ instead of $10\varepsilon$.)}
\label{fig:cellformula}
\end{figure}

The limiting functional  $E \colon PC(\mathbb{R}^2;\mathcal{Z}) \to [0,+\infty)$ is defined by 
\begin{align}\label{eq: limiting func}
E(u) = \int_{J_u} \varphi(u^+(x),u^-(x),\nu_u(x))\,\mathrm{d}\mathcal{H}^1(x). 
\end{align} 
In view of \eqref{eq: PC-def}, functions in  $PC(\R^2;\mathcal{Z})$ lie in $SBV$, and therefore $u^+$, $u^-$, and $\nu_u$ are well defined. The following statement shows that $E$ can be interpreted as the effective limit of the atomistic energies $E_\varepsilon$ in the sense of $\Gamma$-convergence.

\begin{theorem}[$\Gamma$-convergence]\label{th: Gamma}
There holds $E = \Gamma(L^1_{\rm loc})\text{-}\lim_{\varepsilon \to 0} E_\varepsilon$, more precisely:
 
\noindent {\rm (i)} ($\Gamma$-liminf inequality) For each  $u \in PC(\mathbb{R}^2;\mathcal{Z})$  and each sequence $\{X_{\varepsilon}\}_{\eps}$ with $X_{\varepsilon} \to u$ in $L^1_{\rm loc}(\mathbb{R}^2)$ there holds
$$\liminf_{\eps \to 0} E_{\varepsilon}(X_{\varepsilon}) \ge E(u).$$
\noindent {\rm (ii)} ($\Gamma$-limsup inequality) For each $u \in PC(\mathbb{R}^2;\mathcal{Z})$  we find configurations $\lbrace X_{\varepsilon} \rbrace_\eps$ such that $X_{\varepsilon} \to u$ in $L^1_{\rm loc}(\mathbb{R}^2)$ and 
$$\lim_{\eps \to 0} E_{\varepsilon}(X_{\varepsilon}) = E(u).$$ 
\end{theorem}

Here and in the sequel, we follow the usual convention that convergence of the continuous parameter $\eps \to 0$ stands for convergence of arbitrary sequences $\lbrace \eps_k \rbrace_k$ with $\eps_k \to 0$ as $k \to +\infty$.

\begin{remark}[Extension to $L^1$]
{\normalfont 
Defining $E_\varepsilon \colon L^1(\mathbb{R}^2;\mathcal{Z}) \to [0,+\infty]$ by
\begin{align*}
E_\varepsilon(u) 
= \begin{cases} E_\varepsilon(X) &\text{if there exists } X \text{ such that } u = u_\varepsilon^X, \\
  +\infty &\text{otherwise,}
\end{cases}
\end{align*}
and extending $E$ to all of $L^1(\mathbb{R}^2;\mathcal{Z})$ by setting $E(u) = +\infty$ if $u \in L^1(\mathbb{R}^2;\mathcal{Z}) \setminus PC(\mathbb{R}^2;\mathcal{Z})$, in view of Theorem \ref{proposition:compactness}, this indeed implies $\Gamma(L^1_{\rm loc})\text{-}\lim_{\varepsilon \to 0} E_\varepsilon = E$. 
}
\end{remark}

We close this section by providing properties of the density $\varphi$. To this end, we introduce the function $\varphi_{\rm hex}\colon \R^2 \to [0,+\infty)$ defined by
\begin{align}\label{eq: phi-hex-def}
\varphi_{\rm hex}(\nu) =  \frac{2}{\sqrt{3}} \sum\nolimits_{k=1}^3 |\langle  \nu,   \omega^k \rangle|.  
\end{align}
Note that $\varphi_{\rm hex}$ is a  Finsler norm whose unit ball is a regular hexagon in $\R^2$ with vertices in $\frac{1}{2} e^{i\pi/6}\{\pm1,\pm \omega,\pm \omega^2\}$, cf.\ \cite{AuYeungFrieseckeSchmidt:12, DeLucaNovagaPonsiglione:19}.

\begin{theorem}[Properties of $\varphi$]\label{prop: properties of varphi}
Let $\varphi$ be the density given in Proposition \ref{proposition:existence-original}, extended to a function  defined  on $\mathcal{Z} \times \mathcal{Z} \times \mathbb{R}^2$ which is  positively  $1$-homogeneous in the third variable. Then $\varphi$ satisfies the following properties:
\begin{itemize}
\item[(i)] (Solid-vacuum energy) There holds $\varphi(z,\mathbf{0},\nu)= \varphi(\mathbf{0},z,\nu) = \varphi_{\mathrm{hex}}(e^{-i\theta} \nu)$ for all $z = (\theta,\tau,1) \in \mathcal{Z}\setminus \lbrace \mathbf{0}\rbrace$ and $\nu \in \mathbb{S}^1$.\\
\item[(ii)] (Solid-solid energy) There exists a null-set $\mathcal{N}$ in $(\mathcal{Z}\setminus \lbrace \mathbf{0}\rbrace)^2$ (with respect to its six-dimensional Haar measure) 
such that for all pairs $(z^+,z^-) \in (\mathcal{Z}\setminus \lbrace \mathbf{0}\rbrace)^2 \setminus \mathcal{N} $, $z^+ \neq z^-$, and $\nu \in \mathbb{S}^1$ there holds
$$\varphi(z^+,z^-,\nu)= \varphi_{\mathrm{hex}}\big(e^{-i\theta^+} \nu\big)+ \varphi_{\mathrm{hex}}\big(e^{-i\theta^-} \nu\big),$$
and for all $(z^+,z^-)\in \mathcal{N}$, $z^+ \neq z^-$, and $\nu \in \mathbb{S}^1$ there holds
$$\frac{1}{2}\varphi_{\mathrm{hex}}\big(e^{-i\theta^+} \nu\big)+ \frac{1}{2}\varphi_{\mathrm{hex}}\big(e^{-i\theta^-} \nu\big) \leq \varphi(z^+,z^-,\nu)< \varphi_{\mathrm{hex}}\big(e^{-i\theta^+} \nu\big)+ \varphi_{\mathrm{hex}}\big(e^{-i\theta^-} \nu\big), $$
where we write   $z^+ =(\theta^+,\tau^+,1)$ and $z^- =(\theta^-,\tau^-,1)$. 

 \noindent Moreover, there are exceptional sets ${\mathcal{G}_{\mathbb{A}}} \subset \mathbb{A}$ of angles and, for each $\theta \in {\mathcal{G}_{\mathbb{A}}}$, ${\mathcal{G}_{\mathbb{T}}}(\theta) \subset \mathbb{R}^2$ of translation vectors such that ${\mathcal{G}_{\mathbb{A}}}$ is countable and each ${\mathcal{G}_{\mathbb{T}}}(\theta) $ is contained in a finite union of spheres, with  
 
$$\mathcal{N} \subset \big\{ (z^+, z^-) \in (\mathcal{Z}\setminus \lbrace \mathbf{0}\rbrace)^2 \colon \, \theta^+-\theta^- \in {\mathcal{G}_{\mathbb{A}}},\, e^{i\theta^+}\tau^+-e^{i\theta^-}\tau^- \in {\mathcal{G}_{\mathbb{T}}}(\theta^+-\theta^-)\big\}.$$
\item[(iii)] (Convexity) The mapping $\nu \mapsto \varphi(z^+,z^-,\nu) $ is convex  for all $z^+,z^-\in \mathcal{Z}$.
\item[(iv)]  (Rotational invariance) For all $z^\pm=(\theta^\pm,\tau^\pm,1)$, $\nu \in \mathbb{S}^1$, and $\theta \in \mathbb{A}$ there holds
\begin{align*}
\varphi\big((\theta^++\theta,\tau^+,1),(\theta^-+\theta,\tau^-,1),e^{i\theta}\nu\big)=\varphi\big((\theta^+,\tau^+,1),(\theta^-\tau^-,1),\nu\big).
\end{align*}
\item[(v)] (Translational invariance) For all $z^\pm=(\theta^\pm,\tau^\pm,1)$, $\nu \in \mathbb{S}^1$, and $\tau \in \mathbb{T}$ there holds
\begin{align*}
\varphi\Big(\big(\theta^+,\tau^++e^{-i\theta^+}\tau ,1\big), \big(\theta^-,\tau^-+e^{-i\theta^-}\tau,1\big),\nu\Big)=\varphi\big((\theta^+,\tau^+,1),(\theta^-\tau^-,1),\nu\big).
\end{align*}
\end{itemize}
\end{theorem}

We note that the interaction with vacuum, see property (i), has already been addressed in \cite{AuYeungFrieseckeSchmidt:12, DeLucaNovagaPonsiglione:19}. A main novelty of our work lies in the characterization (ii).  For explicit choices of the sets ${\mathcal{G}_{\mathbb{A}}}$ and ${\mathcal{G}_{\mathbb{T}}}(\theta)$ we refer to \eqref{eq: good angles} and the paragraph above Lemma~\ref{lemma:translationproperties}, respectively.  In particular, (ii) states that \emph{generically} the surface energy between two lattices is if each of the two lattices would interact with vacuum. In this case, the continuum energy $E$ of a function $u = \sum\nolimits_{j=1}^\infty  \chi_{G_j} z_j$ corresponds to the \emph{crystalline perimeter} of the grains $\lbrace G_j\rbrace_j$, induced by $\varphi_{\rm hex}$. In the non-generic case    
$(z^+,z^-) \in \mathcal{N}$,  two lattices $\mathscr{L}(z^+)$ and $\mathscr{L}(z^-)$ have many \emph{touching pairs} (i.e., pairs of points with distance $1$) which reduce the energy \eqref{def:energyneighbourhood}. Optimal interfaces for both cases for a normal vector $\nu$ are illustrated in Figure~\ref{fig:low energy}. We remark that the exact characterization of  $\varphi$ seems to be a difficult issue which is beyond the scope of the present analysis. In fact, counting the number of touching pairs depending on the relative orientation of the  two lattices seems to be a non-trivial number theoretic problem, see Remark~\ref{rem: difficult} and Figure \ref{fig:twolattices} below for some details in that direction.  We remark that the properties of ${\mathcal{G}_{\mathbb{A}}}$ and ${\mathcal{G}_{\mathbb{T}}}(\theta)$ imply that $\mathcal{N}$ is of Hausdorff-dimension at most four. \EEE Finally, note that (iv) and (v) express the fact that both the atomistic and the continuum model are \emph{frame indifferent}.

More precisely, our proof in Lemma \ref{lemma: touching} below shows that the non-degeneracy in Theorem~\ref{prop: properties of varphi}(ii) above can be quantized: for every $\eta > 0$ there are only a finite number of differences $\theta$ of lattice rotations and a corresponding finite number of spheres containing the difference of lattice shifts for which 
$$\varphi(z^+,z^-,\nu) \le \varphi_{\mathrm{hex}}\big(e^{-i\theta^+} \nu\big)+ \varphi_{\mathrm{hex}}\big(e^{-i\theta^-} \nu\big) - \eta. $$
These numbers only depend on $\eta$.  Moreover, we remark that the lower bound provided for $\varphi$ is attained, e.g., for $z^- = (0,0,1)$, $z^+ = (0,i,1)$, and $\nu = i$, see Figure \ref{fig:low energy}(c).  (Consider $X = \{x\in\eps\mathscr{L}(0,0,1)\colon \langle x, i \rangle \leq 0\} \cup \{x\in\eps\mathscr{L}(0,i,1)\colon \langle x, i \rangle \geq \eps\}$ in \eqref{def:varphi}.) 

\begin{figure}[H]
 \includegraphics{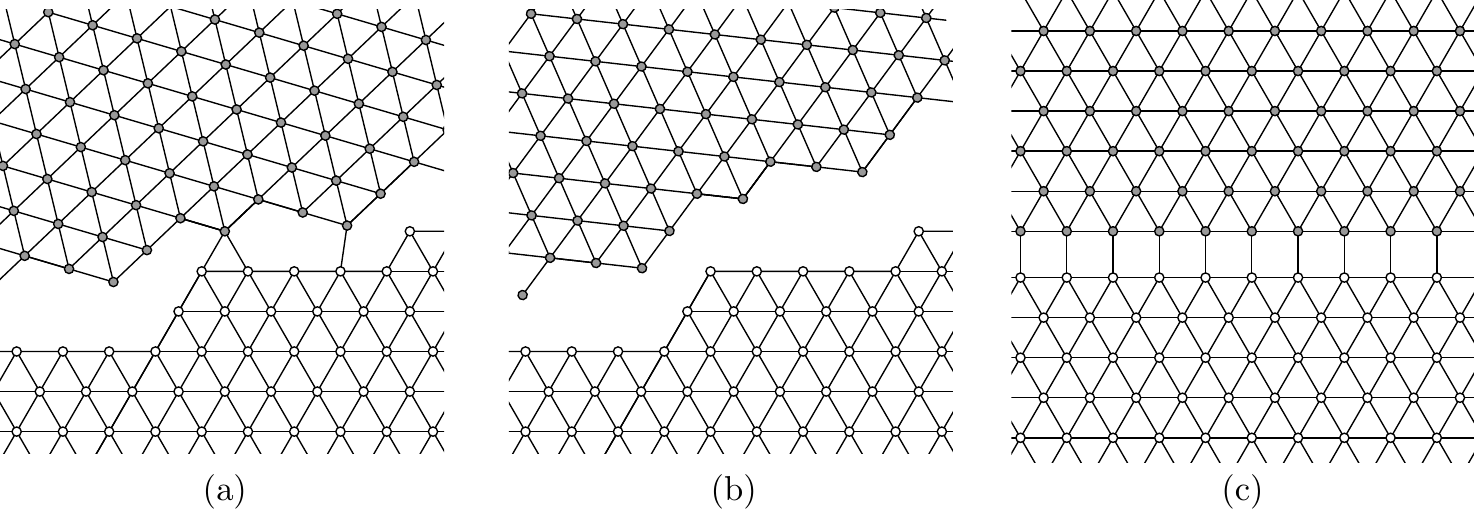}

\caption{Different scenarios of optimal interfaces for a fixed normal $\nu$ and different lattices $\mathscr{L}(z^\pm)$. The dark gray and white points form the lattice $\mathscr{L}(z^+)$ and the lattice $\mathscr{L}(z^-)$, respectively.  Edges are depicted between  points of distance $1$. {\rm (a)}: Two lattices $\mathscr{L}(z^\pm)$  are depicted for which  $\varphi$ is less than twice the interaction energy with the vacuum. {\rm (b)}:  We see  two lattices $\mathscr{L}(z^\pm)$ for which  $\varphi$ is equal to twice the interaction energy with the vacuum. {\rm (c)}: Two lattices for which the lower bound in Theorem \ref{prop: properties of varphi}(ii) is attained.  
 }
\label{fig:low energy}
\end{figure}

\begin{remark} \label{rem: difficult} We finally  point out that for $\theta^+ - \theta^- \in {\mathcal{G}_{\mathbb{A}}}$, $e^{i \theta^+} \tau^+ -  e^{i \theta^-}  \tau^- \in {\mathcal{G}_{\mathbb{T}}}(\theta^+-\theta^-)$ the calculation of $\varphi$ seems to be a difficult problem. In fact, for $e^{i(\theta^+-\theta^-)} =\frac{v_1}{v_2}$ with $v_1,v_2 \in \mathscr{L}$ and $|v_1|=|v_2|$, depending on the factorization of $v_1,v_2$ in $\mathscr{L}$, there may be points $(x,y) \in \mathscr{L}(z^+) \times \mathscr{L}(z^-)$ such that $x,y\notin \mathscr{L}(z^+) \cap \mathscr{L}(z^-)$ and $|x-y|=1$. If this is the case, the relative position of two such atoms is fixed through the prime factors of $v_1,v_2$, respectively. This leads to two major challenges in the calculation of $\varphi$: (i) the characterization of points $(x,y) \in \mathscr{L}(z^+)\times \mathscr{L}(z^-)$ such that $|x-y|=1$ depending on the relative orientation $e^{i(\theta^+-\theta^-)}$ of the two lattices seems to be  a non-trivial number theoretic problem. (ii) even after the characterization of the set of points $(x,y) \in \mathscr{L}(z^+)\times \mathscr{L}(z^-)$ such that $|x-y|=1$ for different normals $\nu$ to the interface, it is not always clear if it is energetically convenient to include such points in the construction of the optimal interface due to their relative orientation. Such a situation is illustrated in Figure \ref{fig:twolattices}.

\end{remark}

\begin{figure}
 \includegraphics{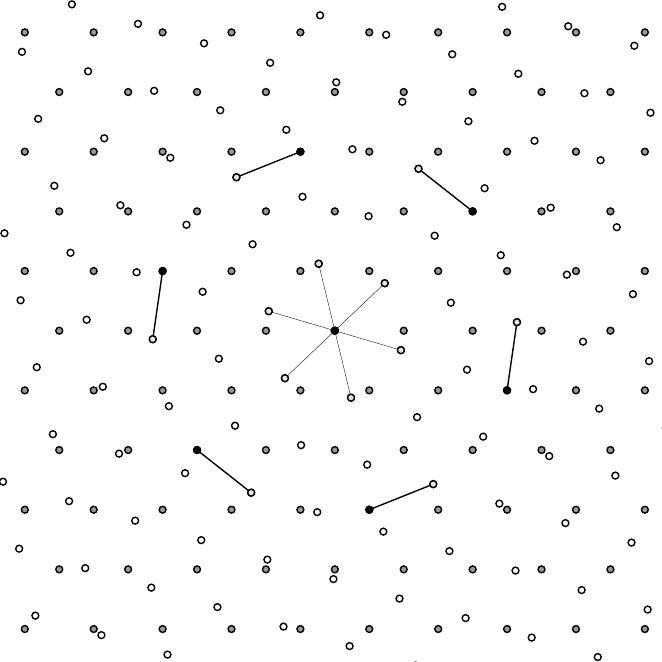}

\caption{Two lattices $\mathscr{L}(z^\pm)$  for which   $\varphi$ is less than twice the interaction energy with  vacuum. The dark gray points form the lattice $\mathscr{L}(z^+)$ and  the white points the lattice $\mathscr{L}(z^-)$.  The black and light gray points are those that are of distance $1$ to the other lattice,  as emphasized  by   an  edge between them.}
\label{fig:twolattices}
\end{figure}

The compactness and $\Gamma$-convergence results will be proved in Section \ref{sec: proof of main results}. The properties of several cell formulas related to $\varphi$,  which are fundamental for the proofs,  are postponed to Sections \ref{section:surfacetension}--\ref{section:surfacetension2}. Finally, the proofs of Proposition \ref{proposition:existence-original}  and Theorem  \ref{prop: properties of varphi} are given in Subsection \ref{sec: properties}.

\section{Proof of the main results}\label{sec: proof of main results}

This section is devoted to the proofs of our main results. We start with some preliminary properties. Then we prove compactness and finally we address the $\Gamma$-convergence result. 

\subsection{Preliminaries}

We state and prove  some elementary properties of the family ${E}_\varepsilon$.  Recall the representation of the energy in \eqref{def:energyneighbourhood} and the definition of sets in \eqref{eq: basic set def}. 

\begin{lemma}[Properties of $E_\eps$]\label{lemma:propertiesofE} Let $\eps >0$ and let $X$ be a configuration with $E_\eps(X) < + \infty$. Then there holds 
\begin{itemize}
\item[(i)] $E_\varepsilon(e^{i\theta}X+\tau,e^{i\theta}A+\tau)=E_\varepsilon(X,A)$  for all $\theta \in [0,2\pi)$,  $\tau \in \mathbb{R}^2$, and $A \subset \mathbb{R}^2$, \vspace{0.1cm}
\item[(ii)]  $E_{\lambda\varepsilon}(\lambda X,\lambda A) = \lambda E_\varepsilon(X,A)$ for all $\lambda >0$ and $A \subset \mathbb{R}^2$, \vspace{0.1cm}
\item[(iii)] $E_\varepsilon(X,A)\leq E_\varepsilon(X,B)$ for all $A \subset B \subset \mathbb{R}^2$, \vspace{0.1cm}
\item[(iv)] $E_\varepsilon(X,A \cup B) = E_\varepsilon(X,B)+E_\varepsilon(X,A)$ for all $A,B \subset \mathbb{R}^2$ with $A \cap B =\emptyset$,
\item[(v)] There exists $C>0$ such that for all $A \subset \mathbb{R}^2$ there holds $\#(X \cap A) \leq C\mathcal{L}^2((A)_\varepsilon)/\varepsilon^2$.         
\end{itemize}
\end{lemma}
\begin{proof} \noindent \emph{Proof of $\mathrm{(i)}$:} Given $\theta \in [0,2\pi)$ and $\tau \in \mathbb{R}^2$, we define $\tilde{x} = e^{i\theta}x + \tau$ for each $x \in \mathbb{R}^2$.  The statement follows by noting that $|\tilde{x}-\tilde{y}|=|x-y|$ for all $x,y \in \mathbb{R}^2$ and $\tilde{x} \in e^{i\theta}A + \tau$ if and only if $x \in A$. This implies $y \in \mathcal{N}_\varepsilon(x)$ if and only if $\tilde{y} \in \mathcal{N}_\varepsilon(\tilde{x})$. 

\noindent \emph{Proof of $\mathrm{(ii)}$:} For $\lambda >0$ and $x \in \mathbb{R}^2$, we define  
 $ x_\lambda=\lambda x$. Clearly, we have  $|x_\lambda-y_\lambda|= \lambda |x-y|$ for all $x,y \in \mathbb{R}^2$ and $ x_\lambda \in \lambda A$ if and only if $x\in A$. This implies $y_\lambda \in \mathcal{N}_{\lambda\eps}(x_\lambda)$ if and only if $y \in \mathcal{N}_\eps(x)$. 

\noindent \emph{Proof of $\mathrm{(iii)}$:} This statement follows from  the fact that for all configurations $X$ with finite energy and all $x \in X$ we have $6-\#\mathcal{N}_\varepsilon(x)\geq 0$ by \eqref{eq: neighborhood bound}.

\noindent \emph{Proof of $\mathrm{(iv)}$:} This follows from the fact that, if $A\cap B=\emptyset$, each term of the summation on the left hand side occurs also in the right hand side and vice versa.

\noindent \emph{Proof of $\mathrm{(v)}$:} Since $X$ is a  configuration with finite energy, there holds $|x-y| \geq \varepsilon$ for all $x,y \in X$, $x\neq y$. Therefore,  $B_{\varepsilon/2}(x) \cap B_{\varepsilon/2}(y) = \emptyset$  for all $x,y \in X$, $x\neq y$. By \eqref{eq: basic set def}, we obtain 
$\bigcup\nolimits_{x \in X \cap A} B_{\varepsilon/2}(x) \subset (A)_\varepsilon$ and therefore
\begin{align*}
\pi  \varepsilon^2/4 \, \#(X\cap A) = \mathcal{L}^2\Big(\bigcup\nolimits_{x \in X \cap A} B_{\varepsilon/2}(x)  \Big)  \leq \mathcal{L}^2 \big((A)_\varepsilon\big).
\end{align*}
From this the claim follows with  $C= 4/\pi$. 
\end{proof}

The following scaling property will be instrumental.

\begin{lemma}[Scaling]\label{lemma:scaling} For $\eps>0$, consider configurations $X_\varepsilon$ satisfying $E_\eps(X_\eps) < +\infty$ and $\lambda X_\eps$ for $\lambda >0$. By  $u_{\lambda\varepsilon}^{\lambda}$  and $u_\varepsilon$ we denote the functions corresponding to $\lambda X_\eps$ and $X_\eps$, respectively, as defined in \eqref{def:u}.  Then,  there holds 
\begin{align}\label{eq:interpolation}
u_{\lambda \varepsilon}^{\lambda}(\lambda x)=u_\varepsilon(x) \text{ for all $x \in \mathbb{R}^2$}.
\end{align}
Moreover, for each bounded $A \subset \mathbb{R}^2$,  we have  $u_{\lambda \varepsilon}^{\lambda} \to u(\lambda^{-1} \, \cdot)$ in $L^1(\lambda A)$ as $\eps \to 0$ if and only if $u_\varepsilon \to u$ in $L^1( A)$.
\end{lemma}
\begin{proof} 
We first prove \eqref{eq:interpolation}. To see this, it suffices to note that $x \in X_\varepsilon$ if and only if $\lambda x \in \lambda X_\varepsilon$, $\#(\mathcal{N}_\varepsilon(x)\cap X)=6$ if and only if $\#(\mathcal{N}_{\lambda\varepsilon}(\lambda x)\cap \lambda X_\varepsilon)=6$, and $(x \cup \mathcal{N}_\varepsilon(x)) \subset \varepsilon e^{i\theta}(\mathscr{L}+\tau)$ if and only if  $(\lambda x \cup \mathcal{N}_{\lambda\varepsilon}(\lambda x)) \subset \lambda \varepsilon e^{i\theta}(\mathscr{L}+\tau)$ for $\theta \in \mathbb{A}$ and $\tau \in \mathbb{T}$. Therefore, in view of \eqref{def:u} and the definition of the Voronoi cells $V_\eps^z(x)$ below \eqref{eq: Voronoi}, \eqref{eq:interpolation} holds true. The equivalence of the convergence follows by a change of variables: we set $y=\lambda x$ and obtain
\begin{align*}
\lambda^2\int_{A} |u_{ \varepsilon}(x)- u(x)| \, \mathrm{d}x =\lambda^2 \int_{A} |u_{\lambda \varepsilon}^\lambda (\lambda x)- u(x)| \, \mathrm{d}x  = \int_{\lambda A} |u_{\lambda \varepsilon}^\lambda (y)- u(\lambda^{-1} y)| \, \mathrm{d}y  
\end{align*}
for every bounded $A \subset \mathbb{R}^2$. 
\end{proof}

\subsection{Compactness}\label{section:compactness}
In this subsection we prove Theorem \ref{proposition:compactness}. As a preparation, we show the following coercivity property.

\begin{proposition}[Coercivity]\label{proposition:coerc} Let $X$ be a  configuration with $E_\eps(X) < +\infty$ and let $A \subset \R^2 $ be a Borel set. Then, there exists a universal $C>0$ such that
\begin{align}\label{eq:energyboundsjump}
   \mathcal{H}^1(J_u\cap {A}) \le CE_\varepsilon(X,(A)_\varepsilon).
\end{align}
where $u$ associated to $X$ is given by \eqref{def:u} and $(A)_\eps$ is defined in \eqref{eq: basic set def}. 
\end{proposition}

\begin{proof} Let $A \subset \R^2 $ be   a Borel set.  Consider $X \subset \mathbb{R}^2$ with  $E_\varepsilon(X) <+\infty$.
In view of \eqref{eq: PC def} and \eqref{def:u}, the function $u$ associated to $X$ can be written in the form $u = \sum\nolimits_{j=1}^\infty  \chi_{G_j} z_j$ for pairwise distinct  $\lbrace z_j \rbrace_j \subset \mathcal{Z} \setminus \lbrace \mathbf{0} \rbrace$ and pairwise disjoint $\lbrace G_j \rbrace_j \subset \mathbb{R}^2$. By \cite[Remark 4.22]{AFP} it suffices to check that
\begin{align}\label{eq: stronger}
\sum\nolimits_{j=1}^\infty \mathcal{H}^1( \partial^* G_j \cap {A}) \le C E_\varepsilon(X,(A)_\varepsilon). 
\end{align} 
Due to  the construction in \eqref{def:u}, each $G_j$ is made of a finite union of regular hexagons  with sidelength $\varepsilon/\sqrt{3}$ \EEE such that at the center of each such hexagon there is an atom  $x \in X$ with $\#\mathcal{N}_\varepsilon(x)=6$. If an edge of such a hexagon is contained in  $\partial^* G_j$,  then there exists a point $y \in \mathcal{N}_\varepsilon(x)$ such that $\#\mathcal{N}_\varepsilon(y) <6$, see Figure \ref{fig:interpolation}.  If the intersection of that edge with ${A}$ is non-empty, then  $y \in (A)_\varepsilon \cap X$, see \eqref{def:neighbourhood} and \eqref{eq: basic set def}. Note that each such $y$ is selected for at most six different edges of hexagons contained in $\partial G^*_j$. By \eqref{def:energyneighbourhood},  this yields
\begin{align*}
\sum\nolimits_{j \in \mathbb{N}} \mathcal{H}^1( \partial^* G_j \cap {A}) \leq  \tfrac{6}{\sqrt{3}}  \eps \, \#\{ y \in X \cap (A)_\varepsilon \colon \#\mathcal{N}_\varepsilon(y) <6\} \leq  \tfrac{12}{\sqrt{3}}  E_\varepsilon(X,(A)_\varepsilon),
\end{align*}
 where we used that each edge of the hexagon has length $\eps/\sqrt{3}$.
\end{proof}

\begin{proof}[Proof of Theorem \ref{proposition:compactness}]
Let $\lbrace X_\eps\rbrace_\eps$ and $\lbrace u_\eps\rbrace_\eps$ be given, as defined in \eqref{def:u}.  Recall that $\mathcal{Z}$ can be embedded into $\mathbb{R}^7$ and that it  is closed and bounded, see \eqref{eq: state space}.  Therefore, for each $B_r$, $r \in \mathbb{N}$, we can use Proposition \ref{proposition:coerc} and a  compactness result for piecewise constant functions, see  \cite[Theorem 4.25]{AFP}, to find a subsequence $\{\varepsilon_k\}_k$ and $u^r \in PC(B_r;\mathcal{Z})$ such that $u_{\eps_k} \to u^r$  in measure and thus also in  $L^1(B_r;\mathcal{Z})$. By lower semicontinuity  there holds $\mathcal{H}^1(J_{u^r} \cap B_r) \le C$ for a constant independent of $r$. By a diagonal argument, we obtain $u\colon \mathbb{R}^2 \to \mathcal{Z}$ with $u=u^r$ on $B_r$ for all $r \in \N$ such that $u_{\eps_k} \to u$ in $L^1_{\rm loc}(\R^2;\mathcal{Z})$. Clearly, $\mathcal{H}^1(J_u)<+\infty$. Thus, to show that $u \in PC(\mathbb{R}^2;\mathcal{Z})$, it remains to check that  $\mathcal{L}^2(\lbrace u \neq \mathbf{0}\rbrace) < + \infty$.  

 Using \eqref{eq:energyboundsjump} with $A=\mathbb{R}^2$, the isoperimetric inequality on $\mathbb{R}^2$, $\mathcal{L}^2( \{u_{\varepsilon_k}\neq \mathbf{0}\}) <+\infty$,  and the fact that $\mathcal{L}^2(\{u \neq \mathbf{0}\})$ is lower semicontinuous with respect to strong $L^1_{\rm loc}$ convergence, we obtain
\begin{align*}
\big(\mathcal{L}^2(\{u \neq \mathbf{0}\})\big)^{1/2} &\leq\liminf_{k \to +\infty} \big(\mathcal{L}^2(\{u_{\varepsilon_k} \neq \mathbf{0}\}) \big)^{1/2}\leq  \liminf_{k \to +\infty} C\mathcal{H}^1(\partial^* \{u_{\varepsilon_k}\neq \mathbf{0}\}) \\&\leq \liminf_{k \to +\infty} C\mathcal{H}^1(J_{u_{\eps_k}})\leq \liminf_{k\to+\infty} C E_{\varepsilon_k}(X_{\varepsilon_k}) <+\infty.
\end{align*}
This implies that $u \in PC(\mathbb{R}^2;\mathcal{Z})$ and concludes the proof.
\end{proof}

\subsection{Lower Bound}\label{section:lowerbound}

This subsection is devoted to the proof of Theorem \ref{th: Gamma}(i).  For the proof, it is instrumental to use a different cell formula. In contrast to imposing boundary conditions as in \eqref{def:varphi}, we require $L^1$-convergence to the function $u^{\nu}_{z^+,z^-} \in PC_{\rm loc}(\mathbb{R}^2;\mathcal{Z})$ defined by 
\begin{align}\label{eq: step function}
u^{\nu}_{z^+,z^-}(x) =      \begin{cases}
z^+ & \text{ if } \langle x, \nu\rangle \ge 0, \\
z^- & \text{ if }\langle x, \nu \rangle < 0,
\end{cases} 
\end{align}
for $x \in \mathbb{R}^2$,  $z^+,z^- \in \mathcal{Z}$, and $\nu \in \mathbb{S}^1$. More precisely, for $z^+,z^- \in \mathcal{Z}$ and $\nu \in \mathbb{S}^1$ we introduce 
\begin{align}\label{def:psi}
\begin{split}
\psi(z^+,z^-,\nu) := \inf\Big\{  \liminf_{\varepsilon \to 0} E_\varepsilon & \big(X_\varepsilon,Q^\nu(y_\eps)\big)\colon \,   y_\eps \in \mathbb{R}^2, \, \lim_{\eps \to 0} \int_{Q^\nu} |u_\eps (x + y_\eps)-   u^{\nu}_{z^+,z^-}(x)| \, {\rm d}x = 0  \Big\},
\end{split}
\end{align}
where $u_\eps$ denotes the function associated to $X_\eps$, as defined in \eqref{def:u}. The density $\psi$ is related to $\varphi$ (see \eqref{def:varphi}) in the following way. 

\begin{proposition}[Relation of $\psi$ and $\varphi$]\label{prop, psi}
 For all $z^+,z^- \in \mathcal{Z}$ and $\nu \in \mathbb{S}^1$ there holds
$$\psi(z^+,z^-,\nu)\ge \varphi(z^+,z^-,\nu).$$

\end{proposition}

We postpone the  proof  of Proposition \ref{prop, psi} to Sections \ref{section:surfacetension}--\ref{section:surfacetension2}. It will follow by combining  Lemma \ref{lemma:cutoff}, Lemma \ref{lemma: calculation}, and  Proposition \ref{proposition:existence}. After a further comment about the definition of $\psi$, we proceed with the proof of the lower bound.

\begin{remark}[Varying cubes in the definition of $\psi$]
{\normalfont
We point out that, in contrast to many other cell formulas in the literature, the \emph{position of the cubes} in \eqref{def:psi} is \emph{not fixed} but may vary along the sequence $\eps \to 0$. This general definition is necessary as the problem is not translation invariant in the variables $z^\pm$, although the discrete energy has such a property, see Lemma \ref{lemma:propertiesofE}(i). To see this issue, consider a sequence $\lbrace X_\eps\rbrace_\eps$ contained in a fixed lattice $X_\eps \subset \eps e^{i\theta}(\mathscr{L}+\tau)$. Then, for a fixed translation $\sigma \in \R^2$, the shifted configurations $\tilde{X}_\eps:= X_\eps + \sigma$ are contained in $\eps e^{i\theta}(\mathscr{L} + \tau_\eps)$, where the translation  $\tau_\eps := (\tau + e^{-i\theta}\sigma/\eps)$ (modulo $\mathscr{L}$)  is in general different from $\tau$ and highly oscillating. This in general implies $\tilde{u}_\eps \neq u_\eps(\cdot - \sigma)$, where ${u}_\eps$ and $\tilde{u}_\eps$ are given in \eqref{def:u}.  This lack of translational invariance is remedied in our approach by minimizing over all possible cell centers. Note that only \emph{a posteriori} we are able to show that the cell formula $\varphi$ is actually independent of the center, see Proposition \ref{proposition:existence-original}.      
}
\end{remark}

\begin{proof}[Proof of Theorem \ref{th: Gamma}(i)]
Let $\lbrace X_{\eps}\rbrace_\eps$ be a sequence with $X_{\eps} \to u$ in $L^1_{\rm loc}(\mathbb{R}^2)$ for  $u \in PC(\mathbb{R}^2;\mathcal{Z})$.  Clearly, it suffices to treat the case  
\begin{align}\label{ineq: equiboundedliminf}
\sup\nolimits_{\varepsilon >0} E_\varepsilon(X_\varepsilon) <+\infty.
\end{align}
We proceed  in two steps. We first identify a limiting measure associated to the discrete configurations (Step 1). Then, we proceed by a blow-up procedure for the jump part of this measure (Step~2). 

\noindent \emph{Step 1: Identification of a limiting measure.} We consider the family of positive measures $\{\mu_\varepsilon\}_\varepsilon$ given by
\begin{align*}
\mu_\varepsilon :=\frac{1}{2} \sum\nolimits_{x\in X} \varepsilon\left(6-\#\mathcal{N}_\varepsilon(x)\right)\delta_x.
\end{align*}
By \eqref{def:energyneighbourhood} we observe that for all open sets $A \subset \mathbb{R}^2$ there holds 
\begin{align}\label{eq: measure energy}
|\mu_\varepsilon|(A) =  \mu_\varepsilon(A) = E_\varepsilon(X_\varepsilon,A).
\end{align}
Therefore, by  \eqref{ineq: equiboundedliminf}  we get  $\sup_{\varepsilon >0} |\mu_\varepsilon|(\mathbb{R}^2) < +\infty$. Thus, as $\R^2$ is locally compact,   up to passing to a subsequence (not relabeled), there exists a positive finite Radon measure $\mu$ such that
\begin{align}\label{conv: weakstarliminf}
\mu_\varepsilon\overset{\ast}{\rightharpoonup}  \mu.
\end{align}
By the Radon-Nykodym Theorem we may decompose $\mu$ into two mutually singular non-negative measures
\begin{align*}
\mu=\xi\mathcal{H}^1|_{J_u} +\mu_s.
\end{align*}
The main point is to prove
\begin{align}\label{eq: goalliminf}
\xi(x_0) \geq \psi(z^+,z^-,\nu) \quad \text{for } \mathcal{H}^1\text{-a.e. } x_0 \in J_u,
\end{align}
where $z^+$ and $z^-$ denote the one-sided limits of $u$ at $x_0$ and $\nu$ denotes the corresponding normal. (For notational convenience, the explicit dependence on $u$ is omitted.)  Once this is shown, the statement follows from \eqref{eq: limiting func}, \eqref{eq: measure energy}, \eqref{conv: weakstarliminf}, and Proposition \ref{prop, psi}. In fact,
$$\liminf_{\eps \to 0} E_{\varepsilon}(X_{\varepsilon}) = \liminf_{\eps \to 0} \mu_\eps(\R^2) \ge \mu(\R^2) \ge \int_{J_u} \xi  \, {\rm d} \mathcal{H}^1 \ge   \int_{J_u} \varphi(z^+,z^-,\nu) \, {\rm d}\mathcal{H}^1  =  E(u).$$

\noindent \emph{Step 2: Blow-up argument.} It remains to prove \eqref{eq: goalliminf}. By the properties of $SBV$-functions and Radon measures we know that for $\mathcal{H}^1$-a.e. $x_0 \in J_u$ there holds 
\begin{itemize}
\item[(a)] 
$ \displaystyle
\lim_{\rho \to 0} \frac{1}{\rho^2}\int_{Q^{\nu}_{\rho}(x_0)}|u(x)-u_{z^+,z^-}^{\nu}(x-x_0)|\, \mathrm{d}x=0,
$
\item[(b)] $\displaystyle \lim_{\rho \to 0} \frac{1}{\rho}\mathcal{H}^1\big(J_u \cap Q^\nu_{\rho}(x_0)\big)=1$,
\item[(c)] $\displaystyle \xi(x_0) = \lim_{\rho \to 0} \frac{\mu(Q^\nu_{\rho}(x_0))}{\mathcal{H}^1\big(J_u \cap Q^\nu_{\rho}(x_0)\big)}$, 
\end{itemize}
see, e.g., \cite[Theorem 2.63, Theorem 3.78, and Remark 3.79]{AFP}.  Here, $u^\nu_{z^+,z^-}$ is defined in \eqref{eq: step function}. \EEE It suffices to prove \eqref{eq: goalliminf} for all $x_0 \in J_u$ such that (a)-(c) hold. We  fix $\rho_n \to 0$ such that $|\mu|(\partial Q^\nu_{\rho_n}(x_0))=0$ for all $n \in \mathbb{N}$. By  \eqref{eq: measure energy}, \eqref{conv: weakstarliminf},   (b), (c),  and the Portmanteu Theorem, we get 
\begin{align*}
\xi(x_0) &= \lim_{\rho \to 0} \frac{\mu(Q^\nu_{\rho}(x_0))}{\mathcal{H}^1(J_u \cap Q^\nu_{\rho}(x_0))} = \lim_{\rho \to 0}\frac{\mu(Q^\nu_{\rho}(x_0))}{\rho} = \lim_{n \to +\infty} \frac{1}{\rho_n} \lim_{\varepsilon \to 0} \mu_\varepsilon\big(Q^\nu_{\rho_n}(x_0)\big)\\& = \lim_{n\to+\infty} \frac{1}{\rho_n} \lim_{\varepsilon \to 0} E_\varepsilon\big(X_\varepsilon,Q^\nu_{\rho_n}(x_0)\big).
\end{align*}
We introduce the configuration $ X_\eps^n := \rho_n^{-1} X_\varepsilon $ and obtain by  Lemma \ref{lemma:propertiesofE}(ii)  (for $\lambda = 1/\rho_n$)
\begin{align}\label{eq: doublelimit}
\xi(x_0)  = \lim_{n\to+\infty} \lim_{\varepsilon \to 0 } E_{\varepsilon/\rho_n} \big( X^n_\eps, Q^\nu(\rho^{-1}_n x_0) \big).
\end{align}
Since $X_{\eps} \to u$ in $L^1_{\rm loc}(\mathbb{R}^2)$, we obtain by definition that $u_\varepsilon \to u$ in $L^1_{\rm loc}(\mathbb{R}^2)$, see the end of Subsection \ref{subsection:definitions}. By $u_\eps^n$ we denote the function corresponding to  $X_\eps^n$. By \eqref{eq:interpolation} we have $u^n_{\varepsilon}( x)=u_\varepsilon(\rho_n x)$ for all $x \in \mathbb{R}^2$.  In particular,  Lemma \ref{lemma:scaling} yields $u_\eps^n \to u^n$ on $Q^\nu(\rho_n^{-1} x_0)$, where $u^n(x):= u(\rho_n x) $ for $x \in \mathbb{R}^2$. By (a),  change of variables,  and the fact that $u^n(x+\rho_n^{-1}x_0) = u(x_0 + \rho_nx)$  as well as $u^{\nu}_{z^+,z^-}(x) = u^{\nu}_{z^+,z^-}(\rho_n  x)$ for $x\in \mathbb{R}^2$, we also get that 
$$\lim_{n\to+\infty} \int_{Q^\nu} | u^n(x+ \rho_n^{-1}x_0) - u^\nu_{z^+,z^-}(x) | \, \dx = \lim_{n\to+\infty} \frac{1}{\rho^2_n}\int_{Q^{\nu}_{\rho_n}(x_0)}|u(x)-u_{z^+,z^-}^{\nu}(x-x_0)|\, \dx   = 0.$$
Therefore, by recalling \eqref{eq: doublelimit} and $u_\eps^n \to u^n$ on $Q^\nu(\rho_n^{-1} x_0)$, by using a standard diagonal argument, we find  an infinitesimal sequence  $\{\varepsilon(n)\}_n$ such that for $ X^n := X^n_{\eps(n)}$ and $u^n:= u^n_{\eps(n)}$ we have
\begin{align}\label{eq: doublelimit2}
\xi(x_0)  = \lim_{n\to +\infty }  E_{\eps_n} \big(X^n,Q^\nu(y^n) \big),
\end{align}
and 
$$\lim_{n \to +\infty}\int_{Q^\nu} |u^n(x   + y^n   )  - u_{z^+,z^-}^\nu(x)    | \, \dx = 0, $$
where  $\eps_n =\eps(n)/\rho_n$ and     $y^n=\rho_n^{-1}x_0$. Since the sequence is admissible in \eqref{def:psi},   \eqref{eq: doublelimit2} implies $\xi(x_0) \geq \psi(z^+,z^-,\nu) $. This shows  \eqref{eq: goalliminf} and concludes the proof. \EEE
 \end{proof}

\subsection{Upper Bound}\label{section:upperbound}

This subsection is devoted to the proof of Theorem \ref{th: Gamma}(ii). The following density result will be instrumental.

\begin{lemma}\label{lemma: sup density}
Let $u \in PC(\mathbb{R}^2;\mathcal{Z})$. Then there exists a sequence $(u_n)_n \subset  PC(\mathbb{R}^2;\mathcal{Z})$ with $u_n \to u$ in $L^1(\R^2)$ and $\limsup_{n\to +\infty} E(u_n) \le E(u)$ such that  each $u_n$ attains only finitely many values and has polygonal jump set, i.e., $J_{u_n}$ consists of finitely many segments.
\end{lemma}

\begin{proof}
Consider $u \in PC(\mathbb{R}^2;\mathcal{Z})$. We proceed in three steps. We first show that $u$ can be approximated by functions with finite support (Step 1). Then, we approximate with functions attaining only finitely many values (Step 2) and finally show that the jump set can be approximated by a finite number of segments (Step 3). Note that it suffices to show that for each $\delta >0$ there exists a function  $u_\delta$ with the desired properties satisfying
\begin{align}\label{eq: delta-prop}
E(u_\delta) \le E(u) + \delta \ \ \ \text{ and } \ \ \  \Vert  u - u_\delta \Vert_{L^1(\R^2)} \le \delta.
\end{align}
We prove \eqref{eq: delta-prop} up to  the multiplication with  a uniform constant  $C > 0$  that is independent of $\delta$. Replacing $u_\delta$ with $u_{\delta/C}$, then yields the result.

 \noindent \emph{Step 1: Reduction to finite support.}  We show that  for every $u \in PC(\mathbb{R}^2;\mathcal{Z})$ and for every $\delta>0$ there exist $R>0$ and $u_\delta\in PC(\mathbb{R}^2;\mathcal{Z})$ such that \eqref{eq: delta-prop} is satisfied and there holds 
 \begin{align}\label{ineq:finitesupport}
\{u_\delta \neq  \mathbf{0}\} \subset B_R.
 \end{align}
 To this end, fix $\delta>0$. Since there holds  $\mathcal{L}^2(\lbrace u \neq \mathbf{0} \rbrace) < +\infty$,  we can  choose $R'>0$ such that 
\begin{align}\label{eq: Rchoice}
\mathcal{L}^2\big(\{u\neq \mathbf{0}\} \cap (\mathbb{R}^2\setminus B_{R'})\big) \leq  \delta.
\end{align}
By the coarea formula and the previous inequality, we can select $R \in (R',R'+1)$ such that
\begin{align}\label{ineq:uneq0BR}
\mathcal{H}^1\big( \{u\neq \mathbf{0}\} \cap \partial B_{R}\big)  \leq \mathcal{L}^2\big(\{u\neq \mathbf{0}\} \cap (B_{R'+1}\setminus B_{R'})\big)\leq \mathcal{L}^2\big(\{u\neq \mathbf{0}\} \cap (\mathbb{R}^2\setminus B_{R'})\big) \leq \delta.
\end{align}
Define $u_\delta \in  PC(\mathbb{R}^2;\mathcal{Z})$ by $u_\delta = u\chi_{B_R}$. Then clearly \eqref{ineq:finitesupport} holds. We choose the orientation of $\nu_{u_\delta}(x)$ for $x \in J_u \cap \partial B_R$ such that $u^+_\delta$ coincides with the trace of $u$ from the interior of $B_R$. As
$\varphi(z,\mathbf{0},\nu) \leq C$ for all $z \in \mathcal{Z}$ and $\nu \in \mathbb{S}^1$ by Theorem \ref{prop: properties of varphi}(i), we use  (\ref{ineq:uneq0BR}) to get  
\begin{align*}
E(u_\delta) & =  \int_{B_{R} \cap J_{u_\delta}} \varphi(u_\delta^+,u_\delta^-,\nu_{u_\delta})\,\mathrm{d}\mathcal{H}^1 + \int_{\partial B_{R} \cap \{u\neq \mathbf{0}\}} \varphi(u_\delta^+,\mathbf{0},\nu_{u_\delta})\,\mathrm{d}\mathcal{H}^1 \\&\leq  \int_{B_{R} \cap J_{u}} \varphi(u^+,u^-,\nu_{u})\,\mathrm{d}\mathcal{H}^1 + C\mathcal{H}^1( \{u\neq \mathbf{0}\} \cap \partial B_{R}) \leq E(u) + C\delta.
\end{align*}
This implies the first inequality of \eqref{eq: delta-prop}. To see the second inequality of \eqref{eq: delta-prop}, note that $|z| \leq C$ for all $z\in \mathcal{Z}$ and therefore by \eqref{eq: Rchoice}
\begin{align*}
 \Vert u_\delta-u\Vert_{L^1(\mathbb{R}^2)} = \Vert u_\delta-u\Vert_{L^1(\mathbb{R}^2\setminus B_R)}\leq C\mathcal{L}^2 \big(\{u\neq \mathbf{0}\} \cap (\mathbb{R}^2\setminus B_{R'} ) \big) \leq C\delta.
\end{align*}

\noindent \emph{Step 2: Reduction to functions attaining finitely many values.} Consider $u \in PC(\mathbb{R}^2;\mathcal{Z})$. By Step~1 we may assume that \eqref{ineq:finitesupport} holds for some $R>0$, i.e., $\{u \neq \mathbf{0}\} \subset B_R$. For each $\delta >0$,  we prove that there exists $u_\delta \in PC(\mathbb{R}^2;\mathcal{Z})$ such that  \eqref{eq: delta-prop} holds and $u_\delta$ attains only finitely many values. Recall by  \eqref{eq: PC def} that $u$ can  be  written in the form $u = \sum\nolimits_{j=1}^\infty  \chi_{G_j} z_j$ for  pairwise distinct  $\lbrace z_j \rbrace_j \subset \mathcal{Z} \setminus \lbrace \mathbf{0} \rbrace$ and pairwise disjoint $\lbrace G_j \rbrace_j \subset \mathbb{R}^2$.   In view of \eqref{eq: boundedness cond for jump}, we can choose  $J_\delta \in \mathbb{N}$  sufficiently large such that
\begin{align}\label{ineq:peruzn}
\sum\nolimits_{j=J_\delta+1}^\infty \mathcal{H}^1\big(\partial^* G_j\big) \le \delta/R.
\end{align}
Note that $G_j \subset B_R$ for all $j \in \N$ since $\{u \neq \mathbf{0}\} \subset B_R$.  Due to the isoperimetric inequality on $B_R$ along with $  \mathcal{L}^2(G_j) \le \mathcal{L}^2(B_R)  =\pi R^2 $ for all $j \in \N$, we obtain 
\begin{align}\label{eq: volume bound}
\sum\nolimits_{j=J_\delta+1}^\infty \mathcal{L}^2(G_j) \le \  \sqrt{\pi R^2} \sum\nolimits_{j=J_\delta+1}^\infty \big(\mathcal{L}^2(G_j)\big)^{1/2} \leq CR\sum\nolimits_{j=J_\delta+1}^\infty \mathcal{H}^1\big(\partial^* G_j\big)\le C\delta,
\end{align}
where $C>0$ is a universal constant. Now we define 
\begin{align*}
u_\delta :=  \begin{cases}\displaystyle u &\displaystyle \text{in } \bigcup\nolimits_{j=1}^{J_\delta} G_j, \\
\mathbf{0} &\text{otherwise.}
\end{cases}
\end{align*}
Then, by \eqref{eq: volume bound} and $\Vert u \Vert_\infty \le C$ we get $\Vert u_\delta-u\Vert_{L^1(\mathbb{R}^2)}= \Vert u_\delta-u\Vert_{L^1(B_R)}\leq C\delta $. Moreover, setting for brevity $\Gamma := \bigcup_{j=J_\delta+1}^\infty\partial^* G_j$ we obtain by \eqref{ineq:peruzn}
\begin{align*}
E(u_\delta) = \int_{J_{u_\delta}}\varphi(u^+_\delta,u^-_\delta,\nu_{u_\delta}) \, \mathrm{d}\mathcal{H}^1 &=  \int_{J_{u_\delta}\cap \Gamma }\varphi(u^+_\delta,u^-_\delta,\nu_{u_\delta})\,\mathrm{d}\mathcal{H}^1  +\int_{J_{u_\delta}\setminus \Gamma}\varphi(u^+_\delta,u^-_\delta,\nu_{u_\delta})\, \mathrm{d}\mathcal{H}^1\\&\leq C \sum\nolimits_{j=J_\delta+1}^\infty \mathcal{H}^1\big( \partial^* G_j \big) +E(u)\leq C\delta+E(u),
\end{align*}
where we have used $\varphi(z_1,z_2,\nu) \leq C$ for all $z_1,z_2 \in \mathcal{Z}$ and $\nu \in \mathbb{S}^1$. Therefore, \eqref{eq: delta-prop} holds, and Step 2 is concluded.  

\noindent \emph{Step 3: Reduction to polyhedral jump  sets.}  Consider $u \in PC(\mathbb{R}^2;\mathcal{Z})$. By Steps 1--2 we can assume that $u$ attains only finitely many values, and its support is contained in $B_R$. By Theorem \ref{prop: properties of varphi}(iii) we get that the mapping $\nu \mapsto \varphi(z_1,z_2,\nu)$ is convex and thus continuous for all $z_1,z_2 \in \mathcal{Z}$. Therefore, by \cite[Theorem 2.1 and Corollary 2.4]{BraidesDensity}  (with $\Omega = B_R$ and $\mathcal{Z}$ being the range of $u$)  we obtain a function $u_\delta \in PC(\mathbb{R}^2;\mathcal{Z})$ with polyhedral jump set such that \eqref{eq: delta-prop} is satisfied. This concludes the proof. 
\end{proof}

We are now in a position to prove Theorem \ref{th: Gamma}(ii).

\begin{figure}
 \includegraphics{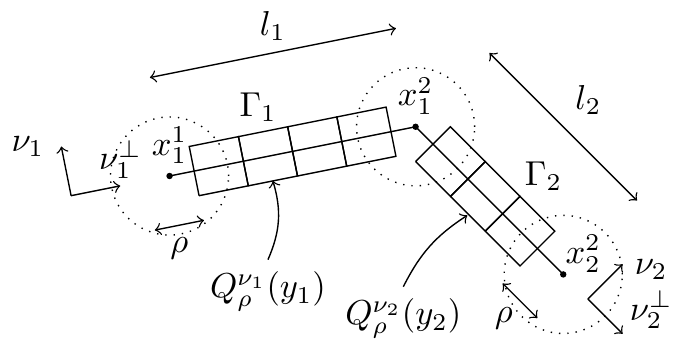}

\caption{The construction for the $\Gamma$-$\limsup$ in the case where the jump set is polyhedral: The part $\Gamma_1 \cup \Gamma_2$ of the jump set is shown. Here, $x_1^2$ equals $x_2^1$. The region $(M)_\delta$ is shown as the dotted circles around the points in $M$. Also the cubes used in the construction to cover the segments $\Gamma_1$ and $\Gamma_2$ are indicated.}
\label{fig:limsup}

\end{figure}

\begin{proof}[Proof of Theorem \ref{th: Gamma}(ii)] 
By Lemma \ref{lemma: sup density} and a general density argument in the theory of $\Gamma$-convergence  (see \cite[Remark 1.29]{Braides:02})\EEE , it suffices to construct recovery sequences for $u \in PC(\mathbb{R}^2;\mathcal{Z})$ such that $u$ attains only finitely many values, and $u$ has a polygonal jump set. Our goal is to prove  that there exists $\lbrace X_\eps \rbrace_\eps$ such that $X_\varepsilon \to u$ in $L^1_{\rm loc}(\R^2) $ and $\limsup_{\varepsilon \to 0} E_\varepsilon(X_\varepsilon) \leq E(u)$.

Let $J_u = \bigcup_{i=1}^N \Gamma_i =\bigcup_{i=1}^N [x_i^1;x_i^2]$, where the sets $\Gamma_i$ are line segments between the points $x_i^1$ and $x_i^2$,  defined in \eqref{def:line segment},   with  length $l_i$, orientation $\nu_i^\bot$, and  normal $\nu_i$. We can assume that the traces $(u^+,u^-) = (u^+_i,u^-_i)$ are constant along each line segment, and that  two segments $\Gamma_i$ and $\Gamma_j$ intersect  at most at  endpoints of $\Gamma_i$ and $\Gamma_j$.  Denote by $M$ the collection of points where at least two of such line segments meet. Fix  $0<\delta<\frac{1}{3} \min\{ |x-y| \colon x,y \in M, \, x\neq y \}$ and choose $\rho \in (0,\delta)$ small enough such that 
 \begin{align}\label{eq: choice of rho}
\rho <  \frac{1}{\sqrt{2}}\mathrm{dist}\Big(\Gamma_i \setminus \big(B_\delta(x_i^1) \cup B_\delta(x_i^2)\big), \, \Gamma_j \setminus \big(B_\delta(x_j^1) \cup B_\delta(x_j^2)\big) \Big) \ \ \ \text{for all $i \neq j$}.
 \end{align}
 This choice of $\rho$ implies that $Q^\nu_\rho(x_1) \cap Q^\nu_\rho(x_2) =\emptyset$ for all $x_1 \in \Gamma_i \setminus (B_\delta(x_i^1) \cup B_\delta(x_i^2))$ and $x_2 \in \Gamma_j \setminus (B_\delta(x_j^1) \cup B_\delta(x_j^2))$, $i\neq j$. As the traces $(u^+,u^-)$ are constant on $\Gamma_i$,   there holds 
 \begin{align}\label{eq:energyu}
 \int_{ \Gamma_i } \varphi(u^+,u^-,\nu_u)\,\mathrm{d}\mathcal{H}^1 = l_i \,\varphi(u^+_i,u^-_i,\nu_i)  \ \text{ for all } i \in \lbrace 1,\ldots, N\rbrace.
 \end{align}
 We define 
 \begin{align*}
  P_i^\rho = \big\{x_i^1 + k\rho \nu^\perp_i \colon\,  k \in \mathbb{N}, \,  0 \leq k \leq \lfloor l_i/\rho \rfloor\big\}, \ \ \  \Gamma_i^\rho = \bigcup\nolimits_{x \in P_i^\rho} Q^\nu_\rho(x), \ \ \  \Gamma_\rho = \bigcup\nolimits_{i=1}^N \Gamma^\rho_i
 \end{align*}
as well as (recall \eqref{eq: eps-rand}) 
 \begin{align*}
H^\eps = \bigcup_{i=1}^N \bigcup_{x \in P_i^\rho}  \big(x +  \partial^H_\varepsilon Q^\nu_\rho\big), \ \  \text{where } \ \partial_\varepsilon^HQ^\nu_\rho :=  \overline{Q^\nu_{\rho + 10\eps} \setminus Q^\nu_{\rho - 10\eps}}   \setminus \big(\partial_\varepsilon^+ Q^\nu_\rho\cup \partial_\varepsilon^- Q^\nu_\rho\big).
 \end{align*}
In view of Proposition \ref{proposition:existence-original}, we can choose $\eps=\eps(\rho, \delta)>0$  sufficiently small such that, for  each $x\in P_i^\rho$,  we can choose a configuration  $X_\eps^x \subset \R^2$ satisfying $X_\eps^x = \varepsilon \mathscr{L}(u^\pm_i) \text{ \rm on } \partial_\varepsilon^\pm Q^\nu_\rho(x)$ and 
 \begin{align}\label{ineq:localoptimality}
 E_\varepsilon(X_\eps^x, Q^\nu_\rho(x)) \leq \rho \,   \varphi(u^+_i, u^-_i, \nu_i) + \delta \rho /l_i.
 \end{align}
 We introduce the configuration 
 \begin{align*}
 X^\delta_\varepsilon =\begin{cases} 
 X_\eps^x &\text{in } Q^\nu_\rho(x)\setminus ((M)_\delta \cup H^\eps), \ \text{for } x \in P_i^\rho \text{ for some } i\in \lbrace 1,\ldots,N\rbrace,\\
 \mathscr{L}(z)&\text{in } \{u=z\} \setminus ((M)_\delta \cup \Gamma_\rho) \text{ for }z \in {\rm Im}(u),\\
 \emptyset &\text{in }(M)_\delta \cup H^\eps, \\
 \end{cases} 
 \end{align*}
see Figure \ref{fig:limsup} for an illustration. Here, $(M)_\delta$ denotes the $\delta$-neighborhood of $M$, see  \eqref{eq: basic set def}, and ${\rm Im}(u)$ denotes the image of $u$.  The set $H^\eps$ is introduced in order to ensure that $E_\varepsilon(X_\varepsilon^\delta)<+\infty$ since atoms in $H^\eps$ of two adjacent cubes could violate the constraint of having at least distance $\varepsilon$. 
 Indeed, by  $X^\delta_\eps = \emptyset$ on $(M)_\delta\cup H^\eps$ and the boundary conditions of $X_\eps^x$, we get  $|x-y| \geq \varepsilon $ for all $x,y \in X_\varepsilon^\delta$, $x\neq y$, and therefore $E_\varepsilon(X_\varepsilon^\delta) <+\infty$.  We  have $\#\mathcal{N}_\varepsilon(x)=6$ for each atom $x \in X^\delta_\eps\setminus ((M)_{\delta+\eps} \cup \Gamma_\rho)$. To see this, we take the boundary conditions of $X_\eps^x$ and the choice of $\rho$ in \eqref{eq: choice of rho} into account. By  \eqref{def:energyneighbourhood} this  implies 
 \begin{align}\label{eq:energyzero}
E_\varepsilon\big(X^\delta_\varepsilon, \mathbb{R}^2 \setminus ((M)_{\delta+\eps} \cup \Gamma_\rho)\big) =0.
\end{align} 
Therefore, it remains to account for the energy contribution inside the cubes $Q^\nu_\rho(x)$, $x\in P_i^\rho$, and the set  $(M)_{\delta+\eps}$.   First,  note that for $ \bar{x}  \in M$ we have that 
 \begin{align}\label{ineq:Cdelta}
 \#\big(X^\delta_\varepsilon \cap B_{\delta +\varepsilon}( \bar{x} )\big) \leq C\delta/\varepsilon.
 \end{align}
 In fact, $(M)_\delta \cap X^\delta_\varepsilon = \emptyset$ by definition and thus $X^\delta_\eps \cap B_\delta( \bar{x}) = \emptyset$.  As $E_\varepsilon(X_\varepsilon^\delta) <+\infty$,    by Lemma \ref{lemma:propertiesofE}(v)  and a simple computation  we get  $\#(X^\delta_\varepsilon \cap (B_{\delta +\varepsilon}(\bar{x})\setminus B_\delta(\bar{x}))) \leq C \varepsilon^{-2}  \mathcal{L}^2(B_{\delta +2\varepsilon}(\bar{x})\setminus B_{\delta-\varepsilon}(\bar{x})) \leq  C\delta/\varepsilon$  for a universal constant $C>0$.  This yields (\ref{ineq:Cdelta})   and then  by \eqref{def:energyneighbourhood} we get  
 \begin{align}\label{ineq:Mdeltapeps}
 E_\varepsilon\big(X^\delta_\varepsilon, (M)_{\delta+\eps}\big) \leq C\delta,
 \end{align}
where $C$ depends also on $\#M$. By  definition of $X^\delta_\eps$,  for $x \in P_i^\rho$ we have that $X_\varepsilon^\delta= X_\varepsilon^x$ in $Q^\nu_{\rho+\varepsilon}(x) \setminus ( H^\eps \cup (M)_\delta)$.  As $E_\varepsilon(X_\varepsilon^\delta) <+\infty$,  we can employ Lemma \ref{lemma:propertiesofE}(v) to deduce that $\#(X_\varepsilon^\delta \cap (H^\eps)_\eps \cap Q^\nu_\rho(x)  ) \leq \varepsilon^{-2}C\mathcal{L}^2( ( H^\eps \cap Q^\nu_\rho(x))_{2\varepsilon}) \leq C$. Hence,  by \eqref{def:energyneighbourhood}  we obtain
\begin{align}\label{ineq: energy cubexPirho}
E_\varepsilon\big(  X^\delta_\eps,  Q^\nu_\rho(x)\big) \leq E_\varepsilon\big(X^x_\varepsilon,Q^\nu_\rho(x)\big) + C\varepsilon
\end{align} 
for all $x \in P_i^\rho$ such that $ \mathrm{dist}(Q^\nu_\rho(x), (M)_\delta)   \geq\varepsilon $.  On the other hand,  for $x \in P_i^\rho$ such that $\mathrm{dist}(Q^\nu_\rho(x), (M)_\delta) <\varepsilon$, we use the estimate in \eqref{ineq:Cdelta} with $\bar{x} \in M$ such that $\mathrm{dist}(Q^\nu_\rho(x), (M)_\delta) = \mathrm{dist}(Q^\nu_\rho(x), B_\delta(\bar{x}))$ (and so $\mathrm{dist}(Q^\nu_\rho(x), (M\setminus\{\bar{x}\})_\delta) > \varepsilon$) and obtain 
 \begin{align}\label{ineq: energy cubexPirhoclose}
E_\varepsilon\big( X^\delta_\varepsilon,Q^\nu_\rho(x)\big) \leq E_\varepsilon\big(X^x_\varepsilon,Q^\nu_\rho(x)\big) + C(\varepsilon+\delta). 
\end{align} 
Consequently, using \eqref{ineq:localoptimality}, \eqref{ineq: energy cubexPirho}--\eqref{ineq: energy cubexPirhoclose},  and Lemma \ref{lemma:propertiesofE}(iii),   we obtain
\begin{align*}
\sum\nolimits_{x \in P_i^\rho} E_\varepsilon\big(X^\delta_\varepsilon,Q^\nu_\rho(x)\setminus  (M)_{\delta+\eps} \big)&\leq \sum\nolimits_{x \in  \tilde{P}_i^\rho  } E_\varepsilon\big(X^\delta_\varepsilon,Q^\nu_\rho(x)\big) \leq \rho \lfloor l_i/\rho\rfloor \,  \varphi(u^+_i,u^-_i,\nu_i) + C\delta +C\varepsilon/\rho, 
\end{align*}
where we have set $\tilde{P}_i^\rho = \{ x \in P_i^\rho : Q^\nu_\rho(x)\not \subset (M)_{\delta+\eps} \}$. Here, $C$ depends on $N$ and $\# M$, but is independent of $\eps$, $\delta$, and $\rho$. Thus, by choosing $\eps$ small enough with respect to $\rho$ (i.e., with respect to $\delta$) we get by (\ref{eq:energyu}) that 
\begin{align}\label{ineq:uepsjump}
\sum\nolimits_{x \in P_i^\rho} E_\varepsilon\big(X^\delta_\varepsilon,Q^\nu_\rho(x)\setminus  (M)_{\delta+\eps} \big) \leq l_i \,\varphi(u^+_i,u^-_i,\nu_i)+C\delta =\int_{\Gamma_i \cap J_u} \varphi(u^+,u^-,\nu_u)\,\mathrm{d}\mathcal{H}^1+C\delta.
\end{align}
Now, by  Lemma \ref{lemma:propertiesofE}(iv),    \eqref{eq:energyzero}, \eqref{ineq:Mdeltapeps}, and \eqref{ineq:uepsjump}  we conclude  
\begin{align*}
E_\varepsilon(X_\varepsilon^\delta) &\leq  \sum_{i=1}^N \sum_{x \in P_i^\rho} E_\varepsilon\big(X^\delta_\varepsilon,Q^\nu_\rho(x)\setminus  (M)_{\delta+\eps} \big) + E_\varepsilon\big(X_\varepsilon^\delta, (M)_{\delta+\eps}\big) + E_\varepsilon\big(X^\delta_\varepsilon, \mathbb{R}^2 \setminus ((M)_{\delta+\eps} \cup \Gamma_\rho)\big)\\& \leq\sum\nolimits_{i=1}^N  \int_{\Gamma_i \cap J_u} \varphi(u^+,u^-,\nu_u)\,\mathrm{d}\mathcal{H}^1 + CN\delta= \int_{J_u} \varphi(u^+,u^-,\nu_u)\,\mathrm{d}\mathcal{H}^1+CN\delta.
\end{align*}
By choosing $\delta = \delta(\varepsilon) \to 0$ sufficiently slowly as  $\varepsilon \to 0$ we obtain $X_\varepsilon^{\delta(\varepsilon)} \to u$ in $L^1_{\rm loc}(\R^2)$ (see Subsection \ref{subsection:definitions} for the definition of this convergence) and 
\begin{align*}
\limsup_{\varepsilon \to 0} E_\varepsilon(X_\varepsilon^{\delta(\varepsilon)}) 
&\leq \int_{J_u} \varphi(u^+,u^-,\nu_u)\,\mathrm{d}\mathcal{H}^1.
\end{align*}
This concludes the proof. 
\end{proof}

To conclude the proof of the main theorems, it remains to show  Proposition \ref{proposition:existence-original}, Theorem \ref{prop: properties of varphi}, and Proposition \ref{prop, psi}.  This is subject to the next sections.

\section{Cell formula Part I: Relation of $L^1$-convergence and boundary values}\label{section:surfacetension}

In this first part about cell formulas, we show that the condition of $L^1$-convergence as given in the cell formula $\psi$, see \eqref{def:psi}, can be replaced by converging boundary values. More precisely, in this section we consider  $\Phi : \mathcal{Z}\times \mathcal{Z} \times \mathbb{S}^1\to [0,+\infty)  $ defined by 
\begin{align}\label{eq: Phi def}
  \Phi(z^+,z^-,\nu) 
  = \min \Big\{\liminf_{\varepsilon \to 0}\inf\Big\{  E_\varepsilon(X_\eps,Q^\nu(y_\eps)) \colon 
     \, y_\eps \in \, \mathbb{R}^2, \, &X_\eps =  \eps  \mathscr{L}(z_\varepsilon^\pm) \text{ on } \partial^\pm_\varepsilon Q^\nu(y_\eps) \Big\} \colon  \notag \\ 
      & \ \  \ 
      \lbrace z^\pm_\varepsilon\rbrace_\eps \subset \mathcal{Z} \text{ with }  z^\pm_\varepsilon \to z^\pm\Big\},
\end{align} 
where the identity  $X_\eps =\varepsilon\mathscr{L}(z_\varepsilon^\pm)$  is defined in \eqref{eq: coincidence with lattice} and $\partial^\pm_\varepsilon Q^\nu(y_\eps)$ in \eqref{eq: eps-rand}. This means that near the boundary of the cube the  configuration is contained in at most two different lattices $\eps \mathscr{L}(z^\pm_\eps)$. (Less is possible if $z^\pm_\eps = \mathbf{0}$.) We note that the minimum in \eqref{eq: Phi def} is attained by  a standard diagonal sequence argument. Our aim is to prove the following statement. 

\begin{lemma}[Relation of $\psi$ and $\Phi$]\label{lemma:cutoff} Let $z^+,z^- \in \mathcal{Z}$ and $\nu \in \mathbb{S}^1$. Then
\begin{align}\label{ineq:approxboundaryconditions}
\begin{split}
\psi(z^+,z^-,\nu) \geq \Phi(z^+,z^-,\nu).
\end{split}
\end{align}
\end{lemma}
In Section \ref{section:surfacetension2}, we will prove $\Phi(z^+,z^-,\nu) = \varphi(z^+,z^-,\nu)$ for all $z^+, z^- \in \mathcal{Z}$ and $\nu \in \mathbb{S}^1$,  see  Lemma \ref{lemma: calculation}, and  Proposition \ref{proposition:existence}. This along with Lemma \ref{lemma:cutoff} will conclude the proof of Proposition \ref{prop, psi}.

As it is customary in the analysis of cell formulas, the proof of Lemma \ref{lemma:cutoff} crucially relies on a cut-off argument which allows to construct configurations attaining the boundary values. Whereas for problems on Sobolev spaces this is usually achieved by a convex combination of functions, our discrete problem is considerably  more delicate. In fact, on the one hand, the system is quite flexible due to the rotational and translational invariance of the atomistic energy, cf.\ Lemma \ref{lemma:propertiesofE}(i). On the other hand, the system is very rigid as  small changes in the configuration may induce a lot of energy  due to the discontinuous interaction potential, see \eqref{def:potential}.    This calls for a refined cut-off construction.

The  construction fundamentally relies on the fact that the energy of an optimal sequence in \eqref{def:psi} is concentrated asymptotically arbitrarily close to the interface. (Similar properties can be observed in related phase transition problems, see e.g.\ \cite{conti.fonseca.leoni, conti.schweizer, davoli}.)  As a preliminary step, we need to show that in the definition of $\psi$ we may replace cubes by rectangles. To this end, we introduce half-open rectangles with sides parallel to $\nu$ by
\begin{align}\label{eq: reci}
R^\nu_{l,h}(y) = y +  \Big\{ x \in \mathbb{R}^2\colon \,  -\frac{h}{2} \le  \langle x, \nu \rangle < \frac{h}{2}, \  -\frac{l}{2} \le  \langle x, \nu^\bot\rangle < \frac{l}{2} \Big\}, 
\end{align}
where $y \in \mathbb{R}^2$, and  $l,h >0$. We simply write $R^\nu_{l,h}$ instead of $R^\nu_{l,h}(y)$ if the rectangle is centered at $y=0$. Recall the definition in \eqref{eq: step function}.

\begin{lemma}[Density  $\psi$ on rectangles]\label{lemma:scalinginvariance}
For all $z^+,z^- \in \mathcal{Z}$, all $\nu \in \mathbb{S}^1$, and all $l,h>0$ there holds 
\begin{align}\label{eq:translational invariance}
\psi(z^+,z^-,\nu) = \inf\Big\{  \liminf_{\varepsilon \to 0}& \, \frac{1}{l} E_\varepsilon\big(X_\varepsilon,R^\nu_{l,h}(y_\eps)\big)\colon  \, y_\eps \in \mathbb{R}^2, \,  \lim_{\eps \to 0} \int_{R^\nu_{l,h}} |u_\eps (x + y_\eps)-   u^{\nu}_{z^+,z^-}(x)| \, {\rm d}x = 0  \Big\}.
\end{align}
\end{lemma}
\begin{proof} For convenience, we denote the function on the right hand side of \eqref{eq:translational invariance} in the variables $(z^+,z^-,\nu,l,h)$ by $\Psi$. We will use certain scaling properties of $\Psi$:
\begin{align}\label{item:scalingofpsi}
\Psi(z^+,z^-,\nu,\lambda \ell,\lambda \kappa) =\Psi(z^+,z^-,\nu,\ell,\kappa) \text{ for all $\lambda >0$}.
\end{align}
\begin{align}\label{item:monotonicityofpsi}
\Psi(z^+,z^-,\nu,\ell, \kappa) \leq \Psi(z^+,z^-,\nu,\ell,\lambda \kappa) \text{ for all $\lambda \geq 1$}.
\end{align}
\begin{align}\label{item:averagingofpsi}
\Psi(z^+,z^-,\nu,\ell,\kappa) \leq \Psi(z^+,z^-,\nu,\lambda \ell,\kappa) \text{ for all $\lambda \in \mathbb{N}$}.
\end{align}
\begin{align}\label{item:monotonicityofpsi2}
\ell_1\Psi(z^+,z^-,\nu,\ell_1,\kappa) \leq \ell_2 \Psi(z^+,z^-,\nu,\ell_2,\kappa) \text{ for all $0<\ell_1 \leq \ell_2$}.
\end{align}
We postpone the  proof of (\ref{item:scalingofpsi})--(\ref{item:monotonicityofpsi2}) to Step 3 of the proof, and first derive the statement. 

\noindent \emph{Step 1: Independence of $l$.}  We start by  proving the independence of the length $l$, i.e.,
\begin{align}\label{eq:invariancepsil}
\Psi(z^+,z^-,\nu,l, h) = \Psi(z^+,z^-,\nu,\mu l,h)
\end{align}
for all $\mu>0$. To this end, consider  first $\mu \in \mathbb{N}$. Using  (\ref{item:scalingofpsi})  and then (\ref{item:monotonicityofpsi}) with $\lambda = \mu$, $\ell = l$, and $\kappa = h /\mu$, we obtain
\begin{align*}
\Psi(z^+,z^-,\nu, \mu l, h) = \Psi\left(z^+,z^-,\nu, l, h/\mu\right) \leq  \Psi\left(z^+,z^-,\nu, l, h\right).
\end{align*}
By (\ref{item:averagingofpsi}) for $\mu\in \mathbb{N}$ there holds $\Psi(z^+,z^-,\nu, \mu l, h)  \ge  \Psi(z^+,z^-,\nu,l,h)$. Combining the estimates we get 
\begin{align}\label{eq:llambda1}
\Psi(z^+,z^-,\nu, \mu l, h)  =  \Psi(z^+,z^-,\nu,l,h)
\end{align}
for $\mu \in \N$. Now substituting $l$ with $\frac{l}{\mu}$ in the previous equation, we obtain
\begin{align}\label{eq:llambda2}
\Psi(z^+,z^-,\nu,  l, h)  = \Psi(z^+,z^-,\nu,l/\mu,h)
\end{align}
for all $\mu \in \mathbb{N}$ and $l >0$. Hence, due to (\ref{eq:llambda1}) and (\ref{eq:llambda2}),  equality \eqref{eq:invariancepsil} holds for all  $\mu \in \mathbb{Q}^+$. 

Now, for general $\mu >0$, we  take a sequence $\{\mu_n\}_n \subset \mathbb{Q}$ such that  $\mu_n \leq \mu_{n+1}$ for all $n \in \mathbb{N}$ and $\mu_n \to \mu$. By  \eqref{item:monotonicityofpsi2} and the fact that  \eqref{eq:invariancepsil} holds for all $\mu \in \mathbb{Q}$ we obtain 
\begin{align*}
\Psi(z^+,z^-,\nu,  l, h) = \Psi(z^+,z^-,\nu, \mu_n l, h) \leq \frac{\mu}{\mu_n}\Psi(z^+,z^-,\nu, \mu l, h).
\end{align*}
Taking $n\to +\infty$ we obtain
\begin{align}\label{eq:llambda3}
\Psi(z^+,z^-,\nu,  l, h) \leq \Psi(z^+,z^-,\nu, \mu l, h).
\end{align}
This yields one inequality in \eqref{eq:invariancepsil}.  Applying \eqref{eq:llambda3} for $\lambda$ in place of $\mu$  and $l/\lambda$ in place of $l$ we also get  
\begin{align*}
\Psi(z^+,z^-,\nu,  l, h)=\Psi(z^+,z^-,\nu, \lambda l/\lambda, h) \geq  \Psi(z^+,z^-,\nu,  l/\lambda, h).
\end{align*}
If we choose $\lambda= \mu^{-1}$, we get the other inequality in  (\ref{eq:invariancepsil}).

\noindent \emph{Step 2: Independence of $h$.}  Let $\mu >0$. By first applying (\ref{item:scalingofpsi}) and then (\ref{eq:invariancepsil}) we obtain
\begin{align*}
\Psi(z^+,z^-,\nu,l,h) = \Psi\left(z^+,z^-,\nu,\mu l,\mu h\right) = \Psi(z^+,z^-,\nu,  l,\mu h).
\end{align*}
This yields the desired independence of the height $h$.  

\noindent \emph{Step 3: Proof of  (\ref{item:scalingofpsi})--(\ref{item:monotonicityofpsi2}).} It remains to prove (\ref{item:scalingofpsi})--(\ref{item:monotonicityofpsi2}).

\noindent \emph{Step 3.1: Proof of \eqref{item:scalingofpsi}.} Fix $\lambda, \ell,\kappa >0$. Let  $X_\varepsilon \subset \mathbb{R}^2$ and $y_\eps \in \mathbb{R}^2$ be given such that $\lim\nolimits_{\eps \to 0} \int_{R^\nu_{\ell,\kappa}} |u_\eps (x + y_\eps)-   u^{\nu}_{z^+,z^-}(x)| \, {\rm d}x = 0$ and 
\begin{align}\label{eq:psi1}
\Psi(z^+,z^-,\nu,\ell,\kappa) =  \liminf_{\varepsilon \to 0} \frac{1}{\ell}E_\varepsilon\big(X_\varepsilon,R^\nu_{ \ell , \kappa}(y_\eps)\big).
\end{align}
(By a standard diagonal sequence argument the infimum on the right hand side of \eqref{eq:translational invariance} is attained.)  Set $X_\varepsilon^\lambda =\lambda X_\varepsilon$. By   \eqref{eq:interpolation} we get that the corresponding functions $u_{\lambda\eps}^\lambda$, see \eqref{def:u}, satisfy $u^\lambda_{\lambda\varepsilon}( x)=u_\varepsilon(\lambda^{-1} x)$ for all $x \in \mathbb{R}^2$.  Change of variables $y = \lambda^{-1}x$  and $u^{\nu}_{z^+,z^-}(y) = u^{\nu}_{z^+,z^-}(\lambda y)$  imply 
$$\lim_{\eps \to 0} \int_{R^\nu_{\lambda \ell,\lambda \kappa}} |u_{\lambda\eps}^\lambda (x +  \lambda y_\eps)-   u^{\nu}_{z^+,z^-}(x)| \, {\rm d}x = \lim_{\eps \to 0} \lambda^2 \int_{R^\nu_{\ell,\kappa}} |u_\eps (y + y_\eps)-   u^{\nu}_{z^+,z^-}(y)| \, {\rm d}y= 0.$$
Using Lemma \ref{lemma:propertiesofE}(ii) along with \eqref{eq:psi1} and the definition of $\Psi$, we obtain
\begin{align*}
\Psi(z^+,z^-,\nu, \lambda \ell,\lambda \kappa) & \leq \liminf_{\varepsilon \to 0} \frac{1}{\lambda \ell}E_{\lambda\varepsilon}\big(X_\varepsilon^\lambda,R^\nu_{\lambda \ell ,\lambda \kappa}( \lambda y_\eps   )\big) \\& = \liminf_{\varepsilon \to 0} \frac{1}{\ell}E_\varepsilon\big(X_\varepsilon,R^\nu_{ \ell , \kappa}(y_\eps)\big) = \Psi(z^+,z^-,\nu,\ell,\kappa).
\end{align*}
By exchanging $\lambda$ with $\frac{1}{\lambda}$ and $\ell,\kappa$ with $\lambda \ell,\lambda \kappa$, respectively, we obtain \eqref{item:scalingofpsi}.

\noindent \emph{Step 3.2: Proof of \eqref{item:monotonicityofpsi}.}
Fix $\lambda \geq 1$ and $\ell, \kappa >0$. Consider  $X_\varepsilon \subset \mathbb{R}^2$ and $y_\eps \in \mathbb{R}^2$  such that $\lim\nolimits_{\eps \to 0} \int_{R^\nu_{\ell,\lambda \kappa}} |u_\eps (x + y_\eps)-   u^{\nu}_{z^+,z^-}(x)| \, {\rm d}x = 0$ and 
\begin{align*}
\Psi(z^+,z^-,\nu,\ell,\lambda\kappa) =  \liminf_{\varepsilon \to 0} \frac{1}{\ell}E_\varepsilon\big(X_\varepsilon,R^\nu_{ \ell , \lambda \kappa}(y_\eps)\big).
\end{align*}
By Lemma \ref{lemma:propertiesofE}(iii)  and the definition of $\Psi$ we get 
\begin{align*}
\Psi(z^+,z^-,\nu,\ell, \kappa) \leq \liminf_{\varepsilon \to 0} \frac{1}{\ell}E_\varepsilon\big(X_\varepsilon,R^\nu_{\ell, \kappa}(y_\eps)\big)\leq \liminf_{\varepsilon \to 0} \frac{1}{\ell}E_\varepsilon\big(X_\varepsilon,R^\nu_{\ell,\lambda \kappa}(y_\eps)) = \Psi(z^+,z^-,\nu,\ell,\lambda \kappa\big).
\end{align*}

\noindent \emph{Step 3.3: Proof of \eqref{item:averagingofpsi}.} Let $\lambda \in \mathbb{N}$ and $\ell, \kappa >0$. Consider  $X_\varepsilon \subset \mathbb{R}^2$ and $y_\eps \in \mathbb{R}^2$  such that 
\begin{align}\label{eq: rect-conv}
\lim\limits_{\eps \to 0} \int_{R^\nu_{\lambda\ell, \kappa}} |u_\eps (x + y_\eps)-   u^{\nu}_{z^+,z^-}(x)| \, {\rm d}x = 0
\end{align}
 and 
\begin{align}\label{eq:psi1-2}
\Psi(z^+,z^-,\nu,\lambda\ell,\kappa) =  \liminf_{\varepsilon \to 0} \frac{1}{\lambda\ell}E_\varepsilon\big(X_\varepsilon,R^\nu_{ \lambda\ell ,  \kappa}(y_\eps)\big).
\end{align}
We decompose the half-open rectangle $R^\nu_{ \lambda\ell ,  \kappa}(y_\eps)$ into pairwise disjoint half-open rectangles of the form 
\begin{align*}
R^\nu_{ \lambda\ell ,  \kappa}(y_\eps) = \bigcup\nolimits_{j=0}^{\lambda-1} R^\nu_{ \ell, \kappa} (y_j^\varepsilon), 
\end{align*}
where $y_j^\varepsilon = y_\eps + \frac{2 j - \lambda + 1}{2} \ell \nu^\perp$.  Now, using Lemma \ref{lemma:propertiesofE}(iv), we derive that there exists $j_0$ such that
\begin{align}\label{eq: good choice}
E_\varepsilon\big(X_\varepsilon,R^\nu_{ \ell , \kappa}(y^\eps_{j_0})\big)&\le \frac{1}{\lambda} \sum\nolimits_{j=0}^{\lambda-1}E_\varepsilon\big(X_\varepsilon,R^\nu_{\ell, \kappa}(y_j^\varepsilon) \big) = \frac{1}{\lambda}E_\varepsilon\big(X_\varepsilon,R_{\lambda \ell, \kappa}^\nu(y_\eps) \big).
\end{align}
By \eqref{eq: rect-conv} and the fact that $u^{\nu}_{z^+,z^-}(x) = u^{\nu}_{z^+,z^-}(x +t\nu^\bot)$ for all $x\in\mathbb{R}^2$ and   $t \in \mathbb{R}$, see \eqref{eq: step function}, we get that $\lim\nolimits_{\eps \to 0} \int_{R^\nu_{\ell, \kappa}} |u_\eps (x +  y^\eps_{j_0})-   u^{\nu}_{z^+,z^-}(x)| \, {\rm d}x = 0$. By the definition of $\Psi$, \eqref{eq:psi1-2}, and \eqref{eq: good choice} this yields
\begin{align*}
\Psi(z^+,z^-, \nu, \ell, \kappa ) &\leq \liminf_{\varepsilon \to 0} \frac{1}{\ell} E_\varepsilon\big(X_\varepsilon,R^\nu_{ \ell , \kappa }(y_{j_0}^\eps)\big)  \leq \liminf_{\varepsilon\to 0}  \frac{1}{\lambda \ell}E_\varepsilon\big(X_\varepsilon,R_{\lambda \ell,\kappa}^\nu(y_\eps)\big) = \Psi(z^+,z^-, \nu, \lambda \ell, \kappa).
\end{align*}
This implies \eqref{item:averagingofpsi}.

\noindent\emph{Step 3.4: Proof of \eqref{item:monotonicityofpsi2}.} Let $0<\ell_1 \leq \ell_2$. Consider $X_\varepsilon \subset \mathbb{R}^2$ and $y_\eps \in \mathbb{R}^2$  such that $\lim\nolimits_{\eps \to 0} \int_{R^\nu_{\ell_2, \kappa}} |u_\eps (x + y_\eps)-   u^{\nu}_{z^+,z^-}(x)| \, {\rm d}x = 0$ and 
\begin{align*}
\Psi(z^+,z^-,\nu,\ell_2,\kappa) =  \liminf_{\varepsilon \to 0} \frac{1}{\ell_2}E_\varepsilon\big(X_\varepsilon,R^\nu_{ \ell_2 , \kappa}(y_\eps)\big).
\end{align*}
By using Lemma \ref{lemma:propertiesofE}(iii) along with $\ell_2 \geq \ell_1$   and the definition of $\Psi$ we get
\begin{align*}
\Psi(z^+,z^-,\nu,\ell_1, \kappa) &\leq \liminf_{\varepsilon \to 0} \frac{1}{\ell_1}E_\varepsilon\big(X_\varepsilon,R^\nu_{\ell_1, \kappa}(y_\eps)\big) \leq  \liminf_{\varepsilon \to 0} \frac{1}{\ell_1}E_\varepsilon\big(X_\varepsilon,R^\nu_{\ell_2, \kappa}(y_\eps)\big)  \\&= \frac{\ell_2}{\ell_1}\liminf_{\varepsilon \to 0} \frac{1}{\ell_2}E_\varepsilon\big(X_\varepsilon,R^\nu_{\ell_2, \kappa}(y_\eps)\big)= \frac{\ell_2}{\ell_1}\Psi(z^+,z^-,\nu,\ell_2, \kappa).
\end{align*}
This yields \eqref{item:monotonicityofpsi2} and concludes the proof.  
\end{proof}

We now proceed with the proof of Lemma \ref{lemma:cutoff}.

\begin{proof}[Proof of  Lemma \ref{lemma:cutoff}]
In view of \eqref{def:psi}, we can choose a subsequence in $\eps$ (not relabeled) and configurations  $X_\eps \subset \mathbb{R}^2$ and  $y_\eps \in \mathbb{R}^2$  such that $ \lim\nolimits_{\eps \to 0} \int_{Q^\nu} |u_\eps (x + y_\eps)-   u^{\nu}_{z^+,z^-}(x)| \, {\rm d}x = 0$ and 
\begin{align}\label{eq: optimal!}
\psi(z^+,z^-,\nu) =   \lim_{\varepsilon \to 0} E_\varepsilon\big(X_\varepsilon,Q^\nu(y_\eps)\big).
\end{align}
 We perform a refined cut-off construction and split the proof into several steps.  As explained above, the construction is quite delicate due to the fact that the energy is very sensitive to   small changes of the configurations. First, we use Lemma \ref{lemma:scalinginvariance} to prove that the energy of $X_\eps$ concentrates around a strip close to the limiting interface (Step 1). This allows us to select one dominant component on each side of the interface, i.e.,  on the upper and the lower half-cube (Step 2). Here, the notion ``component'' refers to a subset of a specific triangular lattice.  

Our  goal in the subsequent steps is to modify the configuration $X_\eps$ such that it coincides with these lattices near the boundary of the upper and lower half-cube, respectively.  In Step 3, we give a precise cardinality estimate on the number of points that differ from the lattices of the two dominant components in terms of ${\rm o}(\eps^{-2})$. In Step 4, we select a ``good layer'' where we can modify our configuration. ``Good'' means here that,  in that layer, the configuration coincides with the lattice of the dominant component up to $o(\varepsilon^{-1})$ atoms. In Step 5, we show that the configuration constructed in Step 4 is an asymptotic energy lower bound for the original configuration. Finally, in Step 6, we conclude by observing that the constructed configuration is a competitor in the definition of  $\Phi$. We will perform this construction under the assumption that in both the upper and the lower half-cube there exist (dominant) lattices. The case of vacuum calls for small adaptions which are described at the end in Step 7.

\noindent \emph{Step 1: The energy concentrates near the line $\{\langle\nu, (x-y_\eps)\rangle=0\}$.}  Recall \eqref{eq: reci}. We show that for all $\delta \in (0,1)$ there holds  
\begin{align}\label{eq: liminfoutsidedelta0}
\lim_{\varepsilon \to 0} E_\varepsilon\big(X_\varepsilon,Q^\nu(y_\eps)\setminus R^\nu_{1,\delta}(y_\eps)\big) = 0.
\end{align}
By Lemma \ref{lemma:propertiesofE}(iii), Lemma \ref{lemma:scalinginvariance}, \eqref{eq: optimal!},  and the fact that $\lbrace X_\varepsilon\rbrace_\eps$ is admissible in the definition of $\psi$ on $R^\nu_{1,\delta}$, see \eqref{eq:translational invariance},  we obtain
\begin{align*}
\psi(z^+,z^-,\nu) \leq \liminf_{\varepsilon \to 0} E_\varepsilon\big(X_\varepsilon,R^\nu_{1,\delta}(y_\eps)\big) \leq \lim_{\varepsilon \to 0} E_\varepsilon\big(X_\varepsilon,Q^\nu(y_\eps)\big) =  \psi(z^+,z^-,\nu).
\end{align*}
Lemma \ref{lemma:propertiesofE}(iv) then implies
\begin{align*}
0 &\leq \limsup_{\varepsilon \to 0} E_\varepsilon\big(X_\varepsilon,Q^\nu(y_\eps)\setminus R^\nu_{1,\delta}(y_\eps)\big)= \limsup_{\varepsilon \to 0} \Big( E_\varepsilon\big(X_\varepsilon,Q^\nu(y_\eps)\big)-E_\varepsilon\big(X_\varepsilon,R^\nu_{1,\delta}(y_\eps)\big)\Big)\\& 
\leq \lim_{\varepsilon \to 0} E_\varepsilon\big(X_\varepsilon,Q^\nu(y_\eps)\big) - \liminf_{\varepsilon\to 0}E_\varepsilon\big(X_\varepsilon,R^\nu_{1,\delta}(y_\eps)\big)=0. 
\end{align*}
This yields (\ref{eq: liminfoutsidedelta0}) and concludes Step 1.

\medskip

In order to shorten the notation, we omit the dependence on the center $y_\eps$ and simply write $Q^\nu_\rho$ instead of $Q^\nu_\rho(y_\eps)$ for $\rho>0$ and $R^\nu_{1,\delta}$ instead of $R^\nu_{1,\delta}(y_\eps)$. For brevity, we also define (omitting the center $y_\eps$) the rectangles  $P_{\delta,\eps}^\pm = Q^{\nu,\pm}_{1-\varepsilon}\setminus  R^\nu_{1-\eps,\delta}$, where $Q^{\nu,\pm}_{1-\varepsilon}$ is defined below \eqref{eq: plus-minus}.  We will prove all auxiliary statements along the proof for the upper half-cube $Q^{\nu,+}$ only since the arguments for the lower one are analogous.  In the following, $\delta \in (0,1)$ is fixed sufficiently small. Without restriction, we may suppose that  $\varepsilon \ll  \delta$.

\medskip

\noindent\emph{Step 2: Single dominant component in the upper and lower half.} We prove that there exist sequences $\{z^\pm_\varepsilon\}_\varepsilon \subset \mathcal{Z}$ such that $z^\pm_\varepsilon \to z^\pm$ and 
\begin{align} \label{ineq : dominantphase}
\mathcal{L}^2\big(\{u_\varepsilon \neq z^\pm_\varepsilon\} \cap P_{\delta,\eps}^\pm\big)\leq C E_\varepsilon\big(X_\varepsilon,Q^\nu\setminus R^\nu_{1,\delta/2}\big),
\end{align}
where $C>0$ is a universal constant independent of $\eps$. 

Recall by  \eqref{eq: PC def} and \eqref{def:u} that the function $u_\eps$ can be written as $u_\eps = \sum\nolimits_{j=1}^\infty  \chi_{G^\eps_j} z^\eps_j$ for pairwise distinct $\lbrace z^\eps_j \rbrace_j \subset \mathcal{Z} \setminus \lbrace \mathbf{0} \rbrace$ and pairwise disjoint $\lbrace G^\eps_j \rbrace_j \subset \mathbb{R}^2$. By Proposition \ref{proposition:coerc} (more precisely, see \eqref{eq: stronger}), \eqref{eq: basic set def},  and Lemma \ref{lemma:propertiesofE}(iii) we have
\begin{align}\label{ineq:compactnessprop}
\sum\nolimits_{j=1}^\infty\mathcal{H}^1( \partial^* G^\eps_j \cap P_{\delta,\eps}^+ ) \le CE_\varepsilon\big(X_\varepsilon,(P_{\delta,\eps}^+)_\eps\big) \leq C E_\varepsilon\big(X_\varepsilon,Q^\nu\setminus R^\nu_{1,\delta/2}\big),
\end{align}
where in the last step we used $(P_{\delta,\eps}^+)_\eps \subset Q^\nu\setminus R^\nu_{1,\delta/2}$.  We also define the vacuum inside $Q^\nu$ by $G^\eps_0 := Q^\nu \setminus \bigcup_{j=1}^\infty G^\eps_j$. By the relative isoperimetric inequality (see e.g.\ \cite[Theorem 2, Section 5.6.2]{EvansGariepy92}), there exists $c>0$ such that for all $j \in \N_0$ there holds
\begin{align}\label{ineq:isoperimetric}
\nonumber
 \min \big\{ \mathcal{L}^2(G_j^\eps \cap P_{\delta,\eps}^+), \mathcal{L}^2(  P_{\delta,\eps}^+\setminus G_j^\eps) \big\} &\leq   \min \big\{ \mathcal{L}^2(G_j^\eps \cap P_{\delta,\eps}^+), \mathcal{L}^2(  P_{\delta,\eps}^+\setminus G_j^\eps) \big\}^{1/2} \mathcal{L}^2(P_{\delta,\eps}^+)^{1/2} \\&\leq c\mathcal{H}^1(\partial^* G_j^\eps\cap P_{\delta,\eps}^+),
\end{align}
 where we used $\mathcal{L}^2(P_{\delta,\eps}^+)\leq 1$. 
 (Note  that the theorem in the reference above is stated and proved in a ball, but that the argument only relies on Poincar\'e inequalities, and thus easily extends to the rectangles $P_{\delta,\eps}^+$. Since the ratio of length and width is controlled, the  constant is independent of $\delta$ and $\eps$.) Then, from \eqref{ineq:compactnessprop}, \eqref{ineq:isoperimetric}, and $\partial^* G^\eps_0 \cap P_{\delta,\eps}^+ \subset \bigcup_{j=1}^\infty (\partial^* G^\eps_j \cap P_{\delta,\eps}^+)$  it follows  
\begin{align}\label{ineq:allphases}
\sum\nolimits_{j=0}^\infty \min\big\{ \mathcal{L}^2(G_j^\eps \cap P_{\delta,\eps}^+), \, \mathcal{L}^2(P_{\delta,\eps}^+\setminus G_j^\eps) \big\}  \leq C E_\varepsilon\big(X_\varepsilon,Q^\nu\setminus R^\nu_{1,\delta/2}\big).
\end{align}
We now get that there is a unique dominant component, i.e., there exists $j_\eps\in \N_0$ such that 
\begin{align}\label{eq: mass}
\mathcal{L}^2(G_{j_\eps}^\eps \cap P_{\delta,\eps}^+)  > \frac{1}{2}\mathcal{L}^2(P_{\delta,\eps}^+).
\end{align}
In fact, assume by contradiction that this were not the case. Then, we get for all $j \in \N_0$
$$ \min \big\{ \mathcal{L}^2(G_j^\eps \cap P_{\delta,\eps}^+), \mathcal{L}^2(  P_{\delta,\eps}^+\setminus G_j^\eps) \big\} = \mathcal{L}^2(G_j^\eps \cap P_{\delta,\eps}^+).$$
 By using (\ref{ineq:allphases}) we obtain
$
\mathcal{L}^2(P_{\delta,\eps}^+) = \sum_{j=0}^\infty \mathcal{L}^2(G_j^\eps \cap P_{\delta,\eps}^+)    \leq C E_\varepsilon\big(X_\varepsilon,Q^\nu\setminus R^\nu_{1,\delta/2}\big).
$
This contradicts \eqref{eq: liminfoutsidedelta0}  for $\varepsilon$ small enough. Now \eqref{ineq:allphases} and \eqref{eq: mass} imply \eqref{ineq : dominantphase} for the choice $z^+_\eps = z_{j_\eps}^\eps$.   

To conclude this step, we note that the convergence $ \lim\nolimits_{\eps \to 0} \int_{Q^\nu} |u_\eps (x + y_\eps)-   u^{\nu}_{z^+,z^-}(x)| \, {\rm d}x = 0$ along with \eqref{eq: mass}  also yields $z^+_\varepsilon \to z^+$.

\medskip The rest of the proof is divided into two cases:  (a) $z_\varepsilon^+ \neq \mathbf{0}$ and (b) $z_\varepsilon^+ = \mathbf{0}$, i.e., $X_\eps$ converges to a lattice in the upper half of the cube or there is vacuum. We perform the proof for case (a). At the end of the proof  (Step 7), we indicate the necessary changes to treat case (b). 

\medskip

\noindent \emph{Step 3: Cardinality estimate.}   We prove that there exists $C>0$ such that
\begin{align}\label{ineq:cardestimate}
\varepsilon^2 \#\left( \big( \eps\mathscr{L}(z^\pm_\eps) \triangle X_\varepsilon\big) \cap P_{\delta,\eps}^\pm \right) \leq CE_\varepsilon\big(X_\varepsilon, Q^\nu\setminus R^\nu_{1,\delta/2}\big),
\end{align}
where here and in the following $\triangle$ denotes the symmetric difference of sets. First, consider some $x \in (\eps\mathscr{L}(z^+_\eps) \setminus X_\varepsilon) \cap P_{\delta,\eps}^+$. Then, by the definition of $u_\varepsilon$  in \eqref{def:u} we get
\begin{align}\label{ineq:uepszplus}
u_\varepsilon(y) \neq z^+_\varepsilon \text{ for all } y \in B_{\eps/4}(x).
\end{align}
Indeed, otherwise we would find $y \in B_{\eps/4}(x)$ and   $x' \in X_\eps \cap B_{\eps/\sqrt{3}}(y)$ with $\# \mathcal{N}_\eps(x') = 6$ and $\lbrace x' \rbrace \cup \mathcal{N}_\eps(x') \subset \eps \mathscr{L}(z^+_\eps)$. The latter follows from the fact that $ V_\eps^{z^+_\eps}(x')  \subset B_{\eps/\sqrt{3}}(x') $. In particular, we have $x' \in \eps\mathscr{L}(z^+_\eps)$ and  $|x-x'| \le |x-y| + |y-x'| \le \eps/4 + \eps/\sqrt{3} < \eps$.  This, however, is impossible since $|x_1-x_2| \ge \eps $ for all $x_1,x_2\in \eps\mathscr{L}(z^+_\eps)$, $x_1 \ne x_2$.

On the other hand, if there exists $x \in (X_\varepsilon \setminus \eps\mathscr{L}(z^+_\eps)) \cap P_{\delta,\eps}^+ $, then we find $x_0 \in \eps\mathscr{L}(z^+_\eps)  \cap P_{\delta,\eps}^+ $ with $|x_0-x| < \eps$. Clearly, $x_0 \notin X_\eps$ by \eqref{def:potential} and the fact that $E_\eps(X_\eps)<+\infty$. Repeating the reasoning in \eqref{ineq:uepszplus} we find 
\begin{align}\label{ineq:uepszplus2}
u_\varepsilon(y) \neq z^+_\varepsilon \text{ for all } y \in B_{\eps/4}(x_0).
\end{align} 
Note that, in this procedure, $x_0$ can be chosen for at most six $x\in X_\varepsilon$ independently of $\varepsilon$  since $\#(X_\eps \cap B_\eps(x_0))\le 6$  due to $E_\eps(X_\eps)<+\infty$. Using \eqref{ineq : dominantphase},  $\mathcal{L}^2(B_{\eps/4}(x) \cap P_{\delta,\eps}^+ ) \ge c\eps^2 $ for all $x \in  \eps\mathscr{L}(z^+_\eps) \cap P_{\delta,\eps}^+$, and \eqref{ineq:uepszplus}--\eqref{ineq:uepszplus2} we conclude
\begin{align*}
\varepsilon^2 \#\left( \big( \eps\mathscr{L}(z^+_\eps) \triangle X_\varepsilon\big) \cap P_{\delta,\eps}^+ \right) \leq C \mathcal{L}^2 \big( \{u_\varepsilon \neq z^+_\varepsilon\} \cap P_{\delta,\eps}^+\big) \leq  CE_\varepsilon\big(X_\varepsilon, Q^\nu\setminus R^\nu_{1,\delta/2}\big).
\end{align*}

\noindent \emph{Step 4: Cut-off construction.}  In this step, we construct a new configuration $Y_\eps^+ \subset \R^2$ such that $ Y_\eps^+ = \varepsilon\mathscr{L}(z_\varepsilon^+) \text{ on } \partial^+_\varepsilon Q^\nu$, see \eqref{eq: eps-rand}. This construction changes the configuration in the upper half-cube $Q^{\nu,+}$. Step 5 then shows that the energy of $Y^+_\varepsilon$ is asymptotically equal to the  one  of $X_\varepsilon$. The procedure can then be repeated on the lower half-cube. We defer this to Step 6 below.

Set $N_\varepsilon= \left\lfloor \frac{\delta}{6\varepsilon}\right\rfloor$. (Here and in the sequel, we do not highlight the dependence on $\delta$ to save notation.) For $k\in\lbrace 0,\ldots,N_\varepsilon+1\rbrace$ we let $r_k = 1-\delta +3k\varepsilon$ and define the layers
\begin{align}\label{eq: Sdef}
 S_k^\varepsilon = \big(Q^{\nu,+}_{r_k} \setminus Q^{\nu,+}_{r_{k-1}} \big) \setminus R^\nu_{1,\delta}.
\end{align}
For $k\in\lbrace 1,\ldots,N_\varepsilon\rbrace$ we also define the ``thickened layers'' $L_{k}^\varepsilon =S_{k-1}^\varepsilon \cup S_{k}^\varepsilon\cup S_{k+1}^\varepsilon$. Our goal is to perform a transition to the lattice   $\varepsilon\mathscr{L}(z^+_\eps)$ on one of these layers. To this end, we choose a convenient layer by an averaging argument: by \eqref{ineq:cardestimate} there exists $k_\varepsilon \in \{1,\ldots,N_\varepsilon\}$ such that
\begin{align}\label{ineq:averaging}
\#\left( ( \eps\mathscr{L}(z_\eps^+) \triangle X_\varepsilon) \cap L_{k_\varepsilon}^\varepsilon \right) &\leq \frac{1}{N_\varepsilon}\sum\nolimits_{k=1}^{N_\varepsilon}\#\left( ( \eps\mathscr{L}(z_\eps^+) \triangle X_\varepsilon) \cap L_{k}^\varepsilon \right)\notag\\&\leq  \frac{3}{N_\eps}\,\#\left(( \eps\mathscr{L}(z_\eps^+) \triangle X_\varepsilon) \cap P_{\delta,\eps}^+ \right) \leq \frac{C}{\varepsilon \delta}  E_\varepsilon\big(X_\varepsilon,Q^\nu\setminus R^\nu_{1,\delta/2}\big).
\end{align}
Here, we used $L_{k}^\varepsilon \subset  P_{\delta,\eps}^+$ for all $k$ and $\eps\delta \le CN_\eps \eps^2$. The factor $3$ is due to the fact that we count each strip  $S_k^\varepsilon$ at most three times. Set  $D^\varepsilon :=Q^\nu_{r_{k_\varepsilon-1}} \cup ( Q^{\nu,-} \setminus R^\nu_{1,\delta})$.  We now define $Y_\eps^+$ by
\begin{align}\label{def:Yeps}
Y_\eps^+ 
= \begin{cases} \eps\mathscr{L}(z^+_\eps)  &\text{in } (P_{\delta,\eps}^+ \setminus Q^\nu_{r_{k_\varepsilon}}) \cup \partial^+_\varepsilon Q^\nu,\\
\emptyset & \text{in } (R^\nu_{1,\delta} \setminus Q^\nu_{r_{k_\varepsilon-1}})\setminus (\partial^+_\varepsilon Q^\nu \cup \partial^-_\varepsilon Q^\nu), \\
X_\varepsilon \cap \eps\mathscr{L}(z^+_\eps) &\text{in }S_{k_\varepsilon}^\varepsilon,\\
X_\varepsilon &\text{in } D^\varepsilon \cup \partial^-_\varepsilon Q^\nu.
\end{cases}
\end{align}
See Figure \ref{Fig:DefinitonYeps} for an illustration of the different regions. We briefly explain the definition. In $D^\eps \cup \partial^-_\varepsilon Q^\nu$, the configuration remains unchanged,  and near the boundary of the upper half-cube it coincides with the lattice $\eps\mathscr{L}(z^+_\eps)$. In $S_{k_\eps}^\eps$,  we use the intersection $X_\varepsilon \cap \eps\mathscr{L}(z^+_\eps)$. In this sense, $S_{k_\eps}^\eps$ can be understood as a transition layer. Eventually, small regions near the boundary close to the interface $\partial Q^{\nu,+} \cap \partial Q^{\nu,-}$ do not contain atoms.  This is convenient since in this region the energy of the original configuration possibly does not vanish.   Note that the latter ensures that  $|y_1-y_2| \ge \eps$ for all  $y_1,y_2 \in Y_\eps^+, y_1 \neq y_2,$ and therefore
\begin{align}\label{eq: least distance}
 E_\eps(Y_\eps^+) < + \infty. 
\end{align} 
 Finally, we point out that  $Y_\eps^+ \not\subset Q^\nu$ due to the definition of $\partial_\eps^\pm Q^\nu$ in \eqref{eq: eps-rand},  see also Figure \ref{fig:cellformula}. 
\begin{figure}[htp]
 \includegraphics{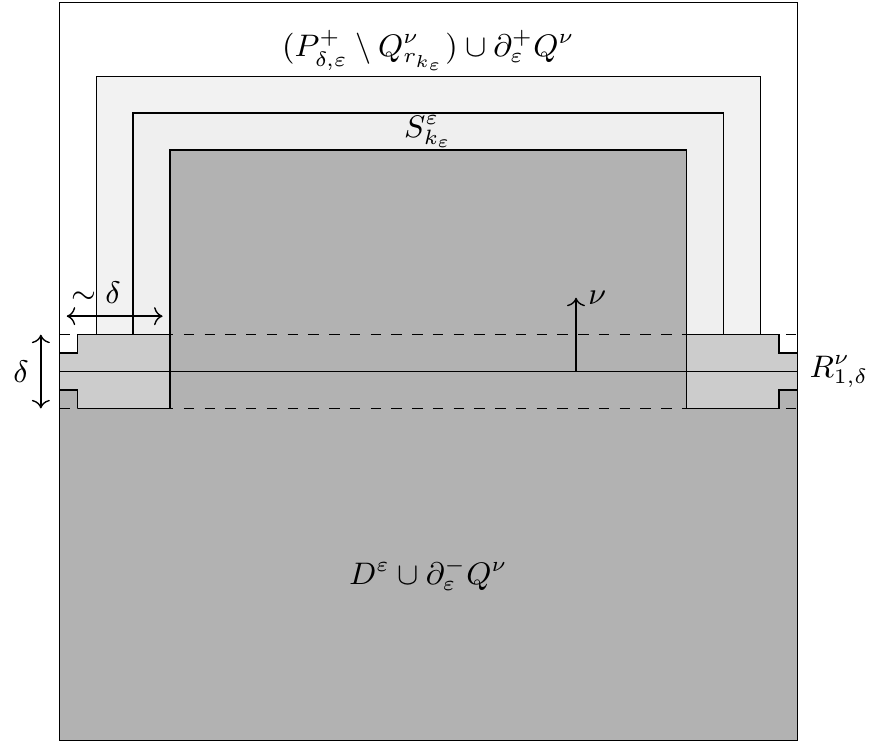}
\caption{The different regions for $Y_\eps^+$  inside $Q^\nu$:  dark gray region $D^\eps \cup \partial^-_\eps Q^\nu$, gray region $(R^\nu_{1,\delta} \setminus Q^\nu_{r_{k_\varepsilon-1}})\setminus  (\partial^+_\varepsilon Q^\nu \cup \partial^-_\varepsilon Q^\nu)$, light gray region $S_{k_\varepsilon}^\varepsilon$, and white region $ (P_{\delta,\eps}^+ \setminus Q^\nu_{r_{k_\varepsilon}}) \cup \partial^+_\varepsilon Q^\nu$. The two dashed lines enclose the region $R^\nu_{1,\delta}$.}

\label{Fig:DefinitonYeps}
\end{figure}

\noindent \emph{Step 5: Energy estimate.}  In this step we show that the energy of the configuration constructed in Step 4  is asymptotically controlled by the original energy, i.e., 
\begin{align}\label{ineq: Energyinequality}
\liminf_{\varepsilon \to 0} E_\varepsilon(Y_\eps^+,Q^\nu) \leq \liminf_{\varepsilon \to 0} E_\varepsilon(X_\varepsilon,Q^\nu) + C\delta
\end{align}
for some universal $C>0$. In order to obtain (\ref{ineq: Energyinequality}), we distinguish  three regions:
\begin{align}\label{eq: Asets}
 A^\eps_1= \overline{(R^\nu_{1,\delta} \setminus Q^\nu_{r_{k_\varepsilon-1}})_\varepsilon}, \ \ \ \ \ \     A^\eps_2= \overline{(S_{k_\varepsilon}^\varepsilon)_\varepsilon}\setminus A^\eps_1, \ \ \ \ \ \    A^\eps_3 = Q^\nu \setminus  (A^\eps_1 \cup A^\eps_2).
\end{align}
\noindent\emph{Energy estimate on $A^\eps_1$:}
We claim that there exists a universal $C>0$ such that
\begin{align}\label{ineq:energyA1}
E_\varepsilon(Y_\eps^+,A^\eps_1) \leq C\delta.
\end{align}
In fact, due to (\ref{def:Yeps}), we have $Y_\eps^+ \cap (R^\nu_{1,\delta} \setminus Q^\nu_{r_{k_\varepsilon-1}})= ( \varepsilon\mathscr{L}(z^+_\eps) \cap  R^\nu_{1,\delta}  \cap \partial^+_\varepsilon Q^\nu )   \cup ( X_\varepsilon \cap R^\nu_{1,\delta} \cap \partial^-_\varepsilon Q^\nu )$. As  $\mathcal{L}^2((R^\nu_{1,\delta}  \cap  \partial^\pm_\varepsilon Q^\nu)_\varepsilon) \leq C \delta \varepsilon$,  see \eqref{eq: eps-rand} and \eqref{eq: reci}, by Lemma \ref{lemma:propertiesofE}(v)   we get  
\begin{align}\label{eq: first}
\#\big(Y_\eps^+ \cap (R^\nu_{1,\delta} \setminus Q^\nu_{r_{k_\varepsilon-1}})\big) \leq C\delta/\eps.
\end{align}
 Here, Lemma \ref{lemma:propertiesofE} is applicable by  \eqref{eq: least distance}.   Additionally, we note that  $R^\nu_{1,\delta} \setminus Q^\nu_{r_{k_\varepsilon-1}}$ consists of two rectangles and we have $\mathcal{H}^1(\partial (R^\nu_{1,\delta} \setminus Q^\nu_{r_{k_\varepsilon-1}})) \leq C\delta$. Hence, by Lemma \ref{lemma:propertiesofE}(v) we obtain
\begin{align*}
\# \big((A^\eps_1\cap Y_\eps^+) \setminus (  R^\nu_{1,\delta}   \setminus Q^\nu_{r_{k_\varepsilon-1}}) \big) &\leq C\varepsilon^{-2} \mathcal{L}^2\Big(\big(\overline{(R^\nu_{1,\delta} \setminus Q^\nu_{r_{k_\varepsilon-1}})_\varepsilon}\setminus (  R^\nu_{1,\delta}  \setminus Q^\nu_{r_{k_\varepsilon-1}})\big)_\varepsilon\Big)\\&\leq C\eps^{-1} \mathcal{H}^1\big(\partial (  R^\nu_{1,\delta}  \setminus Q^\nu_{r_{k_\varepsilon-1}})\big) \leq C\delta/\eps.
\end{align*}
This along with \eqref{eq: first} yields $\#(A^\eps_1\cap Y_\eps^+) \leq C\delta/\eps$, and therefore (\ref{ineq:energyA1}) follows by \eqref{def:energyneighbourhood}.

\noindent\emph{Energy estimate on $A^\eps_2$:} We prove that there exists a universal $C>0$ such that
\begin{align}\label{ineq:energyA2}
 E_\varepsilon(Y_\eps^+,A^\eps_2) \leq (1+C /\delta) \, E_\varepsilon\big(X_\varepsilon,Q^\nu\setminus R^\nu_{1,\delta/2}\big). 
\end{align}
First, the definition of $ L_{k_\varepsilon}^\varepsilon$ below  \eqref{eq: Sdef} implies $(A^\eps_2)_\eps \subset L_{k_\varepsilon}^\varepsilon$. For $x\in Y_\eps^+$, we denote the neighborhood of $x$ with respect to $Y_\eps^+$ by $\mathcal{N}_{\varepsilon,Y}(x)$, cf.\ \eqref{def:neighbourhood}.
We claim that
\begin{align}\label{ineq:cardneighbourhoodA2}
\#\mathcal{N}_{\varepsilon,Y}(x) \geq \#\mathcal{N}_\varepsilon(x) -6\, \#\big(\overline{B_\varepsilon(x)} \cap (X_\varepsilon \setminus\varepsilon \mathscr{L}(z^+_\eps))\big) \  \ \ \text{ for all } x\in X_\varepsilon \cap Y_\eps^+ \cap A^\eps_2.
\end{align}
In fact, if $\overline{B_\varepsilon(x)} \cap (X_\varepsilon \setminus\varepsilon \mathscr{L}(z^+_\eps))\neq \emptyset$, the right hand side is nonpositive since $\#\mathcal{N}_\varepsilon(x)\le 6$, see \eqref{eq: neighborhood bound}. Since $\#\mathcal{N}_{\varepsilon,Y}(x) \geq 0$,  (\ref{ineq:cardneighbourhoodA2}) follows in this case. On the other hand, if $\overline{B_\varepsilon(x)} \cap (X_\varepsilon \setminus\varepsilon \mathscr{L}(z^+_\eps))= \emptyset$, by (\ref{def:Yeps}), we may have only increased the cardinality of the neighborhood by adding atoms in $\eps\mathscr{L}(z_\eps^+) \setminus X_\eps$, i.e., $\#\mathcal{N}_{\varepsilon,Y}(x) \geq \#\mathcal{N}_\varepsilon(x)$. This    again yields \eqref{ineq:cardneighbourhoodA2}. 

We split the sum into $X_\varepsilon \cap Y_\eps^+$ and $Y_\eps^+ \setminus X_\varepsilon$. By using  \eqref{def:energyneighbourhood}, $A^\eps_2 \subset L_{k_\varepsilon}^\varepsilon$, Lemma \ref{lemma:propertiesofE}(iii),  and (\ref{ineq:cardneighbourhoodA2}) we obtain
\begin{align}\label{ineq:A2intermediate}
E_\varepsilon(Y_\eps^+,A^\eps_2) &\leq C\varepsilon\#\, \big\{x \in A^\eps_2 \cap (Y_\eps^+ \setminus X_\varepsilon)\big\}  +  \frac{1}{2}  \underset{x \in A^\eps_2}{\sum_{x \in Y_\eps^+ \cap X_\varepsilon}}\varepsilon\big(6-\#\mathcal{N}_{\varepsilon,Y}(x)\big) \\& 
\nonumber\leq C\varepsilon\#\big\{x \in (Y_\eps^+ \cap L_{k_\varepsilon}^\varepsilon) \setminus X_\varepsilon\big\} +  3 \varepsilon \hspace{-0.35cm} \underset{x \in A^\eps_2}{\sum_{x \in Y_\eps^+ \cap X_\varepsilon}}
 \# \big(\overline{B_\varepsilon(x)} \cap (X_\varepsilon \setminus\varepsilon\mathscr{L}(z_\eps^+))\big)+ E_\varepsilon(X_\varepsilon, L_{k_\varepsilon}^\varepsilon).
\end{align}
Note by (\ref{def:Yeps}) that $Y_\eps^+ \subset \varepsilon\mathscr{L}(z^+_\eps) \cup X_\varepsilon$ in $L_{k_\varepsilon}^\varepsilon$. Therefore, in view of      (\ref{ineq:averaging}), we obtain
\begin{align}\label{ineq:A2intermediate2}
\#\big\{x \in (Y_\eps^+ \cap L_{k_\varepsilon}^\varepsilon) \setminus X_\varepsilon\big\} \leq \#\big\{x \in (\varepsilon\mathscr{L}(z^+_\eps) \triangle X_\varepsilon) \cap L_{k_\varepsilon}^\varepsilon \big\} \leq    \frac{C}{\varepsilon \delta}E_\varepsilon\big(X_\varepsilon,Q^\nu\setminus R^\nu_{1,\delta/2}\big).
\end{align}
Exploiting (\ref{ineq:averaging}) once more, we get 
\begin{align}\label{ineq:A2intermediate3}
\begin{split}
\sum\nolimits_{x \in Y_\eps^+ \cap X_\varepsilon \cap A^\eps_2}
 \#\big (\overline{B_\varepsilon(x)} \cap (X_\varepsilon \setminus\varepsilon\mathscr{L}(z^+_\eps))\big) &\leq C \#\big\{x \in (\varepsilon\mathscr{L}(z^+_\eps) \triangle X_\varepsilon) \cap L_{k_\varepsilon}^\varepsilon \big\}\\&\leq    \frac{C}{\varepsilon \delta}E_\varepsilon\big(X_\varepsilon,Q^\nu\setminus R^\nu_{1,\delta/2}\big).
 \end{split}
\end{align}
Here, the first inequality holds because $|x_1-x_2| \geq \varepsilon$ for $x_1,x_2 \in X_\varepsilon$, $x_1\neq x_2$, and  $\overline{B_\varepsilon(x)} \subset L_{k_\varepsilon}^\varepsilon$ for all $x \in A^\eps_2$. Hence, we get that every point in $ (X_\varepsilon \setminus\varepsilon\mathscr{L}(z^+_\eps)) \cap L_{k_\varepsilon}^\varepsilon$ is only accounted for at most seven times in the sum.   Now, using (\ref{ineq:A2intermediate})--(\ref{ineq:A2intermediate3}), $L_{k_\varepsilon}^\varepsilon \subset Q^\nu\setminus R^\nu_{1,\delta}$, and Lemma \ref{lemma:propertiesofE}(iii), we obtain (\ref{ineq:energyA2}).

\noindent \emph{Energy estimate on $A^\eps_3$:} We claim that 
\begin{align}\label{ineq:energyA3}
E_\varepsilon(Y_\eps^+,A^\eps_3) \leq E_\varepsilon(X_\varepsilon,Q^\nu).
\end{align}
Recalling \eqref{eq: Asets}  we get that each $x \in A^\eps_3\cap  Y_\eps^+$ lies either in  $T^\eps := (P_{\delta,\eps}^+ \setminus Q^\nu_{r_{k_\varepsilon}}) \cup (\partial^+_\eps Q^\nu \setminus R^\nu_{1,\delta})$ or in $D^\varepsilon$. If $x \in A^\eps_3\cap  Y_\eps^+ \cap T^\eps$, then also  $\overline{B_\eps(x)} \subset T^\eps$. (Here, we use the definition of $A^\eps_1$, $A^\eps_2$ and \eqref{eq: eps-rand}.) Then, \eqref{def:Yeps} implies  $\#\mathcal{N}_{\varepsilon,Y}(x) =6$. On the other hand, if $x \in A^\eps_3\cap  Y_\eps^+ \cap D^\eps$, then $X_\varepsilon \cap \overline{B_\eps(x)} = Y^+_\varepsilon \cap \overline{B_\eps(x)}$, which yields $\mathcal{N}_{\varepsilon,Y}(x) = \mathcal{N}_\varepsilon(x)$. 
Thus, by \eqref{def:energyneighbourhood} and Lemma \ref{lemma:propertiesofE}(iii),(iv) we obtain (\ref{ineq:energyA3}). In fact, we get
\begin{align*}
E_\varepsilon(Y_\eps^+,A^\eps_3) = E_\varepsilon\big(Y_\eps^+,A^\eps_3 \cap T^\eps \big) + E_\varepsilon\big(Y_\eps^+,A^\eps_3 \cap D^\eps \big) =E_\varepsilon\big(Y_\eps^+,A^\eps_3 \cap D^\eps \big)  \leq E_\varepsilon(X_\varepsilon,Q^\nu).
\end{align*}

To conclude this step of the proof, it suffices to recall that by  Lemma \ref{lemma:propertiesofE}(iv)
\begin{align*}
E_\varepsilon(Y_\eps^+,Q^\nu)=E_\varepsilon(Y_\eps^+,A^\eps_1)+E_\varepsilon(Y_\eps^+,A^\eps_2)+E_\varepsilon(Y_\eps^+,A^\eps_3).
\end{align*}
Then we obtain \eqref{ineq: Energyinequality} by  \eqref{eq: liminfoutsidedelta0},   (\ref{ineq:energyA1}), (\ref{ineq:energyA2}), and (\ref{ineq:energyA3}).

\noindent \emph{Step 6: Conclusion.} By repeating the cut-off construction in Step   4 on $Q^{\nu,-}$ for $z_\varepsilon^-$, we obtain a configuration $Y_\varepsilon$ such that $Y_\varepsilon = \eps \mathscr{L}(z_\varepsilon^\pm)$ on $\partial^\pm_\varepsilon Q^\nu(y_\eps)$ and 
\begin{align}\label{eq: now forY}
\liminf_{\varepsilon \to 0} E_\varepsilon\big(Y_\eps,Q^\nu(y_\eps)\big) \leq \liminf_{\varepsilon \to 0} E_\varepsilon\big(X_\varepsilon,Q^\nu(y_\eps)\big) + C\delta
\end{align}
 by \eqref{ineq: Energyinequality}, where we reinclude the center $y_\eps$ in the notation for clarification. Since  $z^\pm_\varepsilon \to z^\pm$ by Step 2, we observe by the definition of $\Phi$ in \eqref{eq: Phi def} that
\begin{align*}
\liminf_{\varepsilon \to 0} E_\varepsilon(Y_\varepsilon,Q^\nu(y_\eps))\geq \Phi(z^+,z^-,\nu).
\end{align*}
 By using \eqref{eq: optimal!}, \eqref{eq: now forY} and by passing to $\delta \to 0$, we obtain the statement of the lemma.

\noindent \emph{Step 7: Adaptions in  $\mathrm{(b)}$.} To conclude the proof of the lemma, it remains to describe Steps 3--5 in the case of vacuum, i.e., $z^+_\eps = \mathbf{0}$.

 \noindent \emph{Step 3 for case $\mathrm{(b)}$: Cardinality estimate.} We prove that
\begin{align}\label{ineq:cardestimatecase(b)}
\varepsilon^2\#(X_\varepsilon \cap P_{\delta,\eps}^+) \leq CE_\varepsilon(X_\varepsilon,Q^\nu\setminus   R^\nu_{1,  \delta/2 }) 
\end{align}
for a universal $C>0$. In fact, if $x \in X_\varepsilon$ has $\#\mathcal{N}_\varepsilon(x)=6$, then $u_\eps(x) \neq \mathbf{0}$ on $B_{\eps/2}(x)$ by \eqref{def:u} and the fact that $B_{\eps/2}(x) \subset V_\eps^{z(x)}(x)$. Also note the $B_{\eps/2}(x) \cap B_{\eps/2}(y) =\emptyset$ for $x,y\in X_\eps$, $x\neq y$.   Thus, by \eqref{def:energyneighbourhood}, \eqref{ineq : dominantphase} (with $z^+_\varepsilon = \mathbf{0}$),  and Lemma \ref{lemma:propertiesofE}(iii) we get 
\begin{align*}
\varepsilon^2\#(X_\varepsilon \cap P_{\delta,\eps}^+) &\leq \varepsilon^2\#\{x \in X_\varepsilon \cap P_{\delta,\eps}^+ : \#\mathcal{N}_\varepsilon(x)=6\} + \varepsilon^2\sum\nolimits_{x \in X_\varepsilon  \cap P_{\delta,\eps}^+}(6-\#\mathcal{N}_\varepsilon(x)) \\&\leq C\mathcal{L}^2\big(\{u_\varepsilon \neq \mathbf{0}\}\cap P_{\delta,\eps}^+\big)+  2 \varepsilon E_\varepsilon\big(X_\varepsilon, Q^\nu\setminus  R^\nu_{1,\delta/2}  \big) \leq CE_\varepsilon\big(X_\varepsilon,Q^\nu\setminus  R^\nu_{1,\delta/2}  \big),
\end{align*}
where we again used that  $P_{\delta,\eps}^+ \subset Q^\nu\setminus R^\nu_{1,\delta/2}$. This concludes Step 3 in case (b).

\noindent \emph{Step 4 for case $\mathrm{(b)}$: Cut-off construction.} We now explain the  construction of a new configuration $Y^+_\varepsilon$ such that $ Y_\eps^+ = \mathbf{0} \text{ on } \partial^+_\varepsilon Q^\nu$. Again set $N_\varepsilon= \left\lfloor \frac{\delta}{6\varepsilon}\right\rfloor$ and define $S_k^\varepsilon$ as in \eqref{eq: Sdef}, as well as   $L_{k}^\varepsilon =S_{k-1}^\varepsilon \cup S_{k}^\varepsilon\cup S_{k+1}^\varepsilon$. Similar to \eqref{ineq:averaging}, by averaging over $k$ and using \eqref{ineq:cardestimatecase(b)}, there exists $k_\varepsilon \in \{1,\ldots,N_\varepsilon\}$ such that
\begin{align}\label{ineq:averaging2}
\begin{split}
\#( X_\varepsilon \cap L_{k_\varepsilon}^\varepsilon ) &\leq \frac{1}{N_\varepsilon}\sum\nolimits_{k=1}^{N_\varepsilon}\#( X_\varepsilon \cap L_{k}^\varepsilon )  \leq \frac{3}{N_\eps}\#(X_\varepsilon \cap P_{\delta,\eps}^+) \leq \frac{C}{\varepsilon \delta} E_\varepsilon\big(X_\varepsilon,Q^\nu\setminus R^\nu_{1,\delta/2}\big).
\end{split}
\end{align}
where we again use that each strip   $S_k^\varepsilon$ is counted at most three times.  We define 
 \begin{align}\label{def:Yeps000}
 Y_\eps^+=\begin{cases} \emptyset & \text{in }   \big( (P_{\delta,\eps}^+ \cup R^\nu_{1,\delta})  \setminus (Q_{r_{k_\varepsilon}}^\nu \cup \partial_\varepsilon^- Q^\nu ) \big) \cup \partial_\varepsilon^+ Q^\nu,  \\
 X_\varepsilon & \text{otherwise.}
 \end{cases} 
 \end{align}
 Note that, since $E_\varepsilon( X_\varepsilon) <+\infty$, we have that $E_\varepsilon(Y_\varepsilon^+) <+\infty$.

\noindent \emph{Step 5 for case $\mathrm{(b)}$: Energy estimate.} We again split the estimate into the three sets $A^\eps_1$, $A^\eps_2$, and $A^\eps_3$ defined in \eqref{eq: Asets}. 

\noindent\emph{Energy estimate for $A^\eps_1$:}
We claim that there exists $C>0$ such that
\begin{align}\label{ineq:energyA1b}
E_\varepsilon(Y_\eps^+,A^\eps_1) \leq C\delta.
\end{align}
In fact, due to (\ref{def:Yeps000}), we have  $Y_\eps^+ \cap (R^\nu_{1,\delta} \setminus Q^\nu_{r_{k_\varepsilon}})=X_\varepsilon \cap R^\nu_{1,\delta} \cap \partial_\varepsilon^- Q^\nu$, where, similarly as in \eqref{eq: first}, $\#(X_\varepsilon \cap R^\nu_{1,\delta} \cap \partial_\varepsilon^- Q^\nu) \le C \delta/\eps$. As $R^\nu_{1,\delta} \setminus Q^\nu_{r_{k_\varepsilon-1}}$ consists of two rectangles with $\mathcal{H}^1(\partial (R^\nu_{1,\delta} \setminus Q^\nu_{r_{k_\varepsilon-1}})) \leq C\delta$ and $Y_\eps^+$ satisfies   $E_\varepsilon(Y_\varepsilon^+) <+\infty$,    we obtain by   Lemma \ref{lemma:propertiesofE}(v) 
\begin{align*}
\#(A^\eps_1 \cap Y_\eps^+)
&= \#\big( \big(A^\eps_1 \setminus (  R^\nu_{1,\delta} \setminus Q^\nu_{r_{k_\varepsilon}})\big)  \cap Y^+_\varepsilon\big)  + \# \big( X_\varepsilon \cap R^\nu_{1,\delta} \cap \partial_\varepsilon^- Q^\nu \big) \\
&\leq C\varepsilon^{-2} \mathcal{L}^2\big(\big(A^\eps_1 \setminus (R^\nu_{1,\delta} \setminus Q^\nu_{r_{k_\varepsilon}})\big)_\varepsilon \big) + C\delta/\eps  \\ 
&\leq  C\eps^{-1}\mathcal{H}^1\big(\partial (R^\nu_{1,\delta} \setminus Q^\nu_{r_{k_\varepsilon-1}})\big)  + C\delta/\eps  \leq C\delta/\eps.
\end{align*}
Then (\ref{ineq:energyA1b}) follows by \eqref{def:energyneighbourhood}.

\noindent\emph{Energy estimate for $A^\eps_2$:} We claim that there exists $C>0$ such that
\begin{align}\label{ineq:energyA2b}
E_\varepsilon(Y_\eps^+,A^\eps_2) \leq \frac{C}{\delta}E_\varepsilon\big(X_\varepsilon,Q^\nu\setminus R^\nu_{1,\delta/2}\big). 
\end{align}
In fact, if $x \in Y_\eps^+\cap A^\eps_2$, then $x \in X_\varepsilon \cap L_{k_\varepsilon}^\varepsilon$. Using  \eqref{def:energyneighbourhood} and (\ref{ineq:averaging2}) we obtain (\ref{ineq:energyA2b}).

\noindent\emph{Energy estimate for $A^\eps_3$:} We observe that
\begin{align}\label{ineq:energyA3b}
E_\varepsilon(Y_\eps^+,A^\eps_3) \leq E_\varepsilon(X_\varepsilon,Q^\nu).
\end{align}
Indeed, if $x \in Y_\eps^+ \cap (Q^\nu \setminus (A^\eps_1 \cup A^\eps_2))$, then $\mathcal{N}_{\varepsilon,Y}(x) = \mathcal{N}_\varepsilon(x)$, where the neighborhood of $x$ with respect to $Y_\eps^+$ is again denoted by $\mathcal{N}_{\varepsilon,Y}(x)$. Therefore, \eqref{ineq:energyA3b} follows by \eqref{def:energyneighbourhood} and Lemma \ref{lemma:propertiesofE}(iii). 

Summarizing, \eqref{ineq:energyA1b}--\eqref{ineq:energyA3b} and \eqref{eq: liminfoutsidedelta0} yield 
\begin{align*}
\liminf_{\varepsilon \to 0} E_\varepsilon(Y_\eps^+,Q^\nu) \leq \liminf_{\varepsilon \to 0} E_\varepsilon(X_\varepsilon,Q^\nu) + C\delta,
\end{align*}
which is the analog to \eqref{ineq: Energyinequality}. The rest of the proof (i.e., Step 6) remains unchanged.
\end{proof}

\section{Reduction of the problem to subsets of two lattices}\label{sec:Reduction-two-lattices}

In the previous section, we have seen that the condition of $L^1$-convergence in the definition of $\psi$ (see \eqref{def:psi}) can be replaced by converging boundary values, see the definition of $\Phi$ in \eqref{eq: Phi def}. From now on, it will be convenient to express the problem with lattice spacing equal to $1$.  Recall \eqref{eq: eps-rand} and observe that by Lemma \ref{lemma:propertiesofE}  the cell formula for  $\Phi$ can be written as 
\begin{align}\label{lemma:Phi}
  \Phi(z^+,z^-,\nu) 
  = \min \Big\{ \liminf_{T\to +\infty} \frac{1}{T}\inf\Big\{E_1\big(X_T,Q^\nu_T(y_T)\big)\colon \,  y_T \in \R^2,  & \  X_T  = \mathscr{L}(z^\pm_T) \text{ on } \partial^\pm_1 Q_T^\nu(y_T) \Big\} \colon \notag \\ &   
    \lbrace z^\pm_T\rbrace_T \subset \mathcal{Z} \text{ with }  z^\pm_T \to z^\pm\Big\} 
\end{align} 
for all $z^+,z^- \in \mathcal{Z}$ and $\nu \in \mathbb{S}^1$. This section is devoted to a fundamental ingredient for  the proof of relation of $\Phi$ and $\varphi$, and the properties of $\varphi$, which will be addressed in Sections \ref{sec: Part II} and \ref{section:surfacetension2}. We show that the minimization problem in \eqref{lemma:Phi} can be reduced to configurations that are  subsets of \emph{two lattices only} (or just one if either  $z^+=\mathbf{0}$ or $z^-=\mathbf{0}$). For the formulation of the lemma, we introduce two further notions: we say that a set $Y \subset \R^2$ is \emph{connected} if for each pair $x,y \in Y$ there exists a chain $(v_1,\ldots,v_n)$ with $v_i \in Y$ for $i \in \lbrace 1,\ldots,n\rbrace$, $v_1= x$, $v_n= y$, and $|v_{i+1} - v_i|  = 1$ for $i \in \lbrace 1,\ldots,n-1\rbrace$. Moreover, given a configuration $X$ and $Y \subset X$, we define \emph{the boundary of $Y$ inside $Q^\nu_T(y)$}  by
\begin{align}\label{eq: boundary of Y}
\partial Y = \{ x \in Y \cap Q^\nu_T(y) \colon \#(\mathcal{N}(x) \cap Y) < 6\}. 
\end{align}

\begin{lemma}[Reduction to subsets of two lattices]\label{lemma:reduction} Let $z^+,z^- \in \mathcal{Z}$, $\nu \in \mathbb{S}^1$, $y \in \R^2$, and $T>0$. Let $X \subset \mathbb{R}^2$ be a minimizer of
\begin{align}\label{eq: one inequl}
\min\Big\{E_1\big(X,Q^\nu_T(y)\big)\colon \   X = \mathscr{L}(z^\pm) \text{ \rm on } \partial^\pm_1 Q_T^\nu(y) \Big\}. 
\end{align}
Then, it satisfies the following two properties:
\begin{itemize}
\item[(i)] (Subset of lattices) There holds $X = X^+ \cup X^-$  on $Q_T^\nu(y)$,  where $X^\pm \subset \mathscr{L}(z^\pm)$ and $X^\pm$ is connected. 
\item[(ii)] (Structure of boundaries) The sets $\partial X^+$ and $\partial X^-$ defined in \eqref{eq: boundary of Y}  are connected and satisfy $\# \mathcal{N}(x) \le 5$ for all $x \in \partial X^\pm$, as well as $\max_{x,y \in \partial X^\pm} |x-y| \ge T$. 
\end{itemize}
\end{lemma}

Note that the minimum in \eqref{eq: one inequl} exists since $E_1$ is lower semicontinuous, see \eqref{def:potential} and \eqref{def:energy}, and the problem is finite dimensional. We also point out that $X^+\cap X^- \neq \emptyset$ is possible, see e.g.\ Figure~\ref{fig:low energy},  i.e., the two grains described by $X^+$ and $X^-$ can have common atoms. Resolving this ambiguity by introducing a specific choice, the grain boundary and bonds connecting the two grains can be described in more detail. 

\begin{lemma}[Bonds between grain boundaries]\label{lemma:grain-bonds}
Let $X^\pm$ be the sets found in Lemma \ref{lemma:reduction}. There exist $Y^\pm$ with $X^\pm \setminus \partial X^\mp \subset Y^\pm \subset X^\pm$ such that: 
\begin{itemize} 
\item[(i)] (Partition into grains) $Y^+ \cup Y^- = X^+ \cup X^-$ and $Y^+ \cap Y^- \cap Q_T^\nu(y) = \emptyset$.  
\item[(ii)] (Grain and bulk boundaries) $\partial Y^\pm \subset \partial X^\pm$ and $Y^\pm = \mathscr{L}(z^\pm)$ on $\partial^\pm_1 Q_T^\nu(y)$. 
\item[(iii)] (Neighborhood structure at grain boundary) There holds 
\begin{align*}
  \big| \sum\nolimits_{x \in \partial Y^\pm} \# (\mathcal{N}(x) \cap Y^\pm) - 4 \# \partial Y^\pm \big| \le 2. 
\end{align*}
\end{itemize}
\end{lemma}
We thus have that on average each boundary atom has four neighbors in the same grain. As it has at most five neighbors in the whole configuration, it has on average less than one bond connecting it to the other grain.

From a technical perspective,  Lemma \ref{lemma:reduction}  will provide an important tool to study the properties of the cell formulas. From the physical point of view, it shows that our extremely brittle set-up, while allowing for rebonding, does not support interpolating boundary layers near cracks.  Its proof will require some concepts from graph theory which will be only needed for this part of the article. For this reason,  it is possible   to omit the proofs of Lemmas \ref{lemma:reduction} and \ref{lemma:grain-bonds}  on first reading and to proceed directly with Section \ref{sec: Part II}.  As our graph theoretic description gives in fact a more precise picture of the geometry of grain boundaries, which is of some independent interest, we summarize these findings in Theorem \ref{theorem:grain-boundary} at the end of Section \ref{sec:Reduction-two-lattices}. 

We now address the proof of the lemma and start by introducing some notions from graph theory.

\noindent\emph{The bond graph:} 
We define the \textit{bond graph} of $X\subset \R^2$ as the set of positions $X$ with the set of \emph{bonds} $\{\{x,y\}\colon  \, x \in X, \  y \in \mathcal{N}(x)\}$, where $\mathcal{N}(x) = \mathcal{N}_1(x)$ is defined in \eqref{def:neighbourhood}.  As for  configurations with finite energy $E_1$  there holds $\mathrm{dist}(x,X \setminus \{x\})\geq 1$ for all $x\in X$ and $y \in \mathcal{N}(x)$ only if $|x-y| = 1 < \sqrt{2}$,     the bond   graph is planar.  Indeed, given a quadrilateral  with all sides and one diagonal equal to $1$, the second diagonal is $\sqrt{3} >1$.

A sequence of atoms $p=(v_1,\ldots,v_n) \subset X$ is called a \emph{simple path} in $X$ if the atoms are distinct and  $\lbrace v_{j-1}, v_{j}\rbrace $ are bonds for $j\in\lbrace 1,\ldots,n-1\rbrace$. If  $(v_1,\ldots,v_{n-1})$ is a simple path and $v_{n-1}$ is connected to $v_n = v_1$ by a bond, $p$ is a  \emph{cycle} in $X$.  We say that a configuration is \emph{connected} if each two atoms are joinable through a simple path. (Note that this definition is consistent with the one given before the statement of Lemma \ref{lemma:reduction}.) A bond is called \emph{acyclic} if it is not contained in any cycle of the bond graph.  The \emph{reduced bond graph} of $X$ is obtained by first deleting all acyclic bonds and then all atoms which are not connected to any other atom. By a \emph{face} of $X$ we always mean a face of its reduced bond graph. The boundary of a face is given by a disjoint union of cycles and by a unique cycle if the reduced bond graph is connected. Such a boundary is called a \emph{polygon} and, in particular,  a $j$-gon if it consists of $j\in \N$ atoms. 
 
\noindent \emph{Sub-configuration:} We say that $Z \subset X$ is a \emph{sub-configuration} of $X$. All notions defined above are defined analogously for any sub-configuration $Z$ of $X$.

\noindent \emph{Face defect:} We define the \emph{face defect} of a sub-configuration $Z \subset X$ by
 \begin{align}\label{def:face defect} 
 \eta(Z) = \sum\nolimits_{j \geq 3} \,  (j-3)f_j(Z),
\end{align}  
where $f_j(Z)$ denotes the number of polygons with $j$ atoms in the bond graph of $Z$.

 \noindent \emph{Strong connectedness:} We say that a configuration $Z$ is \emph{strongly connected}  if   $Z \setminus \lbrace x \rbrace$ is connected for every $x \in Z$. Note that strongly connected  graphs with more than two atoms coincide with their reduced bond graph as they do not contain acyclic bonds since removing one of the atoms belonging to the bond would disconnect the configuration.   

\noindent \emph{Maximal components:} Fix $Q^\nu_T(y)$. Let $z^+,z^- \in \mathcal{Z}   $ and consider $X \subset \mathbb{R}^2$ such that $X= \mathscr{L}(z^\pm)$ on $\partial^\pm_1Q^\nu_T(y)$.  We denote the \emph{set of strongly  connected subsets of lattices}  by
\begin{align*}
\mathcal{C}^\pm = \big\{Z \subset X \cap \mathscr{L}(z^\pm) \colon \, Z \cap \partial^\pm_1 Q^\nu_T(y) \neq \emptyset, \, Z   \text{ is strongly  connected}\big\}.
\end{align*}
We introduce the \emph{maximal components},  denoted by $M^\pm$, as  the maximal elements  in $\mathcal{C}^\pm$ with respect to set inclusion. These sets can be written as 
\begin{align}\label{def:Xpm}
M^\pm = \bigcup\nolimits_{Z \in \mathcal{C}^\pm} Z.
\end{align}
Note that $M^+=\emptyset$ or $M^-=\emptyset$ if $z^+=\mathbf{0}$ or $z^-=\mathbf{0}$, respectively. Moreover, we point out that $M^\pm$ are in general not subsets of $Q^\nu_T(y)$. We illustrate $M^\pm\cap Q^\nu_T$ in Figure \ref{fig:Xpm}.

\begin{figure}[H]
 \includegraphics{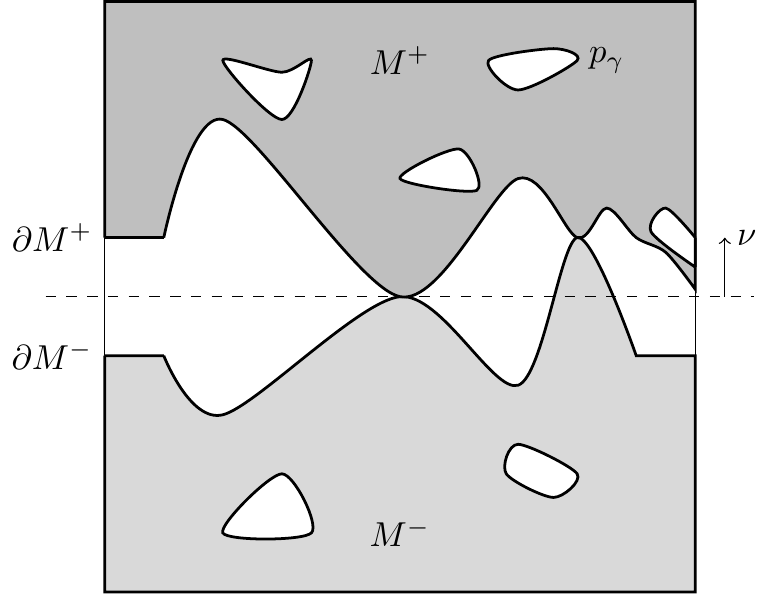}
\caption{A schematic picture of $M^+ \cap Q^\nu_T(y)$, depicted in dark gray, and of $M^- \cap Q^\nu_T(y)$, depicted in light gray. Their boundaries are illustrated in bold. We depict also a curve $p_\gamma$ considered in Step 2 of the proof below.}\label{fig:Xpm}
\end{figure}

\begin{lemma}[Simple paths in maximal components]\label{lemma:auxiliarypath}  Let $\gamma = (x_1,\ldots,x_k)$ be a simple path in $X$ with $x_1,x_k \in M^+$ (or both in $M^-$) such that $x_2,\ldots, x_{k-1} \notin M^+ $ (or $x_2,\ldots, x_{k-1} \notin M^- $, respectively). Then $k \geq 4$.
\end{lemma}

\begin{proof}
Let $\gamma$ be as in the statement, without restriction with $x_1,x_k \in M^+$. Recall that $M^+ \subset \mathscr{L}(z^+)$.   If we had $k=3$, then we would necessarily get $x_2 \in \mathscr{L}(z^+)$, as well, see Figure \ref{fig:pathlength3}. This, however, contradicts the choice of the maximal component $M^+$.  In fact, also $M^+\cup \lbrace x_2 \rbrace$ would be a strongly connected set. 
\end{proof}

\begin{figure}[H]
 \includegraphics{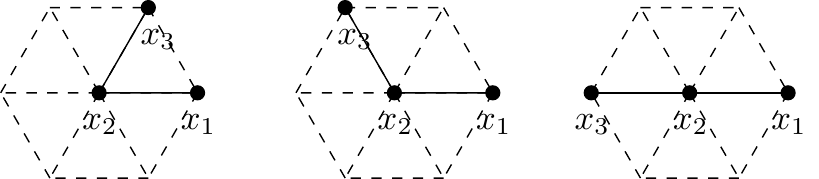}
\caption{The three different (up to rotation and reflection) possibilities of paths of length $3$.}
\label{fig:pathlength3}
\end{figure}

\begin{proof}[Proof of Lemma \ref{lemma:reduction}]  Without restriction we assume $z^+ \ne z^-$.  The proof strategy is as follows: we first show that $X$ consists of at most two connected components which contain the lower and the upper part of the boundary, respectively (Step 1). We are then left with at most two connected components which contain the maximal components $M^\pm$ defined in \eqref{def:Xpm}.  Then, we prove that these components $M^\pm$ do not contain holes. This ensures that $\partial M^\pm \cap Q^\nu_T(y)$ are simple paths  (Step 2). Finally, we show that there are no parts of $X$ that may be connected to $M^\pm$, but that are not subsets of the upper and lower  lattice $\mathscr{L}(z^\pm)$ (Step 3). Steps 1--3 are proved by contradiction, i.e., we suppose that  $X$ did not satisfy the abovementioned properties and then we show  that the configuration can be modified in such a way that the energy strictly decreases.   Some technical estimates are given in Steps 4--5.

Fix $z^\pm \in \mathcal{Z}$,  $\nu \in \mathbb{S}^1$,  $T>0$, and $y \in \R^2$. Denote by $X \subset \R^2$ a minimizer of \eqref{eq: one inequl}.  Without loss of generality we assume that 
\begin{align}\label{eq:draussen-weg} 
  X \subset \{x \in \overline{(Q_T^\nu(y))_{ 1}} \colon \mathcal{N}(x) \cap Q_T^\nu(y) \neq \emptyset \} \cup \partial_1^+ Q_T^\nu(y) \cup \partial_1^- Q_T^\nu(y). 
\end{align} 
In particular, we have  $X=\mathscr{L}(z^\pm)$ on $\partial_1^\pm Q_T^\nu(y)$.  By $M^\pm$ we denote  its maximal upper and lower component, respectively, given by (\ref{def:Xpm}). (Recall that $M^+=\emptyset$ or $M^-=\emptyset$ if $z^+=\mathbf{0}$ or $z^-=\mathbf{0}$.) Without restriction we assume that $z^\pm = (\theta^\pm,\tau^\pm,1)$. Otherwise, we apply all arguments just to the component $z^\pm$ with $z^\pm \neq \mathbf{0}$.

\noindent \emph{Step 1: $X$ has at most two connected components  in $Q_T^\nu(y)$  and $\#\mathcal{N}(x) \ge 2$ for all $x \in X  \cap Q_T^\nu(y)$.} First, we observe that the maximal components $M^+$ and $M^-$ are either contained in one single or in two different connected components of $X$. Assume by contradiction that the configuration $X$ consists of more than the (at most two) connected components containing $M^\pm$. Then we can  remove the other connected components not containing $M^\pm$ and obtain a new configuration  which  has strictly less energy and the same boundary data as $X$. This follows directly from the definition of the energy in \eqref{def:energyneighbourhood}. 

 Moreover, if there exists $ x'  \in X$ such that $\#\mathcal{N}( x' ) \le 1$,  then we can consider the configuration $X\setminus \{ x' \}$ to obtain a configuration with strictly less energy since, by \eqref{def:energyneighbourhood}, we have 
\begin{align*}
  E_1(X,Q^\nu_T(y)) 
	&= \frac{1}{2}\sum\nolimits_{x \in X\cap Q^\nu_T(y)} (6-\#\mathcal{N}(x)) 
	\ge E_1\big(X\setminus \lbrace x' \rbrace,Q^\nu_T(y)\big) + 2. 
\end{align*}  
\noindent \emph{Step 2: $\partial M^\pm$ is a simple path.}  In this step, we show that each of  the sets $\partial M^\pm$ defined in \eqref{eq: boundary of Y} is a  simple path in $X$ joining  the lateral faces of $Q^\nu_T(y)$. More precisely, let  
\begin{align*}
H_{\nu^\perp,-}^T(y)  := \{x \in \mathbb{R}^2\colon  \langle(x-y), \nu^\perp\rangle  <  - T/2 \} \ \  \text{ and } \ \  H_{\nu^\perp,+}^T(y)  = \{x \in \mathbb{R}^2\colon \langle  (x-y), \nu^\perp\rangle  \ge  T/2 \}. 
\end{align*}
 Then there are $v^\pm_- \in M^\pm \cap H_{\nu^\perp,-}^T(y)$ and $v^\pm_+ \in M^\pm \cap H_{\nu^\perp,+}^T(y)$ such that $\{v^\pm_-, v^\pm_+\} \cup \partial M^\pm$ is a simple path with first element $v^\pm_-$ and last element  $v^\pm_+$. 

To prove this, we color each (closed) equilateral triangle of sidelength $1$ all of whose corners are contained in $M^\pm$ in dark/light gray, respectively, see Figure \ref{fig:Xpm}. We first show that there are no cycles in $\partial M^\pm$. Since $M^\pm$ is strongly connected, this also yields that the colored regions inside $Q_T^\nu(y)$ are simply connected and that $\partial M^\pm$ lies on the boundary of the respective colored region. Assume by contradiction that  there exists a cycle $p=(v_1,\ldots,v_{n}) \subset M^\pm$ with $v_{n}=v_1$. Denote by $\mathrm{int}(p)$ the interior connected component of the curve  
 $$p_\gamma=\bigcup\nolimits_{i=1}^{n-1}[v_i; v_{i+1}],$$
see Figure \ref{fig:Xpm}. Now define
\begin{align*}
\tilde{X} = \begin{cases} \mathscr{L}(z^+)&\text{in } \mathrm{int}(p),\\
X &\text{otherwise.}
\end{cases}
\end{align*}
Since we did not change the neighborhood of each atom $x \in Q^\nu_T(y) \setminus \overline{\mathrm{int}(p)}$, we obtain by \eqref{def:energyneighbourhood} and Lemma \ref{lemma:propertiesofE}(iv)
\begin{align*}
E_1\big(\tilde{X},Q^\nu_T(y)\big) &= E_1\big(\tilde{X},\overline{\mathrm{int}(p)}\big) + E_1\big(\tilde{X},Q^\nu_T(y) \setminus \overline{\mathrm{int}(p)}\big) \\&<  E_1\big(X,\overline{\mathrm{int}(p)}\big) + E_1\big(X,Q^\nu_T(y) \setminus \overline{\mathrm{int}(p)}\big) = E_1\big(X,Q^\nu_T(y)\big),
\end{align*} 
 where we have used that $\#\mathcal{N}(x) = 6$ for all  $x \in \tilde{X} \cap \mathrm{int}(p)$ and that every $x \in p$ has at least as many bonds in $\tilde{X}$ as in $X$, while for at least one $x \in p$ the number of bonds has increased.  We have constructed a configuration $\tilde{X}$ with  strictly less energy and the same boundary data as $X$. This contradicts the fact that  $X \subset \R^2$ is a minimizer of \eqref{eq: one inequl}, and shows that there are no such cycles in $M^\pm$.

We next show that even the complement of each colored region inside $Q^\nu_T(y)$ is connected. If this were not the case,  without restriction we assume for contradiction that there are $v, w \in M^+ \cap H_{\nu^\perp,+}^T(y)$  such that there is a simple path with first element $v$, last element $w$, and intermediate elements in $\partial M^+$, whose bonds together with a segment in $\partial Q_T^\nu(y)$ bound a region free of dark triangles. By the boundary conditions, we can suppose that   $ 6  \ge \langle v, \nu\rangle > \langle w, \nu \rangle \ge  -6 $, see also Figure~\ref{fig:Xpm}. We extend it to a cycle $p$ by placing additional atoms in $\mathscr{L}(z^+) \cap \overline{(Q_T^\nu(y))_\varepsilon} \cap H_{\nu^\perp,+}^T(y)$. Our assumptions on $X$ specified in \eqref{eq:draussen-weg}  and Step 1 guarantee that each point in $\mathscr{L}(z^+)$ on or inside of $p$ has distance at least $1$ to every atom of the connected component of $X$ that contains $M^-$. Now let 
\begin{align*}
\tilde{X} 
= \begin{cases} \mathscr{L}(z^+)&\text{in } \overline{\mathrm{int}(p)},\\ 
  X &\text{in } \R^2 \setminus \overline{\mathrm{int}(p)}, \\
  \emptyset &\text{otherwise.}
\end{cases}
\end{align*} 
Similarly as before we get $E_1(\tilde{X},Q^\nu_T(y)) < E_1(X,Q^\nu_T(y))$, which shows that also this situation does not occur. We conclude that each $M^\pm$ is strongly connected and both the dark and the light colored areas have connected complements relative to $Q^\nu_T(y)$.

We claim that $\partial M^\pm$ has to be a simple path. Assume by contradiction that this  were  not the case, e.g., for $M^+$.  Then, since $\partial M^+$ lies on the boundary of the region in dark gray  being the union of triangles, we find $x \in \partial M^+$ which is a corner of exactly two of these triangles and these triangles share only $x$ as a common point, see Figure \ref{fig:nonsimple}.  Since $\partial M^+$ does not contain cycles, we find $x^+,x^- \in \mathcal{N}(x)$ such that each path in $M^+$ connecting $x^+$ with $x^-$ contains $x$. This, however, contradicts the strong connectedness of $M^+$, and shows that $\partial M^+$ is a simple path. This concludes Step 2.

\begin{figure}[H]
 \includegraphics{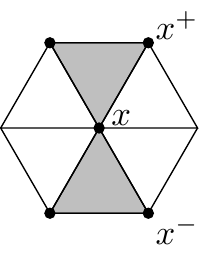}
\caption{A point $x \in \partial M^\pm$ that would make $\partial M^\pm$ a non-simple path. }
\label{fig:nonsimple}
\end{figure}

\noindent \emph{Step 3: Comparison with subsets of the lattice.}  Our goal is to show that there holds $X \subset \mathscr{L}(z^+) \cup \mathscr{L}(z^-)$. Recalling the definition of $M^\pm$ in  \eqref{def:Xpm}, it thus suffices to show that removing the connected components of  $(X \cap Q^\nu_T(y)) \setminus (M^+ \cup M^-)$ would strictly decrease the energy which clearly contradicts the assumption that $X$ is a minimizer. (Recall that we have already reduced to the case that $X$ consists of at most two connected components. Note, however, that $(X \cap Q^\nu_T(y)) \setminus (M^+ \cup M^-)$ might consist of more connected components.)

This will conclude the proof of the statement: it shows that the minimizer $X$ is indeed a  subset of  $\mathscr{L}(z^+) \cup \mathscr{L}(z^-)$. Moreover, the property that  $\partial M^\pm \cap Q^\nu_T(y)$ are  simple  paths joining  the lateral faces of $Q^\nu_T(y)$ has already been addressed in Step 2. Finally, we observe that $\# \mathcal{N}(x) \le 5$ for all $x \in \partial M^\pm$. In fact, $\# \mathcal{N}(x) = 6$ for some $x \in \partial M^\pm$ would entail $\lbrace x \rbrace \cup \mathcal{N}(x) \subset M^\pm$ as $M^\pm \subset \mathscr{L}(z^\pm)$ is the maximal component. This contradicts \eqref{eq: boundary of Y}. 
 
Now, consider a connected component  $X'$ of $(X \cap Q^\nu_T(y)) \setminus (M^+ \cup M^-)$.  We want to prove that 
\begin{align}\label{ineq:comparisonsubset}
 E_1\big(X,Q^\nu_T(y)\big) \geq E_1\big(X \setminus X',Q^\nu_T(y)\big)+1.
\end{align}
We first introduce some further notation. By $\Gamma^\pm \subset \partial M^\pm$ we denote the smallest connected sets $\Gamma^\pm \supset \mathcal{N}(X') \cap M^\pm$, where we define $\mathcal{N}(X') := \bigcup_{x \in X'} \mathcal{N}(x) \setminus X'$.   Define  $\Gamma := \Gamma^+ \cup \Gamma^-$ and $X_\Gamma:= X' \cup \Gamma$. Note that both $\Gamma^-$ and $\Gamma^+$ are simple paths in $X$  since $\partial M^\pm$ are simple paths,   see Figure \ref{fig:cleainingpossibilities}.  For $x \in X_\Gamma$, we introduce the \emph{internal and external neighborhoods} by 
\begin{align}\label{eq: intern/extern}
\mathcal{N}_i(x) = \mathcal{N}(x) \cap X_\Gamma,  \ \ \ \ \ \ \ \  \mathcal{N}_e(x) = \mathcal{N}(x) \setminus X_\Gamma,
\end{align}
i.e., the set of neighbors inside  and outside of  $X_\Gamma$, respectively.  Note that $X_\Gamma$ is connected.  Its reduced bond graph is  delimited by a finite union of   disjoint   cycles.  We denote by $\partial X_\Gamma$ the union of  these  cycles and  by  $d = \# \partial X_\Gamma$  its   cardinality. (The notation is unrelated to \eqref{eq: boundary of Y}.) We further define
\begin{align}\label{eq: many notation}
& f_j = \# j\text{-gons of $X_\Gamma$}, \quad   f= \sum\nolimits_j f_j,   \quad \eta= \eta(X_\Gamma),  \quad  n_{\Gamma} = \#\Gamma, \quad n=\#X_\Gamma,\notag \\ & 
 b_\Gamma= \#\big\{\{x,y\} \colon x,y \in \Gamma, \, y \in \mathcal{N}(x)\big\}, \ \ \ \  b=\#\big\{\{x,y\} \colon x,y \in X_\Gamma, \, y \in \mathcal{N}(x)\big\},\notag \\ & 
 b_{\mathrm{ac}} = \#\big\{\{x,y\} \text{ acyclic}\colon x,y \in X_\Gamma, \, y \in \mathcal{N}(x)\big\},
\end{align}
where $\eta$ was introduced in \eqref{def:face defect}. Note that $f$ corresponds to the number of faces  both in the bond graph and in the reduced bond graph of $X_\Gamma$. We will see that there holds
\begin{align}\label{ineq:eta}
2+d+2b_{\mathrm{ac}}+\eta \geq 3n_\Gamma -b_\Gamma.
\end{align}
We defer the proof of \eqref{ineq:eta} to Steps  4--5  below and proceed to prove \eqref{ineq:comparisonsubset}.

Since in the passage from $X$ to $X \setminus X'$ the neighborhood of atoms outside $X_\Gamma$ is left unchanged and for atoms in $\Gamma$ the neighbors outside of $X_\Gamma \setminus \Gamma$ remain,  in view of  \eqref{def:energyneighbourhood}, we need to check that 
\begin{align}\label{eq: to be checked}
\frac{1}{2}\sum\nolimits_{x \in X_\Gamma}(6- \#\mathcal{N}(x))\geq  \frac{1}{2}\sum\nolimits_{x \in \Gamma} \big(6-(\#\mathcal{N}_e(x) + \#(\mathcal{N}(x) \cap \Gamma)\big)+1. 
\end{align}
We can count the faces to  obtain
\begin{align}\label{eq: doublecounting}
2b-d-2b_{\mathrm{ac}}= \sum\nolimits_{j \geq 3}\, jf_j = \eta + 3f.
\end{align}
Indeed, the first identity follows from the fact that in the summation all bonds contained in the union of cycles delimiting the  reduced bond graph of $X_\Gamma$ are counted only once, the acyclic bonds are not counted, and all other cyclic bonds are counted twice. The second identity follows from  \eqref{def:face defect}. As the bond graph is planar and connected, we can apply Euler's formula (omitting the exterior face) to get $n-b+f =1$. Then, by \eqref{ineq:eta} and   (\ref{eq: doublecounting})  we derive
$$3n - b \ge   3n_\Gamma -b_\Gamma +1.$$
By the definitions in  \eqref{eq: intern/extern}--\eqref{eq: many notation} and the facts that $\sum\nolimits_{x \in X_\Gamma} \#\mathcal{N}_i(x)=2b$, $\sum_{x\in \Gamma} \#(\mathcal{N}(x) \cap \Gamma) = 2b_\Gamma$ this implies 
\begin{align}\label{eq: to be checked-2}
\frac{1}{2}\sum\nolimits_{x \in X_\Gamma}(6- \#\mathcal{N}_i(x))\geq  \frac{1}{2}\sum\nolimits_{x \in \Gamma} \big(6- \#(\mathcal{N}(x) \cap \Gamma)\big)+1. 
\end{align}
Now we note that 
$\#\mathcal{N}(x) -\#\mathcal{N}_e(x) =\#\mathcal{N}_i(x)$ for $x\in \Gamma$ and  $\mathcal{N}(x) = \mathcal{N}_i(x)$ for $x \in X_\Gamma \setminus \Gamma$, see \eqref{eq: intern/extern}. This along with \eqref{eq: to be checked-2} shows the desired estimate \eqref{eq: to be checked}. To conclude the proof, it remains to show  \eqref{ineq:eta}.

\begin{figure}[H]
 \includegraphics{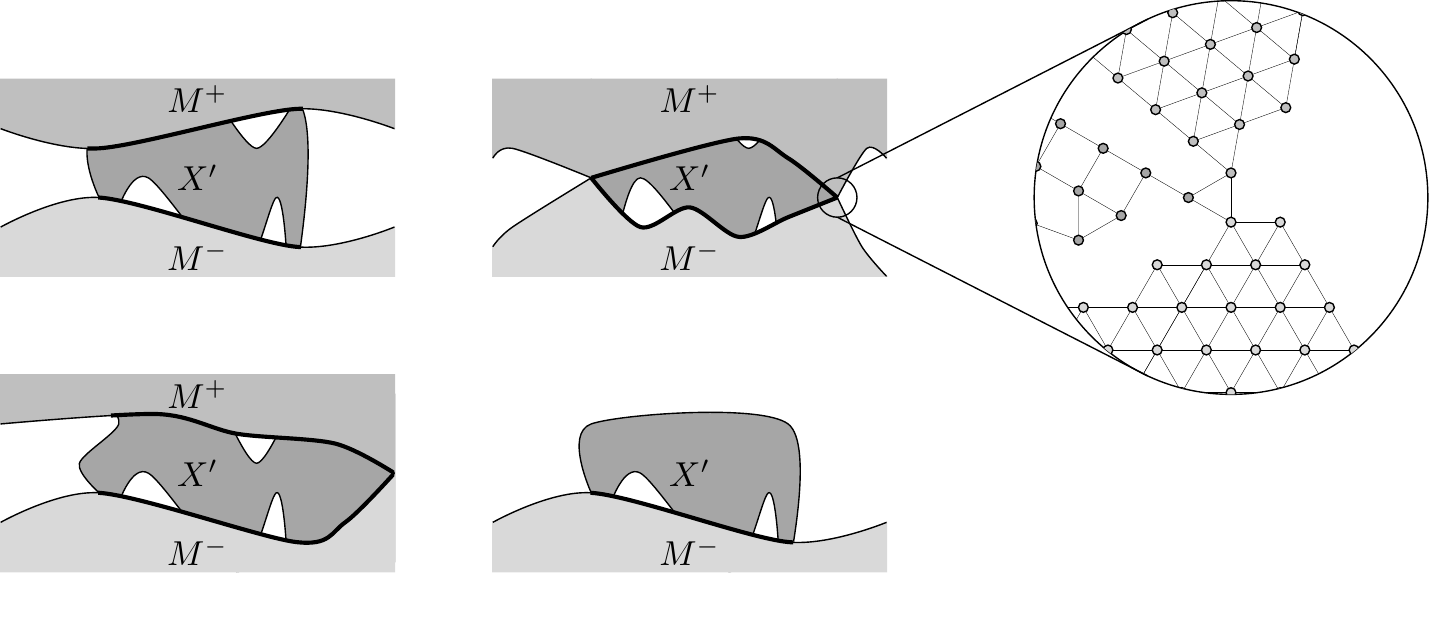}
\caption{The  different possibilities of $X'$ touching $M^\pm$ corresponding to case (a) on the top left, case (b) on the top  right, and case (c) in the two bottom pictures. $M^+$ is always depicted in gray, $M^-$ in light gray, and $X'$ in dark gray. $\Gamma^+$ and $\Gamma^-$ are depicted by the bold black lines.}
\label{fig:cleainingpossibilities}
\end{figure}

  \noindent \emph{Step 4: Proof of  \eqref{ineq:eta}.}  Recall that $\Gamma$ consists of the two simple paths $\Gamma^+$ and $\Gamma^-$. We need to distinguish three cases: 
\begin{align*}
\text{(a) $\Gamma$ is not connected, \ \  \ \ \ (b) $\Gamma$ is a cycle, \ \ \ \  \ (c)  $\Gamma$ is a simple path.}
\end{align*} 
Since $\Gamma^\pm$ are simple paths, and the bond graph of  $X^\prime$ is planar and connected, we see that these are all possibilities that may occur, see Figure \ref{fig:cleainingpossibilities} for an illustration of the different cases. At this point, we also use that $\Gamma^\pm$ are the  smallest connected sets with $\Gamma^\pm \supset \mathcal{N}(X') \cap M^\pm$  and $\Gamma^\pm \subset \partial M^\pm$, where $\partial M^\pm$ is a simple path connecting $H_{\nu^\perp,-}^T(y) \cap \mathscr{L}(z^\pm)$ and $H_{\nu^\perp,+}^T(y) \cap \mathscr{L}(z^\pm)$.  

First of all, we observe that 
\begin{align}\label{eq:nGamma-bGamma}
\text{Case (a):} \ \ n_\Gamma \leq b_\Gamma +2, \ \ \ \ \ \ \text{Case (b):} \ \ n_\Gamma \leq  b_\Gamma, \ \ \ \ \ \ \text{Case (c):} \ \ n_\Gamma \leq b_\Gamma +1.
\end{align}
 This is due to the fact that the bond graph of $\Gamma$ contains  $\Gamma^\pm$  and a  simple path containing $k$ bonds consists of $k+1$ atoms, and in a cycle the number of bonds equals the number of atoms.   (As there may be more bonds present if there are triangles in the bond graph, we get inequalities.) 
 
Using \eqref{eq:nGamma-bGamma}, it suffices to prove
\begin{align}\label{estimate:eta}
d+2b_{\mathrm{ac}}+\eta \geq \begin{cases} 2n_\Gamma &\text{in case (a)},\\
2n_\Gamma-2 &\text{in case (b)},\\
2n_\Gamma -1&\text{in case (c)},
\end{cases}
\end{align}
where   $d$, $\eta$, $n_\Gamma$, and $b_\mathrm{ac}$ are defined in \eqref{eq: many notation}.   This will rely on the estimate
\begin{align}\label{ineq:eta(a,b,c)}
\eta \geq n_\Gamma-2.
\end{align} 
 We first show \eqref{estimate:eta} in the three cases and defer the proof of \eqref{ineq:eta(a,b,c)} to Step 5.  Observe that if a connected component $\tilde{\Gamma}$ of $\Gamma$ satisfies $\tilde{\Gamma} \not \subset \partial X_\Gamma$, then $\#\tilde{\Gamma} = 1$ and $\tilde{\Gamma}$ connects to $X'$ by one acyclic bond. This follows from the observation that, whenever $x \in \tilde{\Gamma}$ satisfies $\mathcal{N}(x) \cap X_{\Gamma} \ge 2$, then $x$ lies on a cycle in $X_\Gamma$ and thus, as an element of $\Gamma$, is contained in $\partial X_\Gamma$.

\noindent \emph{Case {\normalfont (a)}:} Suppose first $\Gamma \subset \partial X_\Gamma$.  Since $\partial X_\Gamma$ is a  disjoint  union of cycles and $\Gamma$ consists of two simple paths, we get  $\#(\partial X_\Gamma \setminus \Gamma ) \ge 2$.  In fact, if $\Gamma^+$ and $\Gamma^-$ intersect the same cycle of $\partial X_\Gamma$, this follows from the fact that $\Gamma^+ \cup \Gamma^-$ is not connected. If $\Gamma^+$ and $\Gamma^-$ intersect different cycles of $\partial X_\Gamma$, it suffices to  use that $\Gamma^\pm$ are not cycles. This shows $d\geq n_\Gamma + 2$. Then \eqref{ineq:eta(a,b,c)} implies \eqref{estimate:eta}.  If $\Gamma^-\subset \partial X_\Gamma$, $\Gamma^+ \not\subset \partial X_\Gamma$, then, as before, $\#(\partial X_\Gamma \setminus \Gamma^-) \ge 1$ and thus $d \ge \#\Gamma^- + 1$. The observation below \eqref{ineq:eta(a,b,c)} gives $\# \Gamma^+=1$ and $b_{\mathrm{ac}} \ge 1$, so particularly $d\geq n_\Gamma$. Then again \eqref{ineq:eta(a,b,c)} implies \eqref{estimate:eta}. The case $\Gamma^-\not \subset \partial X_\Gamma$, $\Gamma^+ \subset \partial X_\Gamma$ is analogous. Finally, if $\Gamma^-, \Gamma^+ \not\subset \partial X_\Gamma$, then $n_{\Gamma} = 2$ and $b_{\mathrm{ac}} \ge 2$ since $\Gamma^-$ and $\Gamma^+$ cannot be connected to $X'$ by the same (acyclic) bond. This proves \eqref{estimate:eta}.

\noindent \emph{Case {\normalfont (b)}:} Since $\Gamma$ is a cycle, we get $\Gamma \subset \partial X_\Gamma$. Thus, we obtain  $n_\Gamma \le d$  and \eqref{ineq:eta(a,b,c)} yields \eqref{estimate:eta}. 

\noindent \emph{Case {\normalfont (c)}:}  Suppose first that $\Gamma \subset \partial X_\Gamma$.  Since  $\Gamma$  is not a cycle  and $\partial X_\Gamma$ is a union of cycles, we get  $\# (\partial X_\Gamma \setminus \Gamma) \ge 1$.   This implies $d \geq n_\Gamma+1$.  Then \eqref{ineq:eta(a,b,c)} again yields \eqref{estimate:eta}.  If $\Gamma \not\subset \partial X_\Gamma$, then $n_{\Gamma} = 1$ and $b_{\mathrm{ac}} \ge 1$, from which \eqref{estimate:eta} follows.

\noindent \emph{Step 5: Proof of \eqref{ineq:eta(a,b,c)}.} It remains to check  \eqref{ineq:eta(a,b,c)}. To this end, we classify the polygons in the  (reduced)  bond graph of $X_\Gamma$ in the following way: for $k \ge 1$, we set  
\begin{align*}
\partial\text{-}k\text{-gon}= \{P \text{ polygon in } X_\Gamma \colon \, \#(P \cap \Gamma) = k \} \ \ \ \text{ and } \ \ \  \partial\text{-gon} = \bigcup\nolimits_{k\geq 1} \partial\text{-}k\text{-gon},
\end{align*}
and define  $D_k=\#\partial\text{-}k\text{-gon}$.  In order to estimate the cardinality of $P \in \partial\text{-}k\text{-gon}$, we introduce the following condition: 
\begin{align}\label{eq:alpha}
\text{ there exist} \ \ \ x_+ \in M^+ \cap P \ \ \ \text{and} \ \ \ x_- \in  (M^- \setminus M^+)  \cap P \ \ \ \text{ with } \ \ \  |x_+-x_-|=1.
\end{align} 
We claim that always $\#P \geq k+1$, while in case \eqref{eq:alpha} does not hold there holds $\#P \geq k+2$. 

To see the first claim we note that clearly $\#P \geq k$. If $\#P = k$, then $P \subset \Gamma$ and $\Gamma$ is a cycle, hence $P = \Gamma$. But then all bonds connecting $\Gamma$ and $X'$ are acyclic. As observed below \eqref{ineq:eta(a,b,c)}, this entails $\# \Gamma = 1$ which, however, is not possible in case $\Gamma$ is a cycle. 

Assume now \eqref{eq:alpha} does not hold. First, suppose that  $P \cap \Gamma \subset M^+$ or $P \cap \Gamma \subset M^-$. If $k=1$, the statement  $\#P \geq k+2$  is clear as $\# P \ge 3$. If $k \ge 2$, we can choose a simple path in $P$ such that only the first and the last atom lie in $M^+$ (or $M^-$, respectively).  The statement then follows from Lemma \ref{lemma:auxiliarypath}. On the other hand, if   $P \cap  (M^+ \setminus M^-) \neq \emptyset$ and $P \cap (M^- \setminus M^+) \neq \emptyset$,   then there exist two simple paths contained in $P$ joining   $M^+ \setminus M^-$ and $M^- \setminus M^+$.   Since \eqref{eq:alpha} does not hold, each of these two paths contains an atom that is not contained in $\Gamma$. This implies $\#P \geq k+2$.

We are now in a position to prove  (\ref{ineq:eta(a,b,c)}).  By the definition of $\eta$ and the cardinality estimate for $\partial\text{-}k\text{-gons}$  we obtain 
\begin{align}\label{eq:LLL1}
\begin{split}
\eta = \sum\nolimits_{j \geq 3}  f_j(j-3) \geq \sum\nolimits_{k\geq 1}\, D_k(k+2-3) -N  \geq \sum\nolimits_{k\geq 1}\, D_k(k-1) - \begin{cases} 0 &\text{in case (a)},\\
2 &\text{in case (b)},\\
1&\text{in case (c)},
\end{cases} 
\end{split}
\end{align}
where $N$ denotes the number of $\partial$-gons satisfying case \eqref{eq:alpha}. We used  that:  in case (a) we have $N = 0$ since otherwise $\Gamma$ would be connected, in case (b) the fact that $X'$ is connected and the planarity of the bond graph  imply that $N \le 2$, and in case (c) we get $N \le 1$ since $\Gamma$ is a simple path.  Finally, we claim that 
\begin{align}\label{eq:LLL2}
\sum\nolimits_{k\geq 1} D_k(k-1)\geq \begin{cases} n_\Gamma -2 &\text{in case (a)},\\
n_\Gamma &\text{in case (b)},\\
n_\Gamma -1&\text{in case (c)},
\end{cases}
\end{align}
 Indeed, this follows from  the fact  that each bond in between two successive atoms $x,y \in \Gamma$ is contained in exactly one  $\partial$-gon and $k-1$ estimates from above the number of bonds between atoms in $\Gamma\cap P$ whenever $P \in \partial\text{-}k\text{-gon}$  as otherwise $P = \Gamma$ and $\# P = k$ which we have excluded above.  (The estimate is strict if $\Gamma \cap P$ is not connected.) By combining \eqref{eq:LLL1}--\eqref{eq:LLL2} we obtain (\ref{ineq:eta(a,b,c)}). This concludes the proof. 
\end{proof}

\begin{proof}[Proof of Lemma \ref{lemma:grain-bonds}]
Without restriction we assume that $z^+ \ne z^-$. Let $X^\pm$ be as in the statement of Lemma \ref{lemma:reduction}, i.e., $X^\pm = M^\pm$. We define 
\begin{align*}
  Y^+ 
  &= X^+ \setminus  (\partial X^+ \cap \partial X^-) \cup \big\{ x \in \partial X^+ \cap \partial X^- \colon \#(\mathcal{N}(x) \cap X^+) \ge \#(\mathcal{N}(x) \cap X^-)\big\}, \\ 
  Y^- 
  &= X^- \setminus  (\partial X^+ \cap \partial X^-)  \cup \big\{ x \in \partial X^+ \cap \partial X^- \colon \#(\mathcal{N}(x) \cap X^+) < \#(\mathcal{N}(x) \cap X^-)\big\}. 
\end{align*}
\noindent \emph{Proof of $\mathrm{(i)}$.}  Property (i) is obviously satisfied by construction.

\noindent \emph{Proof of $\mathrm{(ii)}$.}
 As a preparation, let us note that,  if $x \in X^+ \cap X^-$, then $\mathcal{N}(x) \cap  X^+ \cap X^- = \emptyset$ since $z^+ \ne z^-$. Moreover, if $x \in X^+ \cap X^- \cap Q^\nu_T(y) = \partial X^+ \cap \partial X^-$, then $\#\mathcal{N}(x) \le 5$ by  Lemma \ref{lemma:reduction}(ii).  Since $X^\pm$  is strongly connected, we also have $\#(\mathcal{N}(x) \cap X^\pm) \geq 2$. Our definition of $Y^\pm$ then entails 
\begin{align}\label{imp: choice}
x \in X^\pm \setminus Y^\pm \implies \#(\mathcal{N}(x) \cap X^\pm) = 2.
\end{align}
 This ensures $Y^\pm = X^\pm = \mathscr{L}(z^\pm)$ on $\partial^\pm_1 Q^\nu_T(y)$. Furthermore, it entails $\partial Y^\pm \subset \partial X^\pm$.  Indeed, $y \in \partial Y^\pm \setminus \partial X^\pm$  would give $\# (\mathcal{N}(y)  \cap X^\pm) = 6$ and $\# (\mathcal{N}(y)  \cap Y^\pm) \le  5$, i.e., there exists $x \in X^\pm \setminus Y^\pm$ with $|x-y| = 1$. But then $\#(\mathcal{N}(x) \cap \mathcal{N}(y)\cap X^\pm)=2$, which yields the contradiction $\#(\mathcal{N}(x) \cap X^\pm) \ge 3$. 

\noindent \emph{Proof of $\mathrm{(iii)}$.} Since $X^\pm$ is simply connected and $x \in \partial X^\pm \setminus \partial Y^\pm$ is only possible if $ \# (\mathcal{N}(x) \cap X^\pm) = 2$ (see \eqref{imp: choice}), we get that $\partial Y^\pm$ is a simple path connecting the lateral faces of $Q^\nu_T(y)$. More precisely, by Step 2 of the proof of Lemma \ref{lemma:reduction}, there are $v^\pm_- \in X^\pm \cap H_{\nu^\perp,-}^T(y)$ and $v^\pm_+ \in X^\pm \cap H_{\nu^\perp,+}^T(y)$ such that $\{v^\pm_-, v^\pm_+\} \cup \partial Y^\pm$ is a simple path with first element $v^\pm_-$ and last element $v^\pm_+$. The bonds between any two consecutive atoms in this chain form a polygonal line and we denote by $\alpha(x)$ the (interior) angle it forms at atom $x$. 

As the first and the last segments cross the lateral faces of $Q^\nu_T(y)$ and $Y^\pm$ is strongly connected, we have 
$$ \sum\nolimits_{x \in \partial Y^\pm} (\pi - \alpha(x)) 
   \in \frac{1}{3} \{ -2\pi, -\pi, 0, \pi, 2\pi \}. $$  
Since $X^\pm$ is simply connected, due to \eqref{imp: choice}, the same holds true for $Y^\pm$. Hence, $\alpha(x)$ relates to the number of neighbours of $x$ within $Y^\pm$ by the formula 
$$ \alpha(x) 
   = \frac{1}{3} \big( \#(\mathcal{N}(x) \cap Y^\pm) -1 \big) \pi. $$
As a consequence we obtain 
\begin{align*}
  \Big| \sum\nolimits_{x \in \partial Y^\pm} \big( \# (\mathcal{N}(x) \cap Y^\pm) - 4 \big) \Big| 
	= \Big| \frac{3}{\pi} \sum\nolimits_{x \in \partial Y^\pm} \big( \alpha(x) - \pi \big) \Big| 
	\le 2.  
\end{align*}
This concludes the proof. 
\end{proof}

We summarize our main findings on the structure of grain boundaries obtained in the proof of Lemma \ref{lemma:reduction} in the following theorem.  
\begin{theorem}[Reduction to subsets of two lattices]\label{theorem:grain-boundary} Let $z^+,z^- \in \mathcal{Z}$, $z^+ \ne z^-$, $\nu \in \mathbb{S}^1$, $y \in \R^2$, and $T>0$. Let $X \subset \mathbb{R}^2$ be a minimizer of
\begin{align*}
\min\Big\{E_1\big(X,Q^\nu_T(y)\big)\colon \  X = \mathscr{L}(z^\pm) \text{ \rm on } \partial^\pm_1 Q_T^\nu(y) \Big\}. 
\end{align*}
Then $X = M^+ \cup M^-$ on $Q_T^\nu(y)$, where $M^+, M^-$ are the maximal components of $X$, see \eqref{def:Xpm}. Coloring each (closed) equilateral triangle of sidelength $1$ all of whose corners are contained in $M^\pm$ in dark/light gray, yields two simply connected plain regions containing $\partial^\pm_1 Q^\nu_T(y)$, respectively, whose boundary part inside of $Q^\nu_T(y)$ is given by a simple path of atoms. 
\end{theorem}

\section{Characterization of solid-vacuum/solid-solid interactions}\label{sec: Part II}

This section is devoted to establish a relation between the cell formula $\Phi$ defined in \eqref{eq: Phi def} and the density $\varphi_{\rm hex}$ given in \eqref{eq: phi-hex-def}. In particular, we will analyze the situation where the two lattices $\mathscr{L}(z^+)$ and $\mathscr{L}(z^-)$, which determine the admissible configurations at the boundary, allow for \emph{touching points}, i.e., atoms $x^+ \in \mathscr{L}(z^+)$ and $x^- \in \mathscr{L}(z^-)$ with $|x^+ - x^-| =1$.  We start by formulating the two results of this section.

\begin{lemma}[Relation of $\Phi$ and $\varphi_{\rm hex}$]\label{lemma:vacuumirrational} There exists a universal constant $C>0$ such that  for each $\nu \in \mathbb{S}^1$ and for every sequence of centers $\lbrace y_T\rbrace_T$ the following properties hold:\smallskip\\
\noindent {\rm (i)} If $z^+=(\theta,\tau,1) \in \mathcal{Z}$ and $z^- = \mathbf{0}$ or if $z^+ = \mathbf{0}$  and  $z^-=(\theta,\tau,1) \in \mathcal{Z}$, there holds for all $T>0$
\begin{align*}
\Big| \frac{1}{T} \min\big\{E_1\big(X_T,Q^\nu_T(y_T)\big) \colon \,  X_T = \mathscr{L}(z^\pm) \text{ \rm on } \partial^\pm_1
Q^\nu_T(y_T)\big\}  -\varphi_{\mathrm{hex}}\big(e^{-i\theta} \nu\big)\Big| \le C/T.
\end{align*}
\smallskip
\noindent {\rm (ii)} For all $z^+ =(\theta^+,\tau^+,1)$, $z^- =(\theta^-,\tau^-,1) \in \mathcal{Z}$ there holds for all $T>0$
\begin{align*}
\frac{1}{T} \min\big\{E_1\big(X_T,Q^\nu_T(y_T)\big) \colon \,  X_T = \mathscr{L}(z^\pm) \text{ \rm on } \partial^\pm_1
Q^\nu_T(y_T)\big\} \leq  \varphi_{\mathrm{hex}}\big(e^{-i\theta^+} \nu\big)+ \varphi_{\mathrm{hex}}\big(e^{-i\theta^-} \nu\big) + C/T.
\end{align*}
Moreover, if $z^+ \ne z^-$, then also 
\begin{align*}
\frac{1}{T} \min\big\{E_1\big(X_T,Q^\nu_T(y_T)\big) \colon \,  X_T = \mathscr{L}(z^\pm) \text{ \rm on } \partial^\pm_1
Q^\nu_T(y_T)\big\} \geq \frac{1}{2} \varphi_{\mathrm{hex}}\big(e^{-i\theta^+} \nu\big)+ \frac{1}{2} \varphi_{\mathrm{hex}}\big(e^{-i\theta^-} \nu\big) - C/T.
\end{align*}  
\end{lemma}

Note that this lemma indeed provides a relation between $\varphi_{\rm hex}$ and the density $\Phi$  since
\begin{align}\label{eq: a good inquality}
\Phi(z^+,z^-,\nu) \le  \liminf_{T\to +\infty} \frac{1}{T} \min\big\{E_1\big(X_T,Q^\nu_T(y_T)\big) \colon \,  X_T = \mathscr{L}(z^\pm) \text{ on } \partial^\pm_1 Q^\nu_T(y_T)\big\}
\end{align}
for all $z^\pm\in\mathcal{Z}$, $\nu \in \mathbb{S}^1$, and all $\lbrace y_T\rbrace_T$.  We point out that the energy density $\varphi_{\mathrm{hex}}$ has  already been identified in \cite{AuYeungFrieseckeSchmidt:12,DeLucaNovagaPonsiglione:19}. In our exposition,  once the  technical result about reduction to two lattices  (see Lemma \ref{lemma:reduction}) has been achieved, the proof of Lemma \ref{lemma:vacuumirrational}(i) is rather  simple compared to \cite[Theorem 2.2]{DeLucaNovagaPonsiglione:19}. In addition, this version with  convergence  rate is a novel result and is needed in order to prove Proposition \ref{proposition:existence-original}.

The next lemma is a refinement which addresses the question under which conditions on the difference of the  rotation angles $\theta^+ - \theta^-$ equality holds in (ii).  To formulate this statement, recall $\omega= \frac{1}{2}+\frac{i}{2}\sqrt{3}$ from Subsection \ref{subsection:definitions}.  We introduce the \emph{set of good angles}, denoted by ${\mathcal{G}_{\mathbb{A}}}$, as the angles $\theta \in \mathbb{A}$ which can be written as 
\begin{align}\label{eq: good angles}
e^{i\theta}= \frac{v_1}{v_2}, \ \ \   \text{ with } v_1,v_2 \in  \mathscr{L} \setminus \{0\}. 
\end{align} 
Here,  the division of $v_1,v_2 \in \mathbb{C}$ has to be understood in the sense of complex numbers. I.e., such angles correspond to rotations which transform  one lattice point into another one.  Note that  ${\mathcal{G}_{\mathbb{A}}}$ is clearly countable. From an algebraic standpoint, our notion of ${\mathcal{G}_{\mathbb{A}}}$ coincides with those angles   $\theta$ such that $e^{i\theta}$ is a fraction of the commutative ring $\mathscr{L}$.

\begin{lemma}[Touching lattices] \label{lemma: touching} Let $z^\pm=(\theta^\pm,\tau^\pm,1) \in \mathcal{Z}$  be such that
\begin{align}\label{eq: touching lattices}
\Phi(z^+,z^-,\nu) \le \varphi_{\rm hex}\big(e^{-i\theta^-}\nu\big) + \varphi_{\rm hex}\big(e^{-i\theta^+} \nu\big)  - \eta 
\end{align}
for an $\eta > 0$. Then, there exists an optimal sequence $\lbrace X_T\rbrace_T$  for $\Phi(z^+,z^-,\nu)$, see \eqref{lemma:Phi}, such that for all $T>0$ large enough, there holds $X_T \subset \mathscr{L}(z^+_T) \cup  \mathscr{L}(z^-_T)$, where $z^\pm_T=(\theta^\pm_T,\tau^\pm_T,1) \in \mathcal{Z}$, and the rotation angles satisfy    
\begin{align}\label{eq: difference angle}
\theta^+_T - \theta^-_T  =  \theta^+ - \theta^- \in {\mathcal{G}_{\mathbb{A}}} \ \ \  \text{for all $T>0$.}
\end{align}
More precisely, $e^{i(\theta^+ - \theta^-)} = v_1/v_2$ for lattice vectors $v_1,v_2 \in \mathscr{L} \setminus \{0\}$ with $|v_1|,|v_2| \le C_{\eta}$, where $C_{\eta}>0$ only depends on $\eta$. 
\end{lemma}

Condition \eqref{eq: touching lattices} means that the   surface  energy between sub-lattices of  $\mathscr{L}(z^+)$ and $\mathscr{L}(z^-)$ can be strictly less than   the sum of the surface energies corresponding to each lattice interacting with the vacuum. This indicates that there are many atoms (in a certain sense) in $\mathscr{L}(z^+)$ with distance $1$ to atoms in $\mathscr{L}(z^-)$. Therefore, we speak of lattices    which have ``touching points''. The lemma shows two properties of optimal sequences: (i) they can be chosen as a subset of two lattices only, cf.\ also Lemma \ref{lemma:reduction}, (ii) the difference of the corresponding rotation angles is constant and  lies in ${\mathcal{G}_{\mathbb{A}}}$.

We now proceed with the proofs of the two lemmas.

\begin{proof}[Proof of Lemma \ref{lemma:vacuumirrational}]
For the whole proof, we fix $\nu \in \mathbb{S}^1$ and  a sequence of centers $\lbrace y_T\rbrace_T$.   

\noindent \emph{Proof of $\mathrm{(i)}$.} Let $z= (\theta,\tau,1)\in \mathcal{Z}\setminus \lbrace\mathbf{0}\rbrace$. We only prove the result for $z^+ = z$ and $z^- = \mathbf{0}$ since the argumentation for the reflected boundary conditions is the same. We obtain the statement by showing separately the two inequalities, where one is proved by a slicing argument and the other one in a constructive way.

\noindent \emph{Step 1: First inequality.} The goal of this step is to prove
\begin{align*}
\frac{1}{T} \min\big\{E_1\big(X_T,Q^\nu_T(y_T)\big) \colon \,  X_T = \mathscr{L}(z^\pm) \text{ on } \partial^\pm_1
Q^\nu_T(y_T)\big\}  \ge \varphi_{\mathrm{hex}}\big(e^{-i\theta} \nu\big)  -C/T.
\end{align*}
Consider $X_T \subset \mathbb{R}^2$ satisfying $X_T = \mathscr{L}(z)$ on $\partial^+_1 Q^\nu_T(y_T)$, $X_T=\mathscr{L}(\mathbf{0}) = \emptyset$ on $\partial^-_1 Q^\nu_T(y_T)$, and
\begin{align}\label{eq: for-first}
 E_1\big(X_T,Q^\nu_T(y_T)\big)  =  \min \big\{E_1(\tilde{X}_T,Q^\nu_T({y}_T)\Big) \colon \,  \tilde{X}_T = \mathscr{L}(z^\pm) \text{ on } \partial^\pm_1Q^\nu_T({y}_T)\big\}.  
\end{align}
By Lemma \ref{lemma:reduction},  we get that $X_T \subset \mathscr{L}(z) = e^{i\theta}(\mathscr{L}+\tau)$. Recall the definition $\omega= \frac{1}{2}+\frac{i}{2}\sqrt{3}$. We now perform a slicing argument: for $k \in \lbrace 1,2,3\rbrace$,  we define for each $\mu \in \R$
$${
I_k(\mu) :=   \big\{\lambda e^{i\theta}\omega^k + \mu e^{i\theta}  (\omega^k)^\perp   \colon \,  \lambda \in \mathbb{R} \big\} 
}$$
the line in lattice direction $e^{i\theta}\omega^k$ passing through the line $\R e^{i\theta} (\omega^k)^\perp$  at point $\mu e^{i\theta}(\omega^k)^\perp$. We set
\begin{align*}
\mathcal{I}_k =\Big\{\mu \in \mathbb{R} \colon \,   I_k(\mu) \cap \mathscr{L}(z) \neq \emptyset, \ I_k(\mu) \cap [y_T - \tfrac{T}{2} \nu^\perp; y_T + \tfrac{T}{2} \nu^\perp] \Big\}.
\end{align*}
Due to the boundary conditions, up to a  bounded number of times  independent of both $\nu$ and  $T$,  for each $\mu \in \mathcal{I}_k$  we find $x \in X_T \subset \mathscr{L}(z)$ such that  $x+e^{i\theta} \omega^k \notin X_T$ or $x- e^{i\theta} \omega^k \notin X_T$.  (Note that a bounded number of lattice lines,  independent of $T$,  in direction $e^{i\theta}\omega^k$  and passing through $[y_T - \frac{T}{2} \nu^\perp; y_T + \frac{T}{2} \nu^\perp]$ does not intersect $\partial^+_1 Q^\nu_T(y_T)$.)  By \eqref{def:energyneighbourhood} this yields  
\begin{align}\label{ineq:mathcalIk}
E_1\big(X_T, Q^\nu_T(y_T)\big) \geq \sum\nolimits_{k=1}^3 \#\mathcal{I}_k -C
\end{align}
 for a constant $C>0$ independent of $T$. 
It remains to estimate $\#\mathcal{I}_k$. For $\mu \in \mathbb{R}$ such that $I_k(\mu)\cap \mathscr{L}(z)\neq \emptyset$, we get  $I_k(\mu \pm \sqrt{3}/2)\cap \mathscr{L}(z) \neq \emptyset$ and  $I_k(\mu')\cap \mathscr{L}(z) =\emptyset$ for all $\mu' \in (\mu -\sqrt{3}/2,\mu +\sqrt{3}/2) \setminus \lbrace \mu \rbrace$.
 Finally, we have 
$$\mathcal{L}^1\Big(\Pi_k\big([y_T - \tfrac{T}{2} \nu^\perp; y_T + \tfrac{T}{2} \nu^\perp]\big)\Big) = T \big|\langle  \nu,  e^{i\theta} \omega^k \rangle\big|,$$ 
where $\Pi_k$ denotes the  orthogonal projection onto $\R e^{i\theta} (\omega^k)^\perp$. We therefore obtain
\begin{align}\label{eq:mathcalIk}
\#\mathcal{I}_k 
\geq\frac{2T}{\sqrt{3}} \big| \langle  \nu, e^{i\theta}   \omega^k \rangle \big|   -C
=\frac{2T}{\sqrt{3}}  \big|\langle e^{-i\theta} \nu,   \omega^k \rangle\big|   -C.   
\end{align}  
By \eqref{eq: phi-hex-def} and  \eqref{ineq:mathcalIk}--\eqref{eq:mathcalIk} we conclude 
\begin{align*}
\frac{1}{T} E_1\big(X_T, Q^\nu_T(y_T)\big) \geq \frac{2}{\sqrt{3}}\sum\nolimits _{k=1}^3  \big|\langle e^{-i\theta} \nu,   \omega^k \rangle\big| -C/T = \varphi_{\mathrm{hex}}\big(e^{-i\theta} \nu\big) - C/T.
\end{align*}
This along with \eqref{eq: for-first} shows the first inequality. 

\noindent \emph{Step 2: Second inequality.} The goal of this step is to prove
\begin{align}\label{ineq:vacuumphi}
\frac{1}{T} \min\big\{E_1\big(X_T,Q^\nu_T(y_T)\big) \colon \,  X_T = \mathscr{L}(z^\pm) \text{ on } \partial^\pm_1
Q^\nu_T(y_T)\big\}  \le \varphi_{\mathrm{hex}}\big(e^{-i\theta} \nu\big)  +  C/T.
\end{align}
This is achieved by constructing an explicit competitor for the minimization problem: we define $X^+_T$ by
\begin{align}\label{def:Xepsphihexvacuum}
X^+_T = \begin{cases}
\mathscr{L}(z) &\text{in } \{x\colon \langle x-y_T,  \nu \rangle \geq 5\},\\
\emptyset &\text{otherwise,}
\end{cases}
\end{align}
i.e., $X^+_T$ is a (discrete version of a) half space. We directly see that $X^+_T = \mathscr{L}(z)$ on $\partial^+_1 Q^\nu_T(y_T)$ and  $X^+_T=\emptyset$ on $\partial^-_1 Q^\nu_T(y_T)$.  To estimate its energy, we start by observing that for this choice of $X^+_T$ equality holds in \eqref{ineq:mathcalIk}  with $\mathcal{I}_k$ as defined above,  up to an error of order ${\rm O}(1)$. Indeed, if $x \in \mathscr{L}(z)\setminus X^+_T$, then either $x+ \lambda e^{i\theta} \omega^k \notin X^+_T$ for all $\lambda \in \mathbb{N}$ or  $x-\lambda e^{i\theta} \omega^k \notin X^+_T$ for all $\lambda \in \mathbb{N}$. Then, the equalities in \eqref{ineq:mathcalIk} and \eqref{eq:mathcalIk} along with \eqref{eq: phi-hex-def}   yield   
\begin{align}\label{ineq:Xepsphihexvacuum-new}
\frac{1}{T} E_1\big(X^+_T, Q^\nu_T(y_T)\big) \le \frac{2}{\sqrt{3}}\sum\nolimits _{k=1}^3  \big| \langle e^{-i\theta} \nu,   \omega^k \rangle \big| +C/T = \varphi_{\mathrm{hex}}\big(e^{-i\theta} \nu\big) + C/T.
\end{align}
This shows \eqref{ineq:vacuumphi}. For purposes of the proof of (ii) below, we note that construction \eqref{def:Xepsphihexvacuum} with $-\nu$ in place of $\nu$ can be applied to obtain a configuration $X^-_T \subset \mathbb{R}^2$ with ${X}^-_T = \mathscr{L}(z)$ on $\partial^-_1 Q^\nu_T(y_T)$ and  $X^-_T=\emptyset$ on $\partial^+_1 Q^\nu_T(y_T)$ which satisfies \eqref{ineq:Xepsphihexvacuum-new}.

\noindent \emph{Proof of $\mathrm{(ii)}$.} Fix $z^+ =(\theta^+,\tau^+,1) \in \mathcal{Z}$ and $z^- =(\theta^-,\tau^-,1) \in \mathcal{Z}$. We show the first  inequality by an explicit construction.  The second one is obtained with the help of Lemma \ref{lemma:grain-bonds}. 

\noindent \emph{Step 1: First inequality.} We define $X_T = X_T^+ \cup X_T^-$, where
\begin{align*}
X^+_T = \begin{cases}
\mathscr{L}(z^+) &\text{in } \{x\colon \langle x-y_T,\nu \rangle \geq 5\},\\
\emptyset &\text{otherwise.}
\end{cases}, \ \ \  \ \ \ \    X^-_T = \begin{cases}
\mathscr{L}(z^-) &\text{in } \{x\colon \langle x-y_T,  \nu \rangle \leq -5\},\\
\emptyset &\text{otherwise.}
\end{cases}
\end{align*}
Then, $X_T$ clearly satisfies the boundary conditions $X_T = \mathscr{L}(z^\pm)$ on $\partial^\pm_1 Q^\nu_T(y_T)$ and by repeating the reasoning in \eqref{ineq:Xepsphihexvacuum-new} we find  
\begin{align*}
 \frac{1}{T}E_T\big(X_T,Q^\nu_T(y_T)\big) &=   \frac{1}{T}\Big( E_1\big(X_T^+,Q^\nu_T(y_T)\big) + E_1\big(X_T^-,Q^\nu_T(y_T)\big) \Big)  \\& \le \varphi_{\mathrm{hex}}\big(e^{-i\theta^+} \nu\big)+ \varphi_{\mathrm{hex}}\big(e^{-i\theta^-} \nu\big) +C/T.
\end{align*}

\noindent \emph{Step 2: Second inequality.} Consider $X_T \subset \mathbb{R}^2$ satisfying $X_T = \mathscr{L}(z^\pm)$ on $\partial^\pm_1 Q^\nu_T(y_T)$ and
\begin{align*}
 E_1\big(X_T,Q^\nu_T(y_T)\big) = \min \big\{E_1\big(\tilde{X}_T,Q^\nu_T({y}_T)\big) \colon \,  \tilde{X}_T = \mathscr{L}(z^\pm) \text{ on } \partial^\pm_1Q^\nu_T({y}_T)\big\}.  
\end{align*}
By Lemmas \ref{lemma:reduction} and \ref{lemma:grain-bonds} there holds $X_T = X^+_T \cup X^-_T = Y^+_T \dot\cup Y^-_T$ on $Q_T^\nu(y)$, where $Y^\pm_T = \mathscr{L}(z^\pm)$ on $\partial^\pm_1 Q^\nu_T(y_T)$ and 
\begin{align*}
  \big| \sum\nolimits_{x \in \partial Y^\pm_T} \# (\mathcal{N}(x) \cap Y^\pm_T) - 4 \# \partial Y^\pm_T \big| \le 2. 
\end{align*}
Since $\#\mathcal{N}(x) \le 5$ for any $x \in \partial Y^\pm_T$ ($\subset \partial X^\pm_T$), we get 
\begin{align*}
  \frac{1}{2} \sum\nolimits_{x \in Y^\pm_T \cap Q^\nu_T(y_T)} (6 - \#\mathcal{N}(x)) 
  \ge \frac{1}{2} \# \partial Y^\pm_T 
	\ge \frac{1}{4} \sum\nolimits_{x \in \partial Y^\pm_T} \big( 6 - \# (\mathcal{N}(x) \cap Y^\pm_T) \big) - 1/2. 
\end{align*}
So observing that $Y^\pm_T$ is a competitor in (Step 1 of) (i) above and using that $Y^+_T \cap Y^-_T \cap Q^\nu_T(y_T) = \emptyset$, we find   that 
\begin{align*}
  \frac{1}{T} E_1\big(X_T,Q^\nu_T(y_T)\big) 
  &\ge \frac{1}{2T} E_1\big(Y^+_T,Q^\nu_T(y_T)\big) + \frac{1}{2T} E_1\big(Y^-_T,Q^\nu_T(y_T)\big) - 1/T \\ 
	&\ge \frac{1}{2} \varphi_{\mathrm{hex}}\big(e^{-i\theta^+} \nu\big)+ \frac{1}{2} \varphi_{\mathrm{hex}}\big(e^{-i\theta^-} \nu\big) - C/T.  
\end{align*}
This concludes the proof.
\end{proof}

\begin{proof}[Proof of Lemma \ref{lemma: touching}]  Let $\lbrace X_T \rbrace_T$ be an optimal sequence    for $\Phi(z^+,z^-,\nu)$ and denote by $\lbrace y_T\rbrace_T$ the corresponding centers of the cubes.  Due to Lemma \ref{lemma:reduction}, we may without restriction assume that  $X_T = X^+_T \cup X_T^-$,  for sub-configurations $X^\pm_T$ satisfying   $X_T^\pm  \subset \mathscr{L}(z^\pm_T) $, where  $z^\pm_T = (\theta^\pm_T,\tau^\pm_T,1) \to z^\pm =(\theta^\pm,\tau^\pm,1)$ as $T \to +\infty$. Moreover, the sets  $\partial X^\pm$ defined in  \eqref{eq: boundary of Y}  are connected, and there holds $X_T= \mathscr{L}(z^\pm_T)$ on $\partial_1^\pm Q^\nu_T(y_T)$.  In what follows, we fix a subsequence (not relabeled) such that by \eqref{eq: touching lattices}  we have
\begin{align}\label{def:delta}
\varphi_{\rm hex}\big(e^{-i\theta^+}\nu\big) + \varphi_{\rm hex}\big(e^{-i\theta^-}\nu\big)-\lim_{T\to +\infty} \frac{1}{T}E_1\big(X_T,Q^\nu_T(y_T)\big) \ge \eta > 0.
\end{align}
 Our strategy to show \eqref{eq: difference angle} lies in  proving 
\begin{align}\label{eq: main prop to prove}
e^{i(\theta_T^+ - \theta_T^-)} = \frac{v_{T}^+}{v_{T}^-} \ \ \ \text{with} \ \ \ v_{T}^+,v_{T}^- \in \mathscr{L} \ \ \ \text{satisfying} \ \ \ |v_{T}^+| = |v_{T}^-|\leq C_\eta 
\end{align}
for all $T$ sufficiently large, where $C_\eta$ only depends on $\eta$. From this estimate, the statement in  \eqref{eq: difference angle} easily follows. In fact, given \eqref{eq: main prop to prove},  since $\mathscr{L}$ is a discrete set and $\theta^\pm_T \to \theta^\pm$, $e^{i(\theta^+_T - \theta^-_T)} = v_{T}^+/v_{T}^-$ is eventually constant and we find $\theta^+ -\theta^- = \theta_T^+- \theta_T^- \in {\mathcal{G}_{\mathbb{A}}}$ for all $T$ large enough.

Let us come to the proof of  \eqref{eq: main prop to prove}.  Recall by Lemma \ref{lemma:reduction} that $X_T$ is contained in the two  components $X_T^+$ and $X_T^-$. We  further define the set of \emph{touching points}  
\begin{align*}
 \mathcal{T}^+_{T} & = \{x \in X_T^+\colon \,  \exists \, y \in X_T^- \text{ such that } |x-y|=1\}, \\ \mathcal{T}_{T}^- &= \{x \in X_T^-\colon \,  \exists \, y \in X_T^+ \text{ such that } |x-y|=1\}.
\end{align*}
Note that $\mathcal{T}_{T}^\pm \subset \bigcup_{x\in  \partial X_T^\pm} (\lbrace x\rbrace \cup \mathcal{N}(x))$, see definition \eqref{eq: boundary of Y}.  ($\mathcal{T}_{T}^\pm \setminus  \partial X_T^\pm \neq \emptyset$ is possible if $X_T^+ \cap X_T^- \neq \emptyset$.) By \eqref{eq: neighborhood bound} we also observe that 
\begin{align}\label{eq: control}
 \# \mathcal{T}_{T}^+ / 6 \le \# \mathcal{T}_{T}^- \le 6\# \mathcal{T}_{T}^+. 
\end{align}
We start with a brief outline of the proof. Steps 1--4 are devoted to some preliminary estimates: we first show that the cardinality of the sets $ \partial X_T^\pm$  and  $ \mathcal{T}_{T}^\pm$ scales like $T$ by providing a lower bound for $ \mathcal{T}_{T}^\pm$ (Step 1) and an upper bound for $\partial X_T^\pm$ (Step 2). Then we show that, for the majority of  points in $ \mathcal{T}_{T}^\pm$, neighborhoods contain many points of $ \partial X_T^\pm$ (Step 3)  and also elements of   $ \mathcal{T}_{T}^\pm$ (Step 4).  Based on this, we can find quadrilaterals consisting of two points in  $\mathcal{T}^+_{T}$ and two points in $ \mathcal{T}^-_{T}$ where two sides have length $1$ and the other two sides are parallel to lattice vectors of the form $e^{i\theta^+_T}w_T^+$ and    $e^{i\theta^-_T}w_T^-$, respectively, for some $w_T^+, w_T^- \in \mathscr{L}$ with controlled norm. From this, \eqref{eq: main prop to prove} can be derived (Step 5 and Step 6).

\noindent \emph{Step 1: Cardinality of touching points.} We show  $\# \mathcal{T}_{T}^\pm \geq \frac{\eta}{22}T$ for $T$ large enough. By \eqref{eq: neighborhood bound}, \eqref{def:energyneighbourhood}, and the fact that $X^\pm_T= \mathscr{L}(z^\pm_T)$ on $\partial_1^\pm Q^\nu_T(y_T)$, we obtain 
\begin{align*}
E_1\big(X_T,Q^\nu_T(y_T)\big)\geq &\frac{1}{2}\sum_{x \in X_T^+ \cap Q^\nu_T(y_T)} \big(6-\#(\mathcal{N}(x) \cap X_T^+)\big)   + \frac{1}{2} \sum_{x \in X_T^- \cap Q^\nu_T(y_T)} \big(6-\#(\mathcal{N}(x) \cap X_T^-)\big)  \\ & \ \ \ - 3(\#\mathcal{T}_{T}^+ + \#\mathcal{T}_{T}^-)
\end{align*}
and therefore 
\begin{align*}
3(\#\mathcal{T}_{T}^+ + \#\mathcal{T}_{T}^-) 
\ge E_1\big(X_T^+,Q^\nu_T(y_T)\big) + E_1\big(X_T^-,Q^\nu_T(y_T)\big) - E_1\big(X_T,Q^\nu_T(y_T)\big). 
\end{align*}
We note by the definition of $X_T$ that the subconfigurations   $X_T^+$ and $X_T^-$  are competitors for the minimization problems appearing  in Lemma  \ref{lemma:vacuumirrational}(i).  Dividing by $T$ and passing to the $\liminf$ along $T\to +\infty$, by \eqref{def:delta} we therefore conclude 
\begin{align*}
\liminf_{T\to +\infty} \, \frac{1}{T} ( \#\mathcal{T}_{T}^+ + \#\mathcal{T}_{T}^- ) 
\ge \eta/3.
\end{align*}
This yields $\liminf_{T\to +\infty} \frac{1}{T}\#\mathcal{T}_{T}^\pm \ge \frac{\eta}{21}$ by \eqref{eq: control}, and concludes Step 1. 

\noindent \emph{Step 2: A priori bound on the length of the boundaries.} We claim that  for  $T > 0$ large enough the boundaries $\partial X_T^\pm \subset Q^\nu_T(y_T)$ (cf.\ \eqref{eq: boundary of Y}) satisfy 
\begin{align}\label{ineq:lengthboundpartialXplus}
 \# ( \partial X_T^+ \cup \partial X_T^- )  \leq 8T. 
\end{align}
 In fact, by  Lemma \ref{lemma:reduction}(ii) there holds $\#\mathcal{N}(x) \leq 5$ for all $x \in \partial X_T^\pm$ and therefore for $T$ sufficiently large we get by \eqref{def:energyneighbourhood},  \eqref{def:delta}, and  the fact that $\Vert \varphi_{\rm hex} \Vert_{L^\infty(\mathbb{S}^1)}   = 2 $ (see \eqref{eq: phi-hex-def})
\begin{align*}
\# ( \partial X_T^+ \cup \partial X_T^- )  &\leq \sum\nolimits_{x \in X_T \cap Q_T^\nu(y_T)} (6-\#\mathcal{N}(x)) = 2\, E_1\big(X_T,Q^\nu_T(y_T)\big) \\
&   \leq 2T \big ( \varphi_{\rm hex}\big(e^{-i\theta^+}\nu\big) + \varphi_{\rm hex}\big(e^{-i\theta^-}\nu\big) \big)\leq  8T. 
\end{align*}

\noindent \emph{Step 3: Atomic density lower bound for $\partial X_T^\pm$.} We claim that  there exists a universal $0 < c <1$ such that  for all $T> r \ge 1$ we have
\begin{align}\label{ineq:auxiliaryestimateballs}
\#\big(  \partial X_T^\pm \cap B_r(x)  \big) \geq c r \quad  \text{for all $x \in \R^2$ with } \mathrm{dist}(x, \partial X_T^\pm) \le 1.
\end{align}
To prove this estimate we assume without restriction that $T > 3r$. Due  to  Lemma \ref{lemma:reduction}(ii), $\partial X^\pm_T$  is connected and  $\partial X_T^\pm \setminus B_r(x) \neq \emptyset$. Therefore, there has to exist a  simple path  in $\partial X^\pm_T$  that connects  some atom in $\partial X_T^\pm \setminus B_r(x)$ with  an atom in $\overline{B_1(x)}$ and has at least  $c r$ atoms inside $B_r(x)$.

\noindent   \emph{Step 4: Bounded gap between points in $\mathcal{T}_{T}^\pm$.} Given $R>0$, we introduce the set of $R$-\emph{isolated points} by 
\begin{align}\label{eq: bounded gap8}
\mathcal{I}^\pm_{T,R} := \big\{ x\in \mathcal{T}_T^\pm \colon \,   B_R(x) \cap \mathcal{T}_T^\pm \subset \overline{B_2(x)} \big\}.
\end{align}
We claim that there exists  a universal $\bar{c} > 0$ such that for $R=\bar{c}/\eta$ and all $T$ sufficiently large 
\begin{align}\label{eq: bounded gap5}
\#    \mathcal{I}^\pm_{T,R}   \le \#     \mathcal{T}_T^\pm /2.
\end{align}
To see this, note that due to \eqref{ineq:lengthboundpartialXplus}, \eqref{ineq:auxiliaryestimateballs} for $r=R/2$ (use that $ \mathrm{dist}(x, \partial X_T^\pm) \le 1$ for all $x \in \mathcal{T}_T^\pm $) and Step 1 we have 
\begin{align*}
\# \mathcal{I}^\pm_{T,R}  
\leq \frac{2}{cR} \sum\nolimits_{x \in  \mathcal{I}^\pm_{T,R}} \#\big(\partial X_T^\pm \cap B_{R/2}(x) \big) 
\leq \frac{C}{cR} \#\partial X_T^\pm \leq \frac{C}{cR}T 
\leq \frac{C}{c \bar{c}} \# \mathcal{T}_{T}^\pm,
\end{align*}
where $C>0$ denotes a universal constant varying from step to step. Here, in the second step we accounted for possible multiple counting by using that, due to the definition of $\mathcal{I}^\pm_{T,R}$, the intersection $B_{R/2}(x) \cap B_{R/2}(y)$, $x,y \in \mathcal{I}^\pm_{T,R_T}$, can be non-empty only if $|x-y| \le 2$. The assertion follows if $\bar{c}$ is chosen big enough.

\noindent \emph{Step 5: Bounded gap between pairs of points having the same relative position.} We choose two arbitrary lattice vectors $\xi_1,\xi_2$ satisfying $e^{-i\theta^-_T}\xi_1, e^{-i\theta_T^+}\xi_2  \in B_{2R} \cap (\mathscr{L}\setminus \{0\})$ with $R>0$ given by Step 4. Define
\begin{align*}
\mathcal{D}^{\xi_1,\xi_2}_{T}=\big\{(x_1,y_1) \in  \mathcal{T}_{T}^- \times \mathcal{T}_{T}^-\colon \, & \text{there exist $x_2,y_2 \in  \mathcal{T}_{T}^+ $ such that} \\
&    |x_1 - x_2| = 1, \, |y_1 - y_2| = 1   \text{ and }  x_1-y_1=\xi_1,\,  x_2-y_2=\xi_2 \big\}.
\end{align*}
The set consists of pairs $(x_1,y_1)$ in $\mathcal{T}_T^-$ whose difference is $\xi_1$ and which have corresponding neighbors in $\mathcal{T}_T^+$ with difference $\xi_2$. 

 We observe by \eqref{eq: bounded gap8} that for $x_1 \in \mathcal{T}_{T}^- \setminus \mathcal{I}^-_{T,R}$ we find $\xi_1 \in B_R \cap e^{i\theta_T^-} \mathscr{L}$  with $|\xi_1| > 2$  and $y_1 \in \mathcal{T}_{T}^-$ such that $x_1 - y_1 = \xi_1$. We denote the corresponding neighbors in $\mathcal{T}_{T}^+$ by $x_2$ and $y_2$, respectively. Since $x_2,y_2 \in  e^{i\theta_T^+}  (\mathscr{L}+  \tau_T^+  )$ and $|x_1-y_1| = |x_2 - y_2|  =1$, we find $\xi_2 \in B_{2R} \cap e^{i\theta_T^+}\mathscr{L}$ such that $x_2 - y_2 = \xi_2$. Clearly, $\xi_2 \neq 0$ as $|\xi_1|>2$. This discussion along with  \eqref{eq: bounded gap5} implies
\begin{align}\label{eq: bounded gap3}
\frac{1}{2} \#\mathcal{T}_{T}^- \le \#     \big(   \mathcal{T}_T^-\setminus   \mathcal{I}^-_{T,R}\big)  \le  \sum\nolimits_{(\xi_1,\xi_2)} \# \mathcal{D}^{\xi_1,\xi_2}_{T},
\end{align}
where the sum runs over all pairs $(\xi_1,\xi_2)$ with $e^{-i\theta_T^-}\xi_1, e^{-i\theta_T^+}\xi_2  \in B_{2R} \cap (\mathscr{L}\setminus \{0\})$. Choose  $(\zeta_1^T,\zeta_2^T) \in (B_{2R} \cap e^{i\theta_T^-}(\mathscr{L}\setminus \{0\})) \times (B_{2R} \cap e^{i\theta_T^+}(\mathscr{L}\setminus \{0\}))$  such that $\#\mathcal{D}^{\zeta_1^T,\zeta_2^T}_{T} \geq \#\mathcal{D}^{\xi_1,\xi_2}_{T}$ for all  $(\xi_1,\xi_2) \in (B_{2R} \cap e^{i\theta_T^-}(\mathscr{L}\setminus \{0\})) \times (B_{2R} \cap e^{i\theta_T^+}(\mathscr{L}\setminus \{0\}))$.   Then, \eqref{eq: bounded gap3} and the fact that the  number of pairs  $e^{-i\theta_T^-}\xi_1,e^{-i\theta_T^+}\xi_2 \in B_{2R} \cap (\mathscr{L}\setminus \{0\})$ is controlled by $CR^4$ yield
\begin{align}\label{eq: bounded gap2}
\#\mathcal{T}_{T}^- \le CR^4 \,  \# \mathcal{D}^{\zeta^T_1,\zeta^T_2}_{T} 
\end{align} 
for a universal $C > 0$.  We write $\mathcal{D}^{\zeta_1^T,\zeta_2^T}_{T} =\{x_{j}^T,y_{j}^T\}_{j=1}^{M_T}$ for some $M_T \in \mathbb{N}$. We claim that there  is a universal $c' > 0$ such that for $\varrho = c' \eta^{-5}$ 
\begin{align}\label{eq: bounded gap}
\text{there exist $j,k,l \in \lbrace 1,\ldots,M_T\rbrace$ pairwise distinct such that $x_k^T,x_l^T \in B_\varrho(x_j^T)$.}
\end{align}
Assume that, on the contrary, $\varrho$ is such that each $B_{\varrho}(x_k^T) \setminus \lbrace x_k^T \rbrace$  contains at most one point $\{x_{j}^T\}_{j=1}^{M_T}$. Then, it is elementary to see that we can choose $\lbrace \tilde{x}_j^T \rbrace_{j=1}^{\lceil M_T/2 \rceil} \subset \{x_{j}^T\}_{j=1}^{M_T}$ such that $B_{\varrho/2}(\tilde{x}^T_j) \cap  B_{\varrho/2}(\tilde{x}^T_k) = \emptyset$ for $j,k \in \lbrace 1,\ldots, \lceil M_T/2 \rceil\rbrace$, $j\neq k$. This along with    \eqref{ineq:lengthboundpartialXplus}, \eqref{ineq:auxiliaryestimateballs}, and $2\lceil M_T/2 \rceil \ge \#\mathcal{D}^{\zeta_1^T,\zeta_2^T}_{T}$ implies 
\begin{align*}
\#\mathcal{D}^{\zeta_1^T,\zeta_2^T}_{T} &\leq  2\Big\lceil \frac{M_T}{2} \Big\rceil 
\le \frac{4}{c\varrho} \sum_{j=1}^{\lceil M_T/2 \rceil} \hspace{-0.3cm} \#\big(\partial X_T^- \cap B_{\varrho/2}(\tilde{x}_j^T) \big) 
\le \frac{4}{c\varrho} \# \partial X_T^- 
\le \frac{32T}{c\varrho}. 
\end{align*}
From \eqref{eq: bounded gap2}, $\# \mathcal{T}_{T}^- \geq \frac{\eta}{22}T$ (see  Step 1), and the choice $R = \bar{c}/\eta$ in Step 4 we then get $\varrho \le c' \eta^{-5}/2$ for a universal $c' > 0$. The assertion of \eqref{eq: bounded gap} is thus guaranteed for $\varrho = c' \eta^{-5}$. This concludes Step 5.  

\noindent \emph{Step 6: Conclusion.} We denote the three atoms identified in \eqref{eq: bounded gap} by $x_1^1, x_1^2, x_1^3$ (for convenience, we use a different notation and labeling), and denote by $y_1^1,y_1^2,y_1^3$ the corresponding points such that $(x^j_1,y^j_1) \in \mathcal{D}^{\zeta_1^T,\zeta_2^T}_T$ for $j\in\lbrace 1,2,3\rbrace$. In particular, recall that 
\begin{align}\label{eq: rho bound}
|x^1_1-x^2_1|, \ \ 
|x^1_1-x^3_1|, \ \ 
|x^2_1-x^3_1|\leq 2 \varrho.
\end{align} 
By the definition of $\mathcal{D}^{\zeta_1^T,\zeta_2^T}_T$, there exist $(x^1_2,y^1_2),(x^2_2,y^2_2),(x^3_2,y^3_2)$ such that $|x^j_1-x^j_2|=|y^j_1-y^j_2|=1$, $\zeta_1^T=x^j_1-y^j_1$, and $\zeta_2^T=x^j_2-y^j_2$ for $j \in \lbrace 1,2,3 \rbrace$. Now for each $j$, the four points  $\{x^j_1,x^j_2,y^j_2,y^j_1\}$ form a quadrilateral (possibly self-intersecting) with two edges of length one and two edges oriented in $\zeta^T_1$ and $\zeta^T_2$, respectively. Now there are two cases to consider: (a) $\zeta_1^T=\zeta_2^T$ and (b) $\zeta_1^T\neq \zeta_2^T$.

\noindent \emph{Case $\mathrm{(a)}$:} We have that $x^1_1-y^1_1=x^1_2-y^1_2$, where $x^1_1-y^1_1=e^{i\theta_T^-} v_1$ and $x^1_2-y^1_2= e^{i \theta_T^+}v_2$ for $v_1,v_2 \in (\mathscr{L}\setminus \lbrace0 \rbrace) \cap B_{2R}$. Then $e^{i\theta_T^-} v_1 = e^{i\theta_T^+} v_2$ and thus \eqref{eq: main prop to prove} holds for $v^+_T =v_1$ and $v^-_T = v_2$  with $|v^+_T|, |v^-_T| \le 2 R = 2 \bar{c}/\eta$. 

\noindent \emph{Case $\mathrm{(b)}$:} Note that two of the three quadrilaterals  $\{x^j_1,x^j_2,y^j_2,y^j_1\}$, $j \in \lbrace 1,2,3\rbrace$, are necessarily translates of each other. In fact, there are only two different quadrilaterals (up to translation) with fixed order of the sides, prescribed side-length $1$ of two  opposite  edges, and prescribed length and orientation of the other two edges,  see Figure \ref{fig:squares}. 

\begin{figure}[H]
 \includegraphics{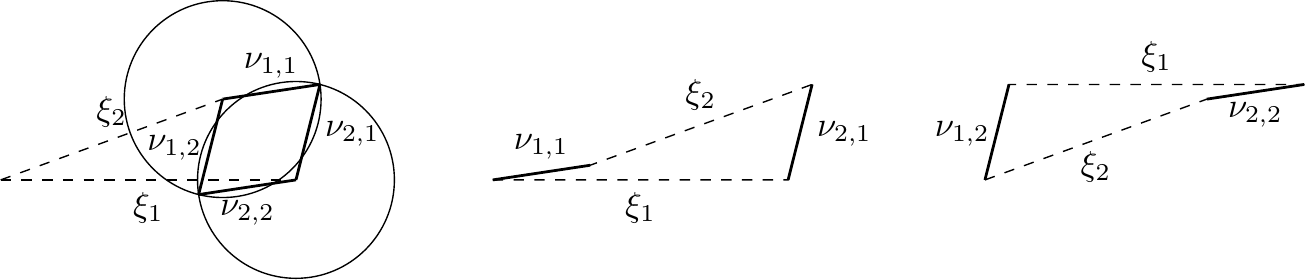}
\caption{The two possible  quadrilaterals  in Step 6,  where $\xi_1,\xi_2$ are given  unlike  vectors and $\nu_{1,1}, \nu_{2,1}, \nu_{1,2}, \nu_{2,2}$ denote the possible sides of length $1$.}
\label{fig:squares}
\end{figure}

Without restriction, assume that the quadrilaterals for $j=1$ and $j=2$ are translates of each other. Then we get $x_1^1 - x_2^1 =x_1^2 - x_2^2$.   We write  $x_1^j = e^{i \theta_T^-} (b_1^j + \tau_T^-)$ and $x_2^j = e^{i \theta^+_T} (b_2^j + \tau^+_T)$  for suitable $b^j_1, b^j_2 \in \mathscr{L}$ for $j \in \lbrace 1,2 \rbrace$. (Note that the lattice vectors depend on $T$ which we do not include in the notation for convenience.) Then $x_1^1 - x_2^1 =x_1^2 - x_2^2$ implies  $e^{i \theta_T^-} (b_1^1 - b_1^2) = e^{i \theta^+_T} (b_2^1 - b_2^2)$.  Since $x_1^1 \ne x_1^2$ we have $b_1^1 - b_1^2 \ne 0$ and thus also $b_2^1 - b_2^2 \ne 0$,  and therefore
\begin{align*}
e^{i (\theta^+_T - \theta^-_T)} = \frac{b_1^1 - b_1^2}{b_2^1 - b_2^2}.
\end{align*}
Due to \eqref{eq: rho bound}, we obtain $|b_1^1 - b_1^2| = |x_1^1 - x_1^2|\leq 2\varrho$ and, since $|b_1^1 - b_1^2| =  |b_2^1 - b_2^2|$, also $|b_2^1 - b_2^2| \leq 2\varrho$. As we clearly also have $b_1^1 - b_1^2, b_2^1 - b_2^2 \in \mathscr{L}$, we derive that  \eqref{eq: main prop to prove} holds for $v_{T}^+:= b_1^1 - b_1^2$ and $v_{T}^- := b_2^1 - b_2^2$  with $|v^+_T|, |v^-_T| \le 2 \varrho = 2c'\eta^{-5}$. As explained below \eqref{eq: main prop to prove},  \eqref{eq: main prop to prove}  implies \eqref{eq: difference angle}, and therefore the proof is concluded. 
\end{proof}

\section{Cell formula  Part II:  Relation of converging and fixed boundary values}\label{section:surfacetension2}

In this final section about cell formulas we show that converging boundary conditions as in the cell formula $\Phi$, see \eqref{eq: Phi def}, can be replaced by fixed boundary values. Moreover, we show Proposition \ref{proposition:existence-original}  and the properties of $\varphi$ stated in  Theorem \ref{prop: properties of varphi}.  We introduce the auxiliary function 
 \begin{align}\label{eq barphi}
\bar{\varphi}(z^+,z^-,\nu) := \liminf_{T \to +\infty}  \frac{1}{T}  \inf  \big\{E_{1}\big(X_T,Q^\nu_{T}(y_T)\big) \colon\, y_T \in \R^2, \,   X_T = \mathscr{L}(z^\pm)  \text{ on } \partial^\pm_{1}Q^\nu_{T}(y_T) \big\} 
\end{align}
for $z^\pm \in \mathcal{Z}$ and $\nu\in \mathbb{S}^1$. The  main goal of this section is to prove the following two statements.

\begin{lemma}\label{lemma: calculation} For each $z^+, z^- \in \mathcal{Z}$ and $\nu \in \mathbb{S}^1$ there holds
\begin{align}\label{eq 5N}
\Phi(z^+,z^-,\nu) = \bar{\varphi}(z^+,z^-,\nu).
\end{align}
Moreover, for $z^\pm = (\theta^\pm,\tau^\pm,1) \in \mathcal{Z}$ with $\{(x,y) \in \mathscr{L}(z^+) \times \mathscr{L}(z^-) \colon \, |x - y| = 1 \} = \emptyset$, we have
$ \bar{\varphi}(z^+,z^-,\nu) 
   = {\varphi}_{\rm hex}\big(e^{-i\theta^+}\nu\big) + {\varphi}_{\rm hex}\big(e^{-i\theta^-}\nu\big).$
\end{lemma}

\begin{proposition}\label{proposition:existence} For every $z^+,z^- \in \mathcal{Z}$,   $\nu \in \mathbb{S}^1$, and every sequence $\lbrace y_T\rbrace_T \in \R^2$ there exists
\begin{align}\label{eq:existenceoflimit}
\bar{\varphi}(z^+,z^-,\nu) = \lim_{T\to +\infty}\frac{1}{T}\min\left\{E_1\big(X_T,Q^\nu_T(y_T)\big)\colon \,   X_T = \mathscr{L}(z^\pm) \text{ \rm on } \partial_1^\pm Q^\nu_T(y_T) \right\}
\end{align}
and is independent of $\lbrace y_T\rbrace_T$. In particular, we get $\varphi  \equiv \bar{\varphi}$, and the statement of  Proposition \ref{proposition:existence-original} holds. 
\end{proposition}

 We point out that Lemma \ref{lemma: calculation}, Proposition \ref{proposition:existence}, and  Lemma \ref{lemma:cutoff} conclude the proof of Proposition \ref{prop, psi}. Subsection \ref{sec: calculation} is devoted to the proof of Lemma \ref{lemma: calculation}. Afterwards, in Subsection \ref{sec: properties}, we show  Proposition \ref{proposition:existence} (which particularly yields Proposition \ref{proposition:existence-original}) and we prove further properties of the density $\varphi$ stated in Theorem \ref{prop: properties of varphi}. Then, all proofs of our main results announced in Subsection \ref{subsection:limitfunctional} are concluded.

\subsection{Converging and fixed boundary values}\label{sec: calculation}
 
This subsection is devoted to the proof of Lemma \ref{lemma: calculation}. By definition it is clear that  $\Phi(z^+,z^-,\nu) \leq \bar{\varphi}(z^+,z^-,\nu)$ for all  $z^+,z^- \in \mathcal{Z}$ and  $\nu \in \mathbb{S}^1$.  To see \eqref{eq 5N}, it therefore suffices  to prove the  opposite inequality
\begin{align}\label{eq: real inequality}
\Phi(z^+,z^-,\nu) \ge \bar{\varphi}(z^+,z^-,\nu).
\end{align}
Moreover, we observe that if $z^+ = \mathbf{0}$ or $z^- = \mathbf{0}$, then Lemma \ref{lemma:vacuumirrational}(i) and the continuity of $\varphi_{\rm hex}$ imply $\Phi(z^+,z^-,\nu) =  \bar{\varphi}(z^+,z^-,\nu) = \varphi_{\mathrm{hex}}(e^{-i\theta} \nu)$, where $\theta$ is the angle corresponding to $z^+$ or $z^-$, respectively. Therefore, it suffices to treat the case $z^\pm = (\theta^\pm,\tau^\pm,1) \in  \mathcal{Z}$. To this end, it is crucial that converging boundary values as in \eqref{eq: Phi def} can be replaced by fixed ones. We split the analysis into two steps by first addressing the rotations and then the translations. 
We start with the rotations.  In view of  Lemma \ref{lemma: touching}, we may without restriction assume that $\theta^+ - \theta^- \in {\mathcal{G}_{\mathbb{A}}}$ since otherwise $\Phi(z^+,z^-,\nu) \ge {\varphi}_{\rm hex}(e^{-i\theta^+}\nu) + {\varphi}_{\rm hex}(e^{-i\theta^-}\nu)$ and \eqref{eq: real inequality} follows from Lemma \ref{lemma:vacuumirrational}(ii). Lemma \ref{lemma: touching} already implies that the difference of rotations $\theta^+_T - \theta^-_T$ is constant in $T$. The next lemma shows that also $\theta_T^+$ and $\theta_T^-$ can be chosen to be   constant.

\begin{lemma}[Fixed rotations]\label{lemma:thetaT-rotation} Consider $z_T^\pm=(\theta_T^\pm,\tau^\pm_T,1) \in \mathcal{Z}$ such that $\theta_T^+-\theta_T^-= \theta^+-\theta^-$ for all $T>0$ for some $\theta^+,\theta^-\in \mathbb{A}$  and $\theta_T^\pm \to \theta^\pm$.  Let $\nu \in \mathbb{S}^1$. Then, there holds 
\begin{align*}
\liminf_{T \to +\infty}  & \frac{1}{T} \inf\big\{E_1\big(X_T,Q^\nu_T(y_T)\big) \colon \,  y_T \in \mathbb{R}^2, \, X_T= \mathscr{L}(z^\pm_T) \text{ on } \partial^\pm_1 Q^\nu_T(y_T) \big\} \\
& \ \ \ \ge \liminf_{T \to +\infty} \frac{1}{T} \inf\big\{E_1\big(X_T,Q^\nu_T(y_T)\big) \colon \,  y_T \in \mathbb{R}^2, \, X_T =  \mathscr{L}(\hat{z}^\pm_T) \text{ on } \partial^\pm_1 Q^\nu_T(y_T) \big\},
\end{align*}
where $\hat{z}^\pm_T := (\theta^\pm,\tau^\pm_T,1)$. 
\end{lemma}
We defer the proof and proceed with the properties of translations.  Again consider $z^\pm=(\theta^\pm,\tau^\pm,1) \in \mathcal{Z}$ with $\theta^+ - \theta^- \in {\mathcal{G}_{\mathbb{A}}}$. Recall by \eqref{eq: good angles} that  there holds $e^{i(\theta^+ - \theta^-)}= \frac{v_1}{v_2}$ for $v_1,v_2 \in \mathscr{L}\cap \mathbb{C}$ with $|v_1|=|v_2|$. We consider the \emph{coincidence site lattice}
\begin{align}\label{eq: a,b}
e^{i\theta^+}\mathscr{L} \cap e^{i\theta^-}\mathscr{L} = \{ ja +kb\colon j,k\in \mathbb{Z} \},
\end{align}
where $a,b \in e^{i\theta^+}\mathscr{L} \cap e^{i\theta^-}\mathscr{L} $ are spanning vectors of minimal length. Then, for later purposes, we define the \emph{fundamental parallelogram} of $e^{i\theta^+}\mathscr{L} \cap e^{i\theta^-}\mathscr{L}$ by
\begin{align}\label{eq: fundamental parallelogram}
{P}_{\theta^+,\theta^-} = \big\{ \lambda_1a+\lambda_2b :  \, 0\leq \lambda_1 < 1, \,  0 \leq \lambda_2 < 1\big\}.
\end{align}
We will use the following uniform closedness property of the set of touching points  between sequences of translates of two perfect lattices. 
\begin{lemma}[Closedness of touching points]\label{lemma:translation} 
Consider $z_n^\pm = (\theta^\pm,\tau^\pm_n,1) \in \mathcal{Z}$ for $n \in \N$ and $z^\pm = (\theta^\pm,\tau^\pm,1) \in \mathcal{Z}$ such that  $\theta^+ -\theta^-\in {\mathcal{G}_{\mathbb{A}}}$ and $\tau^\pm_n \to \tau^\pm$. For $x \in \mathscr{L}(z^+)$, $y \in \mathscr{L}(z^-)$ we set 
\begin{align*}
  x_n^+ = x + e^{i\theta^+}(\tau^+_n - \tau^+) \in \mathscr{L}(z^+_n), \qquad 
	y_n^- = y + e^{i\theta^-}(\tau^-_n - \tau^-) \in \mathscr{L}(z^-_n). 
\end{align*}
Then, there is an $n_0 \in \N$ such that for all $n \ge n_0$ and all $x \in \mathscr{L}(z^+)$, $y \in \mathscr{L}(z^-)$ the following implications are  verified:  
\begin{align*}
{\rm (i)} \ \ |x - y| < 1 \implies |x_n^+ - y^-_n| < 1 
\ \ \ \ \text{and} \ \ \ \ 
{\rm (ii)} \ \ |x - y| > 1 \implies |x_n^+ - y^-_n| > 1. 
\end{align*}
In particular, $|x_n^+ - y^-_n| =1 $ for some $n \ge n_0$ implies $|x-y|=1$. 
\end{lemma}

We again defer the proof and  now proceed with the proof of Lemma \ref{lemma: calculation}.
 
\begin{proof}[Proof of Lemma \ref{lemma: calculation}] Let $z^+,z^- \in \mathcal{Z}$, $\nu \in \mathbb{S}^1$. Recalling the discussion at the beginning of the subsection, we note that it suffices to show inequality \eqref{eq: real inequality}. Moreover, we can assume that $z^\pm= (\theta^\pm,\tau^\pm,1)$  and  that $\theta^+ - \theta^- \in {\mathcal{G}_{\mathbb{A}}}$.

 Let $\lbrace X_T\rbrace_T$ be an optimal sequence for $\Phi$ with corresponding centers $\lbrace y_T\rbrace_T$ of the cubes, i.e., 
\begin{align}\label{eq:recoverypsi}
\liminf_{T\to+\infty} \frac{1}{T}E_1\big(X_T,Q^\nu_T(y_T)\big) = \Phi(z^+,z^-,\nu) <+\infty.
\end{align}
By applying   Lemma \ref{lemma: touching}, we can suppose that $X_T = X_T^+ \cup X_T^-$ with $X_T^\pm \subset \mathscr{L}(z_T^\pm)$ and $X_T = \mathscr{L}(z_T^\pm)$ on $\partial_1^\pm Q_T^\nu(y_T)$, where $z^\pm_T = (\theta_T^\pm,\tau^\pm_T,1) \to z^\pm$. By \eqref{eq: difference angle} and Lemma \ref{lemma:thetaT-rotation} we can also assume that $\theta_T^\pm = \theta^\pm$ for all $T$.  We distinguish the two cases  (a) $\{(x,y) \in \mathscr{L}(z^+) \times \mathscr{L}(z^-) \colon \, |x - y| = 1 \} = \emptyset$ and (b) $\{(x,y) \in \mathscr{L}(z^+) \times \mathscr{L}(z^-) \colon \, |x - y| = 1 \} \neq \emptyset$. 

\noindent \emph{Case $\mathrm{(a)}$:  $\{(x,y) \in \mathscr{L}(z^+) \times \mathscr{L}(z^-) \colon \, |x - y| = 1 \} = \emptyset$.} By Lemma \ref{lemma:translation}  we can assume that $\{(x,y) \in \mathscr{L}(z^+_T) \times \mathscr{L}(z^-_T) \colon \, |x - y| = 1 \} = \emptyset$ for all $T$. Thus,  we get $\mathcal{N}(x) \cap X_T^- = \emptyset$ for all $x \in X_T^+$ and viceversa. Therefore,  by \eqref{def:energyneighbourhood} we obtain
\begin{align*}
 \Phi(z^+,z^-,\nu)  = \liminf_{T\to+\infty} \frac{1}{T}E_1\big(X_T,Q^\nu_T(y_T)\big) = \liminf_{T\to+\infty} \Big(\frac{1}{T}E_1\big(X^+_T,Q^\nu_T(y_T)\big) +\frac{1}{T}E_1\big(X^-_T,Q^\nu_T(y_T)\big) \Big).
\end{align*}
Note that  $X_T^\pm= \mathscr{L}(z^\pm_T)=\mathscr{L}(\theta^\pm,\tau^\pm_T,1)$  on $\partial^\pm_1 Q^\nu_T(y_T)$ and $X_T^\pm =\emptyset$ on $\partial^\mp_1 Q^\nu_T(y_T)$. 
By Lemma  \ref{lemma:vacuumirrational}(i), the energy on each sublattice $X^+_T$ on $X^-_T$ can be estimated separately, and we obtain 
\begin{align}\label{eq: N6}
 \Phi(z^+,z^-,\nu)  \ge {\varphi}_{\rm hex}\big(e^{-i\theta^+}\nu\big) + {\varphi}_{\rm hex}\big(e^{-i\theta^-}\nu\big).
 \end{align}
 Then, Lemma  \ref{lemma:vacuumirrational}(ii) and \eqref{eq barphi} imply $\Phi(z^+,z^-,\nu) \ge  \bar{\varphi}(z^+,z^-,\nu)$  and $\bar{\varphi}(z^+,z^-,\nu) = {\varphi}_{\rm hex}\big(e^{-i\theta^+}\nu\big) + {\varphi}_{\rm hex}\big(e^{-i\theta^-}\nu\big)$.  This concludes the proof of \eqref{eq: real inequality} in  case (a). We also point out that the property stated below \eqref{eq 5N} holds. (In case $\theta^+ - \theta^- \notin {\mathcal{G}_{\mathbb{A}}}$, \eqref{eq: N6} is immediate from \eqref{eq: touching lattices}.) 

\noindent \emph{Case $\mathrm{(b)}$:  $\{(x,y) \in \mathscr{L}(z^+) \times \mathscr{L}(z^-) \colon \, |x - y| = 1 \} \neq \emptyset$.} Our goal is to construct a new competitor $\tilde{X}_T = \tilde{X}_T^+ \cup \tilde{X}_T^-$ such that $\tilde{X}_T^\pm \subset \mathscr{L}(z^\pm)$, $\tilde{X}^\pm_T =  \mathscr{L}(z^\pm)$  on  $\partial^\pm_1 Q_{T+22}^\nu(y_T)$,  and
\begin{align}\label{ineq:lowerboundXtilde}
E_1\big(\tilde{X}_T,Q^\nu_{T+22}(y_T)\big) \leq E_1\big(X_T,Q^\nu_T(y_T)\big)+C.
\end{align}
Once this is established,    by \eqref{eq barphi} and \eqref{eq:recoverypsi} we clearly get
\begin{align*}
\Phi(z^+,z^-,\nu) = \liminf_{T\to+\infty}  \frac{1}{T}  E_1\big(X_T,Q^\nu_T(y_T)\big) \ge \liminf_{T \to +\infty} \frac{1}{T+ 22  }E_1\big(\tilde{X}_T,Q^\nu_{T+ 22  }(y_T)\big) \ge  \bar{\varphi}(z^+,z^-,\nu).
\end{align*}
To construct $\tilde{X}_T$, we first extend $X_T$ to $\hat{X}_T$  by  
  \begin{align}\label{def:Xhat}
  \hat{X}_T = \begin{cases} X_T &\text{on } Q^\nu_{T+10}(y_T)\setminus A_T, \\
  \mathscr{L}(z^\pm_T) &\text{on } \{x\colon \pm\langle\nu,x-y_T\rangle \geq 2\} \cap \big( Q^\nu_{T+34}(y_T)\setminus (Q^\nu_{T+10}(y_T) \cup A_T)\big), \\
     \emptyset &\text{on } A_T  \cup   \big( \mathbb{R}^2 \setminus  Q^\nu_{T+34}(y_T)\big),  
  \end{cases}
  \end{align}
where $A_T=Q^\nu_{10}(y_T+(T/2)\nu^\perp)\cup Q^\nu_{10}(y_T-(T/2)\nu^\perp) \cup (\{|x\colon\langle \nu,x-y_T\rangle|< 2\} \setminus Q^\nu_{T+10}(y_T))$. By definition, we get  $E_1(\hat{X}_T)<+\infty$ since $|x-y| \geq 1$ for all $x,y \in \hat{X}_T$, $x\neq y$. Note that we can write $\hat{X}_T  = \hat{X}_T^+  \dot\cup \hat{X}_T^-$, where $\hat{X}_T^\pm \subset \mathscr{L}(z^\pm_T)$ and $\hat{X}^\pm_T =  \mathscr{L}(z^\pm_T)$  on  $\partial^\pm_1 Q_{T+22}^\nu(y_T)$. We claim that
  \begin{align}\label{ineq:XhattauXT}
  E_1\big(\hat{X}_T,Q^\nu_{T+32}(y_T)\big) \leq E_1\big(X_T,Q^\nu_T(y_T)\big) +C.
\end{align}   
In fact, if there exists $x\in  \hat{X}_T \cap  Q^\nu_{T}(y_T)$  such that $\#(\mathcal{N}(x) \cap \hat{X}_T) < \#(\mathcal{N}(x) \cap X_T)$, then necessarily $x\in (A_T)_1\cap  Q^\nu_{T}(y_T)$. However, $\mathcal{L}^2((A_T\cap  Q^\nu_{T}(y_T))_{2}) \leq C$ and therefore, due to Lemma \ref{lemma:propertiesofE}(v), we get 
\begin{align}\label{ineq:insideQTp5}
\#\big\{x \in \hat{X}_T \cap  Q^\nu_{T}(y_T) \colon \, \#(\mathcal{N}(x) \cap \hat{X}_T) < \#(\mathcal{N}(x) \cap X_T)\big\} \leq C.
\end{align}
In a similar fashion,  if $x\in \hat{X}_T\cap (Q^\nu_{T+32}(y_T) \setminus  Q^\nu_{T}(y_T))$ such that $ \#(\mathcal{N}(x) \cap \hat{X}_T) < 6$, then necessarily $ x\in (A_T)_1 \cap Q^\nu_{T+32}(y_T)$.  Thus, again by Lemma \ref{lemma:propertiesofE}(v), only a bounded number of atoms in $Q^\nu_{T+32}(y_T) \setminus  Q^\nu_{T}(y_T)$ independently of $T$ has less than six neighbors. This along with \eqref{ineq:insideQTp5} and \eqref{def:energyneighbourhood} yields \eqref{ineq:XhattauXT}.

 Let us now define $\tilde{X}_T$.  We recall the notation in \eqref{eq: basic set def} and define $\tilde{X}_T =  \tilde{X}_T^+ \cup \tilde{X}_T^-$ by
\begin{align*}
\tilde{X}^+_T= \big(\hat{X}_T^+ +  e^{i\theta^+}(\tau^+ -\tau^+_T)\big), \ \ \ \ \ \ \tilde{X}^-_T= \big(\hat{X}_T^- +  e^{i\theta^-} (\tau^- -\tau^-_T)\big).
\end{align*}
For convenience, we denote the atoms of $\hat{X}_T$ by $\lbrace x^j_T\rbrace_j$ and the corresponding atoms of $\tilde{X}_T$ by $\lbrace \tilde{x}^j_T\rbrace_j$,  i.e., $\tilde{x}^j_T=x^j_T+e^{i\theta^\pm}(\tau^\pm -\tau^\pm_T)$ if $x^j_T\in \hat{X}_T^\pm$ . By \eqref{def:Xhat}  and the choice of $\hat{X}_T$, it is obvious that $\tilde{X}_T^\pm \subset \mathscr{L}(z^\pm)$ and $\tilde{X}^\pm_T =  \mathscr{L}(z^\pm)$ on $\partial^\pm_1 Q^\nu_{T+ 22}(y_T)$ for $T$ large enough.  Here, the  extension  $\hat{X}_T= \mathscr{L}(z_T^\pm)$ on  $\{x\colon\pm\langle\nu,x-y_T\rangle\geq 2\} \cap (Q^\nu_{T+34}(y_T)\setminus (Q^\nu_{T+10}(y_T) \cup A_T))$ is  crucial in order to ensure that these boundary conditions hold for $\tilde{X}_T$.  (The value $2$ is for definiteness only. Every value less than $5$ works, provided $T$ is sufficiently large.)     To show \eqref{ineq:lowerboundXtilde}, we prove 
\begin{align*}
E_1\big(\tilde{X}_T,Q^\nu_{T+22}(y_T)\big) \leq  E_1\big(\hat{X}_T,Q^\nu_{T+32}(y_T)\big).
\end{align*}
 Then, the result follows from \eqref{ineq:XhattauXT}. To this end, we need to check the following  for large $T$:  
 \begin{align}\label{eq: this is to check}
 {\rm (i)}  \ \  \text{$|x^j_T-x^k_T| =1$ \ $\implies$ \ $|\tilde{x}^j_T - \tilde{x}^k_T|=1$},  \ \ \ \ \text{and} \ \  \ \ {\rm (ii)} \ \ |\tilde{x}^j_T-\tilde{x}^k_T|\ge 1 \text{ for all $j,k$, $j\neq k$}.
 \end{align}
 In fact,  due to \eqref{eq: this is to check}(ii),  $\tilde{X}_T$ is a  configuration with finite energy. Moreover, \eqref{eq: this is to check}(i) shows that $x^k_T \in \mathcal{N}(x^j_T)$ implies $\tilde{x}^k_T \in \mathcal{N}(\tilde{x}^j_T)$, and therefore the energy can only decrease, see \eqref{def:energyneighbourhood}. 

Let us finally check \eqref{eq: this is to check}. If both atoms are in $\hat{X}_T^-$ or $\hat{X}_T^+$, then it is clear by the definition of $\tilde{X}_T$ that $x^j_T-x^k_T = \tilde{x}^j_T-\tilde{x}^k_T$,  which gives (i) and (ii) due to \eqref{eq:recoverypsi} and \eqref{ineq:XhattauXT}. Otherwise, if $x^j_T \in \hat{X}_T^-$ and $x^k_T \in \hat{X}_T^+$ or vice versa, (i) follows from Lemma \ref{lemma:translation}, whereas (ii) follows from Lemma \ref{lemma:translation}(i), \eqref{eq:recoverypsi} and \eqref{ineq:XhattauXT}. 
\end{proof}

To conclude the proof of Lemma \ref{lemma: calculation}, it remains to give the proofs of Lemmas \ref{lemma:thetaT-rotation} and \ref{lemma:translation}.

\begin{proof}[Proof of Lemma \ref{lemma:thetaT-rotation}] 
Let $z_T^\pm=(\theta_T^\pm,\tau_T^\pm,1) \in \mathcal{Z}$ and $\nu \in \mathbb{S}^1$  be given as in the statement. 

\noindent \emph{Step 1: Rotation to boundary conditions with fixed rotation angles.} Choose $\tilde{y}_T \in \mathbb{R}^2$ and $\tilde{X}_T \subset \mathbb{R}^2$ satisfying $\tilde{X}_T = \mathscr{L}(z^\pm_T)$ on $\partial^\pm_1 Q^\nu_T(\tilde{y}_T)$ such that
\begin{align}\label{eq:convergenceXT}
E_1\big(\tilde{X}_T,Q^\nu_T(\tilde{y}_T)\big) \le \inf\left\{E_1\big({X}_T,Q^\nu_T({y}_T)\big) \colon \,  {y}_T \in \mathbb{R}^2,\,  {X}_T = \mathscr{L}(z^\pm_T) \text{ on } \partial^\pm_1 Q^\nu_T(y_T) \right\} + 1/T.
\end{align}
We define  ${X}^{\rm rot}_T :=e^{i(\theta^+-\theta_T^+)}\tilde{X}_T$, $\nu_T:=e^{i(\theta^+-\theta^+_T)}\nu $,  $y_T^{\rm rot} := e^{i(\theta^+-\theta^+_T)}\tilde{y}_T$, and  $\hat{z}^\pm := (\theta^\pm,\tau^\pm_T,1)$. Then,   by   Lemma \ref{lemma:propertiesofE}(i)  and $\theta_T^+ - \theta_T^- = \theta^+-\theta^-$ for all $T$, there holds $X^{\rm rot}_T =\mathscr{L}(\hat{z}^\pm_T)$ on $\partial^\pm_1 Q^{\nu_T}_T(y_T^{\rm rot})$ and
\begin{align*}
E_1\big(\tilde{X}_T,Q^{\nu}_T(\tilde{y}_T)\big) & = E_1\big(X^{\rm rot}_T, Q^{\nu_T}_T(y_T^{\rm rot})) \\&  \geq \inf\big\{E_1\big({X}_T,Q^{\nu_T}_T({y}_T)\big) \colon\, {y}_T \in \mathbb{R}^2, \,  {X}_T=\mathscr{L}(\hat{z}^\pm_T) \text{ on } \partial^\pm_1 Q^{\nu_T}_T(y_T)\big\}  
\end{align*}
for all $T>0$.   Therefore, in view of \eqref{eq:convergenceXT},  to show the statement it suffices to prove
\begin{align}\label{ineq:liminfinfnuT}
\liminf_{T \to +\infty} \frac{1}{T} & \inf\big\{E_1\big(X_T,Q^{\nu_T}_T(y_T)\big) \colon \,  y_T \in \mathbb{R}^2, \, X_T =  \mathscr{L}(\hat{z}^\pm_T) \text{ on } \partial^\pm_1 Q^{\nu_T}_T(y_T) \big\}\notag \\
& \ge \liminf_{T \to +\infty} \frac{1}{T} \inf\big\{E_1\big(X_T,Q^\nu_T(y_T)\big) \colon \,  y_T \in \mathbb{R}^2, \, X_T =  \mathscr{L}(\hat{z}^\pm_T) \text{ on } \partial^\pm_1 Q^\nu_T(y_T) \big\}.
\end{align}
Note that the difference of the two formulas  lies only in the fact that $\nu$ is replaced by $\nu_T$, where $\nu_T \to \nu$ as $T \to +\infty$.

\noindent \emph{Step 2: Proof of \eqref{ineq:liminfinfnuT}.}  Fix $\delta>0$ and let $T >0$ be  sufficiently large such that $|\nu_T-\nu| <\delta$. We choose $\tilde{y}_T \in \mathbb{R}^2$ and $\tilde{X}_T \subset \mathbb{R}^2$ satisfying  $\tilde{X}_T = \mathscr{L}(\hat{z}^\pm_T)$ on $\partial^\pm_1 Q^{\nu_T}_T(\tilde{y}_T)$ such that 
\begin{align}\label{ineq:infT}
E_1\big(\tilde{X}_T,Q^{\nu_T}_T(\tilde{y}_T)\big) \leq  \inf\big\{E_1\big(X_T,Q^{\nu_T}_T(y_T)\big) \colon \,  y_T \in \mathbb{R}^2, \, X_T=\mathscr{L}(\hat{z}^\pm_T) \text{ on } \partial^\pm_1 Q^{\nu_T}_T(y_T) \big\} +\delta.
\end{align}
 Recall \eqref{eq: basic set def} and \eqref{def:line segment}.  We set $T_\delta = (1+2\delta)T$ and  define  
\begin{align*}
A^\delta_T =\left(\left[  \tilde{y}_T - \frac{T}{2}\nu_T^{\bot};  \tilde{y}_T - \frac{T_\delta}{2}\nu^{\bot} \right]  \cup  \left[\tilde{y}_T+\frac{T}{2}\nu^{\bot}_T; \tilde{y}_T+ \frac{T_\delta}{2}\nu^{\bot} \right]\right)_{\kappa T\delta}  \setminus \Big(\partial^+_1 Q^\nu_{T_\delta}(\tilde{y}_T) \cup \partial^-_1 Q^\nu_{T_\delta}(\tilde{y}_T)\Big),
\end{align*}
where  $\kappa > 1$ is chosen sufficiently large later. We define the configuration $\hat{X}_T \subset \mathbb{R}^2$ by
\begin{align}\label{def:XtildeT}
\hat{X}_T = \begin{cases} \tilde{X}_T &\text{in } Q^{\nu_T}_T(\tilde{y}_T),\\
\emptyset &\text{in } A^\delta_T \setminus Q^{\nu_T}_{T}(\tilde{y}_T),\\
\mathscr{L}(\hat{z}^\pm_T) &\text{in } \{x\colon  \pm\langle\nu, (x-\tilde{y}_T)\rangle \geq 5\} \setminus \big( A^\delta_T\cup Q^{\nu_T}_T(\tilde{y}_T)\big).
\end{cases}
\end{align}
Here, $\kappa >1$ is chosen large enough (independently of $T$) such that $|x-y| \geq 1$ for all $x,y \in \hat{X}_T$, $x\neq y$. In principle, $|x-y|<1$ may occur for points $x \in \tilde{X}_T \cap Q^{\nu_T}_T(\tilde{y}_T)$ and $y \in  \mathbb{R}^2 \setminus Q^{\nu_T}_T(\tilde{y}_T)$ if $x \in Q^{\nu_T}_T(\tilde{y}_T) \setminus Q^{\nu_T}_{T-2}(\tilde{y}_T)$,  $\pm \langle\nu_T, (x-\tilde{y}_T)\rangle \geq -5$ and $\pm\langle \nu, (y-\tilde{y}_T) \rangle \leq  - 5 $,  but for $\kappa$ big enough such pairs of points are contained in $A^\delta_T$. 

We note that $\partial^\pm_1 Q^\nu_{T_\delta}(\tilde{y}_T) \cap Q^\nu_T(\tilde{y}_T) = \emptyset$ for $T$ large enough since $\nu_T \to \nu$ as $T\to +\infty$. Thus, by construction we get  $\hat{X}_T = \mathscr{L}(\hat{z}^\pm_T)$ on $\partial^\pm_1 Q^\nu_{T_\delta}(\tilde{y}_T)$ for $T$ sufficiently large. Therefore, we obtain
\begin{align}\label{ineq:inf}
  \inf\big\{E_1\big(X_T,Q^{\nu}_{T_\delta}({y}_T)\big) \colon \,  y_T \in \mathbb{R}^2, \,  X_T=\mathscr{L}(\hat{z}^\pm_T) \text{ on } \partial^\pm_1 Q^\nu_{T_\delta}(y_T)\big\} \leq E_1\big(\hat{X}_T,Q^\nu_{T_\delta}(\tilde{y}_T)\big).
\end{align}
We claim  that 
\begin{align}\label{ineq:energyestimate-new}
E_1\big(\hat{X}_T,Q^\nu_{T_\delta}(\tilde{y}_T)\big)\leq E_1\big(\tilde{X}_T,Q^{\nu_T}_T(\tilde{y}_T)\big) + C \kappa \delta T
\end{align}
for a universal $C>0$. We defer the proof of this estimate to Step 3 below and conclude the proof of \eqref{ineq:liminfinfnuT}. Dividing \eqref{ineq:energyestimate-new} by $T_\delta$ and letting $T \to +\infty$, we derive 
\begin{align*}
\liminf_{T \to +\infty}\frac{1}{T_\delta} E_1\big(\hat{X}_T,Q^\nu_{T_\delta}(\tilde{y}_T)\big)& \leq \liminf_{T \to +\infty}\frac{1}{T}E_1\big(\tilde{X}_T,Q^{\nu_T}_T(\tilde{y}_T)\big) + C \kappa \delta.
\end{align*}
This along with  \eqref{ineq:infT} and  \eqref{ineq:inf}, and the fact that $\delta>0$ was arbitrary shows \eqref{ineq:liminfinfnuT}. 
 It thus remains to prove \eqref{ineq:energyestimate-new}. 

\noindent \emph{Step 3: Proof of \eqref{ineq:energyestimate-new}.}  We divide the proof into the two estimates
\begin{align}\label{ineq:estimate1}
E_1\big(\hat{X}_T,Q^{\nu_T}_{T}(\tilde{y}_T)\big)\leq E_1\big(\tilde{X}_T,Q^{\nu_T}_{T}(\tilde{y}_T)\big) + C\kappa \delta T, 
\end{align}
\begin{align}\label{ineq:estimate2}
E_1\big(\hat{X}_T,Q^{\nu}_{T_\delta}(\tilde{y}_T) \setminus Q^{\nu_T}_{T}(\tilde{y}_T)\big) \leq C \kappa \delta T,
\end{align}
for a universal $C>0$. Clearly, \eqref{ineq:estimate1}--\eqref{ineq:estimate2}  and Lemma \ref{lemma:propertiesofE}(iv)  imply \eqref{ineq:energyestimate-new}. We first prove \eqref{ineq:estimate1}. Recall by \eqref{def:XtildeT}  and the boundary values of $\tilde{X}_T$  that $\hat{X}_T = \tilde{X}_T$  in   $ \overline{Q^{\nu_T}_{T+2}(\tilde{y}_T)} \setminus (A^\delta_T  \setminus Q^{\nu_T}_T(\tilde{y}_T))$. Thus, $x \in Q^{\nu_T}_T(\tilde{y}_T)$ can have less neighbors in  $\hat{X}_T$ than in $\tilde{X}_T$ only if $x \in  (A^\delta_T)_1\cap (Q^{\nu_T}_T(\tilde{y}_T) \setminus Q^{\nu_T}_{T-2}(\tilde{y}_T))$. As $\mathrm{diam}(A^\delta_T) \leq C\kappa\delta T$ and therefore  $\mathcal{L}^2(( (A^\delta_T)_1 \cap (Q^{\nu_T}_T(\tilde{y}_T) \setminus Q^{\nu_T}_{T-2}(\tilde{y}_T)))_{1}) \le C\kappa \delta T$,  this implies by   Lemma \ref{lemma:propertiesofE}(v) that a number of atoms $x \in Q^{\nu_T}_T(\tilde{y}_T)$ bounded by $C\kappa \delta T$ have less neighbors in  $\hat{X}_T$ than in $\tilde{X}_T$. This shows \eqref{ineq:estimate1} by \eqref{def:energyneighbourhood}. To see \eqref{ineq:estimate2}, again due to \eqref{def:XtildeT}, all atoms $x \in \hat{X}_T \cap (Q^{\nu}_{T_\delta}(\tilde{y}_T) \setminus ( Q^{\nu_T}_{T}(\tilde{y}_T) \cup {(A^\delta_T)_1)}$ have six neighbors. Hence, their energy contribution is zero.  As $\hat{X}_T = \emptyset$ in $A^\delta_T  \setminus Q^{\nu_T}_{T}(\tilde{y}_T)  $ and  $\mathcal{L}^2(( {( A^\delta_T)_1}  \setminus A^\delta_T)_1) \le C  \kappa \delta T$,  this implies, as before,  that 
\begin{align*}
\#\left(\hat{X}_T \cap \big(  {( A^\delta_T)_1}  \cap Q^{\nu}_{T_\delta}(\tilde{y}_T)  \big)  \setminus Q^{\nu_T}_{T}(\tilde{y}_T)\right) \leq   C  \mathcal{L}^2(({( A^\delta_T)_1}  \setminus A^\delta_T)_1)\leq C \kappa \delta T.
\end{align*}
Again in view of  \eqref{def:energyneighbourhood}, this implies \eqref{ineq:estimate2} and concludes the proof.  
\end{proof}

\begin{proof}[Proof of Lemma \ref{lemma:translation}] 
Suppose first that $y \in {P}_{\theta^+,\theta^-}$ with ${P}_{\theta^+,\theta^-}$ defined in \eqref{eq: fundamental parallelogram}. Then (i) follows from $x^+_n \to x$, $y^-_n \to y$, and the observation that there are only finitely many pairs $(x,y) \in \mathscr{L}(z^+) \times ({P}_{\theta^+,\theta^-} \cap \mathscr{L}(z^-))$ with $|x-y|<1$. The same argument applies to show that (ii) holds true for all pairs $(x,y) \in (({P}_{\theta^+,\theta^-})_3 \cap \mathscr{L}(z^+)) \times ({P}_{\theta^+,\theta^-} \cap \mathscr{L}(z^-))$ for large $n$. Choosing $n$ so big that also $|\tau^\pm_n-\tau^\pm| < 1$ gives (ii) for all $(x,y) \in \mathscr{L}(z^+) \times ({P}_{\theta^+,\theta^-} \cap \mathscr{L}(z^-))$. 

Now, consider a general $y\in \mathscr{L}(z^-)$. One finds $v \in e^{i\theta^+}\mathscr{L} \cap e^{i\theta^-}\mathscr{L}$ such that $y-v \in {P}_{\theta^+,\theta^-}$. The assertion then follows by applying the special case described above to $x - v$ and $y -v$, and by observing that $(x-v)^+_n = x^+_n-v$ and $(y-v)^-_n = y^-_n-v$. Finally, the implication $|x_n^+ - y^-_n| =1  \Rightarrow |x-y|=1$ follows from (i) and (ii) by contraposition. 
\end{proof}

\subsection{Well  definedness  and properties of the energy density $\varphi$}\label{sec: properties}

This final subsection is devoted to the proofs of Proposition \ref{proposition:existence} and Theorem \ref{prop: properties of varphi}. Our proofs in this subsection follow standard strategies. Due to the discrete character of our model, however, careful constructions are needed. As a preliminary step, we show that in \eqref{eq barphi} the sequence $T\to +\infty$ can be chosen independently of the centers of the cells.

\begin{proposition}\label{proposition:xindependence} For each $z^+,z^- \in \mathcal{Z}$ and $\nu \in \mathbb{S}^1$ there exists a sequence $\{T_j\}_j$ such that  $T_j \to +\infty$ as $j \to +\infty$  and for all $\{y_j\}_j \subset \mathbb{R}^2$ there holds 
\begin{align*}
 \frac{1}{T_j}  \min  \big\{E_{1}\big(X,Q^\nu_{T_j}(y_j)\big) \colon\,  X = \mathscr{L}(z^\pm)  \text{ on } \partial^\pm_{1}Q^\nu_{T_j}(y_j) \big\}\le \bar{\varphi}(z^+,z^-,\nu) + \eta_j,
\end{align*}
where $\lbrace \eta_j\rbrace_j \subset (0,+\infty)$ is a null sequence which  depends on $z^\pm$ and $\nu$, but is independent of $\{y_j\}_j$.   
\end{proposition}
\begin{proof}
First, if $z^+ =\mathbf{0}$ or $z^-=\mathbf{0}$, the statement follows from Lemma \ref{lemma:vacuumirrational}(i) and the definition of $\bar{\varphi}$ in \eqref{eq barphi} for any sequence $\lbrace T_j\rbrace_j$. Now consider $z^\pm = (\theta^\pm, \tau^\pm,1)$. If $\theta^+-\theta^- \notin {\mathcal{G}_{\mathbb{A}}}$, the statement follows from Lemma \ref{lemma:vacuumirrational}(ii), \eqref{eq: a good inquality}, and Lemma \ref{lemma: touching} for any sequence $\lbrace T_j\rbrace_j$. Therefore, it remains to treat the case   $\theta^+-\theta^- \in {\mathcal{G}_{\mathbb{A}}}$.

Consider a sequence $S_j \to +\infty$, $\{x_j\}_j\subset \mathbb{R}^2$, and configurations $\lbrace X_j\rbrace_j \subset \mathbb{R}^2$ satisfying  $X_j = \mathscr{L}(z^\pm)$ on $\partial^\pm_{1}Q^\nu_{S_j}(x_j)$ such that
 \begin{align}\label{eq:convergencephi}
 \bar{\varphi}(z^+,z^-,\nu) = \lim_{j\to +\infty} \frac{1}{S_j}E_{1}\big(X_j,Q^\nu_{S_j}(x_j)\big).
 \end{align}
 By Lemma \ref{lemma:reduction} it is not restrictive to assume that $X_j \subset \mathscr{L}(z^\pm)$ for all $j \in \mathbb{N}$. Our goal is to find a sequence $l_j \to 1$ such that for all $\lbrace y_j\rbrace_j$ there are  configurations $\lbrace \tilde{X}_j \rbrace_j\subset \R^2$ satisfying  $\tilde{X}_j =\mathscr{L}(z^\pm)$ on $\partial^\pm_{1} Q^\nu_{l_jS_j}(y_j)$ such that  
\begin{align}\label{ineq:energyestimate}
E_{1} \big (\tilde{X}_j,Q^\nu_{l_jS_j}(y_j)\big) \leq E_{1}\big(X_j,Q^\nu_{S_j}(x_{j})\big) + C
\end{align}
 for a constant $C>0$ only depending on $z^\pm$ and $\nu$. Once this is achieved, we obtain the statement as  follows: we introduce the sequence $T_j := l_jS_j$, divide \eqref{ineq:energyestimate} by $T_j$, and use \eqref{eq:convergencephi} to get
 \begin{align*}
 \frac{1}{T_j} \min  \big\{E_{1}\big(X,Q^\nu_{T_j}(y_j)\big) \colon\,  X = \mathscr{L}(z^\pm)  \text{ on } \partial^\pm_{1}Q^\nu_{T_j}(y_j) \big\}& \le 
 \frac{1}{T_j}  E_{1} \big (\tilde{X}_j,Q^\nu_{l_jS_j}(y_j)\big)\\
 & \le \frac{1}{l_jS_j} E_{1}\big(X_j,Q^\nu_{S_j}(x_j)\big) + \frac{C}{T_j} 
\\& \le  \bar{\varphi}(z^+,z^-,\nu)  + \eta_j,
 \end{align*}   
 where $\lbrace\eta_j\rbrace_j$ is a null sequence only depending on $z^+,z^-,\nu$, and $\lbrace T_j\rbrace_j$, but independent of the centers $\lbrace y_j\rbrace_j$.
 
Consider any sequence of centers  $\lbrace y_j\rbrace_j$. We now construct $\tilde{X}_j$ and confirm  \eqref{ineq:energyestimate}.  We choose  $\bar{y}_j \in  (\mathscr{L}(z^+) \cap \mathscr{L}(z^-))+x_j$ such that $|y_j-\bar{y}_j| \leq \kappa$, where $\kappa := |a| + |b| + 5$  only depends on the spanning vectors $a,b$ in \eqref{eq: a,b}, but is independent of $j$. Let $l_j := 1 + 4\kappa/S_j$. We set
$$
  A_j = \left(\left[\overline{y}_j-\frac{S_j}{2}\nu^\perp;y_j-\frac{l_jS_j}{2}\nu^\perp  \right]\right)_{4\kappa} \cup \left(\left[\overline{y}_j+\frac{S_j}{2}\nu^\perp;y_j+\frac{l_jS_j}{2}\nu^\perp  \right]\right)_{4\kappa}.
$$
Note that $\partial^\pm_{1}Q^\nu_{l_jS_j}(y_j) \cap Q^\nu_{S_j}(\bar{y}_j) = \emptyset$ since $S_jl_j - S_j = 4\kappa$,  $|y_j-\bar{y}_j| \leq \kappa$, and $\kappa \ge 5$. We define $\tilde{X}_j \subset \mathbb{R}^2$ by 
\begin{align*}
\tilde{X}_j = \begin{cases} X_j + \bar{y}_j-x_j &\text{in } Q^\nu_{S_j}(\bar{y}_j) \setminus A_j,\\
\emptyset &\text{in } A_j \setminus \big( \partial^+_{1}Q^\nu_{l_jS_j}(y_j) \cup  \partial^-_{1}Q^\nu_{l_jS_j}(y_j)\big),\\
\mathscr{L}(z^\pm) &\text{in } \big( \{\pm\langle \nu,  x- y_j \rangle \geq 5\} \setminus \big(A_j \cup Q^\nu_{S_j}( \bar{y}_j)\big)\big) \cup  \partial^\pm_{1}  Q^\nu_{l_jS_j}(y_j). 
\end{cases}
\end{align*} 
By definition, $\tilde{X}_j $ attains the correct boundary conditions, and therefore it remains to check \eqref{ineq:energyestimate}. First, as $x_j - \bar{y}_j \in \mathscr{L}(z^+) \cap \mathscr{L}(z^-)$ and  $X_j = \mathscr{L}(z^\pm)$ on $\partial^\pm_{1}Q^\nu_{S_j}(x_j)$, we observe that $\tilde{X}_j = \mathscr{L}(z^\pm)$ on $(\partial^\pm_{1} Q^\nu_{S_j}(\bar{y}_j) \cap Q^\nu_{S_j}(\bar{y}_j)) \setminus A_j$. This along with the definition of $A_j$ implies $|x-y|\geq 1$ for all $ x, y \in \tilde{X}_j$,  $x\neq y$, and thus $E_{1} \big (\tilde{X}_j,Q^\nu_{l_jS_j}(y_j)\big) < +\infty$. Moreover,   by Lemma \ref{lemma:propertiesofE}(i) we obtain
\begin{align}\label{ineq:Xjtilde1}
E_{1}\big(\tilde{X}_j,Q^\nu_{S_j}(\bar{y}_j)\big) \leq E_{1}\big(X_j,Q^\nu_{S_j}(x_j)\big) + C.
\end{align}
Here, the extra term $C>0$  is due the fact that we take into account the interactions of points $x \in \tilde{X}_j \cap Q^\nu_{S_j}(\bar{y}_j) \cap (A_j)_1$ . Since  $\mathcal{L}^2((A_j)_2) \le C_\kappa$ for $C_\kappa$ depending only $\kappa$ and $E_1(\tilde{X}_j) <+\infty$, by Lemma \ref{lemma:propertiesofE}(v), the cardinality of these points can be controlled by $C_\kappa$. Then, by \eqref{def:energyneighbourhood} we indeed get  \eqref{ineq:Xjtilde1}. Additionally, there holds 
\begin{align}\label{ineq:Xjtilde2}
E_1\Big(\tilde{X}_j,Q^\nu_{l_jS_j}(y_j) \setminus Q^\nu_{S_j}(\bar{y}_j)\Big) \leq C,
\end{align}
where $C$ again only depends on $\kappa$. In fact, all points $x \in \tilde{X}_j \cap (Q^\nu_{l_jS_j}(y_j) \setminus Q^\nu_{S_j}(\bar{y}_j))$ with $\mathrm{dist}(x,A_j) > 1$  satisfy $\#\mathcal{N}(x) =6$ and therefore they do not contribute to the energy. Again due to Lemma \ref{lemma:propertiesofE}(v), the cardinality of $x \in \tilde{X}_j$ with $\mathrm{dist}(x,A_j) \le 1$ can be estimated by $C_\kappa$. This gives \eqref{ineq:Xjtilde2}. Now, \eqref{ineq:Xjtilde1}--\eqref{ineq:Xjtilde2} along with Lemma \ref{lemma:propertiesofE}(iv)  imply \eqref{ineq:energyestimate}. This concludes the proof. 
\end{proof}

\begin{proof}[Proof of Proposition \ref{proposition:existence}]
We first show that, once \eqref{eq:existenceoflimit} has been established, the result in Proposition \ref{proposition:existence-original} follows. Indeed, given $x_0 \in \R^2$ and $\rho>0$, estimate  \eqref{def:varphi} readily follows from \eqref{eq:existenceoflimit} for the sequence of centers  $y_T = (T/\rho)x_0$ and a scaling argument, see Proposition \ref{lemma:propertiesofE}(ii) for $\eps =\rho/T$, $\lambda=T/\rho$, and $A = Q^\nu_\rho(x_0)$.

It remains to prove \eqref{eq:existenceoflimit}. Let $z^\pm \in \mathcal{Z}$, $\nu \in \mathbb{S}^1$, and a sequence $\lbrace y_T\rbrace_T \subset \R^2$ be given. In view of the definition of $\bar{\varphi}$, see \eqref{eq barphi}, it suffices to show
\begin{align}\label{eq: mainEbound2}
\limsup_{T \to +\infty}\frac{1}{T}\min\big\{E_1\big(X_T,Q^\nu_T(y_T)\big)\colon \,  X_T = \mathscr{L}(z^\pm) \text{ on } \partial_1^\pm Q^\nu_T(y_T) \big\}\le\bar{\varphi}(z^+,z^-,\nu).
 \end{align}

\noindent \emph{Step 1: Comparison via construction.}
Consider $1 \ll  S \ll T$. Without restriction, we can assume that $S \in \lbrace T_j\rbrace_j$, where $\lbrace T_j\rbrace_j$ is the sequence identified in  Proposition  \ref{proposition:xindependence}. For simplicity, if $S = T_j$, we will write $\eta_S$ instead of $\eta_{T_j}$ for the null sequence given by  Proposition  \ref{proposition:xindependence}. Define $N_{S,T}:= \lfloor T/S\rfloor$. For $j \in \lbrace 1,\ldots, N_{S,T} \rbrace$  we set $x_j = y_T +  (-T/2- S/2 + j S)\nu^\perp$. We choose $X_j \subset \mathbb{R}^2$ such that $X_j = \mathscr{L}(z^\pm)$    on   $\partial_1^\pm Q^\nu_S(x_j)$ and 
\begin{align}\label{eq:EnergyQi}
E_1\big(X_j,Q^\nu_S(x_j)\big)&=\min\left\{E_1\big(X,Q^\nu_S(x_j)\big) \colon \,  X = \mathscr{L}(z^\pm) \text{ on } \partial_1^\pm Q^\nu_S(x_j) \right\}  \le S\big(\bar{\varphi}(z^+,z^-,\nu) + \eta_S\big),
\end{align}
where the inequality follows from Proposition  \ref{proposition:xindependence}.
For $j=1,\ldots, N_{S,T}$, we introduce the  set $A_j = Q^\nu_{10}(x_j+(S/2)\nu^\perp) \cup Q^\nu_{10}(x_j-(S/2)\nu^\perp)$ and let $X_T$ be defined by 
\begin{align*}
X_T = \begin{cases} X_j &\text{in }  Q^\nu_S(x_j) \setminus A_j , \ j\in \lbrace 1,\ldots,N_{S,T}\rbrace, \\
\emptyset &\displaystyle\text{in } \{x\colon\, |\langle\nu, x - y_T \rangle | < 5\} \setminus Q^*,\\
\mathscr{L}(z^\pm) &\displaystyle \text{in }\{x\colon\,\pm\langle\nu , x - y_T \rangle  \geq 5\}\setminus Q^*,
\end{cases}
\end{align*}
where for brevity we have set $Q^* := \bigcup\nolimits_{j=1}^{N_{S,T}} (Q^\nu_S(x_j)\setminus A_j)$. Note that $X_T = \mathscr{L}(z^\pm)$ on $\partial^\pm_1 Q^\nu_T(y_T)$. For an illustration of the construction, we refer to Figure \ref{fig: lim-ex}. We will show that
\begin{align}\label{eq: mainEbound}
 E_1\big(X_T,Q^\nu_T(y_T)\big)\le \lfloor T/S\rfloor \, S\big(\bar{\varphi}(z^+,z^-,\nu) + \eta_S\big) + CT/S + CS
 \end{align}
 for a universal constant $C>0$.  Once this is achieved, we divide by $T$, take first the $\limsup$ as $T\to+\infty$, and then the limit as $S\to +\infty$ (with $S$ chosen from the sequence $\lbrace T_j\rbrace_j$ given by Proposition  \ref{proposition:xindependence}). As $\eta_S\to 0$, this yields  \eqref{eq: mainEbound2} and thus the statement of the proposition.


\begin{figure}[htp]
 \includegraphics{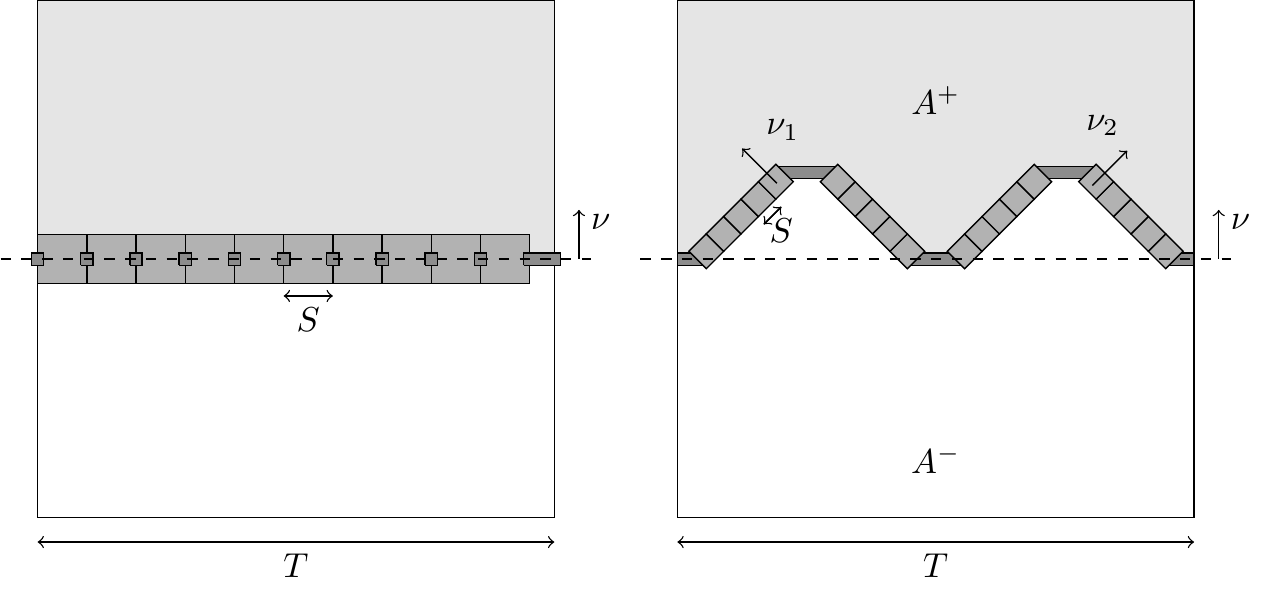}
\caption{Illustration of the construction for the existence of the limit on the left as well as the convexity in the third variable on the right. On the white region $X_T= \mathscr{L}(z^-)$, on the light gray region $X_T= \mathscr{L}(z^+)$, and on the dark gray region $X_T=\emptyset$. The dark gray cubes, that are cut out in order to ensure that $X_T$ has finite energy, are illustrated on the left, but they are also present in the construction on the right. In the  gray  cubes, we set $X_T$ equal to the minimizer with boundary conditions  $\mathscr{L}(z^\pm)$. For illustration purposes, we suppose that $w = 0$ in \eqref{eq: w relevance}.}\label{fig: lim-ex}
\end{figure}

\noindent \emph{Step 2: Proof of \eqref{eq: mainEbound}.}  It remains to prove \eqref{eq: mainEbound}. First, by construction,  the definition of $A_j$,  and the boundary values of the configurations $X_j$, we get $|x-y|\geq 1$ for all $ x, y \in X_T$,  $x\neq y$, and   therefore $E(X_T) <+\infty$.  By Lemma \ref{lemma:propertiesofE}(iv) and \eqref{eq:EnergyQi} there holds 
\begin{align}\label{ineq:XT}
 E_1\big(X_T,Q^\nu_T(y_T)\big) &= \sum\nolimits_{j=1}^{N_{S,T}} E\big(X_T,Q^\nu_S(x_j)\big) + E\Big(X_T, Q^\nu_T(y_T) \setminus \bigcup\nolimits_{j=1}^{N_{S,T}} Q^\nu_S(x_j)\Big) \\& \le  \lfloor T/S\rfloor \, \Big( S\big(\bar{\varphi}(z^+,z^-,\nu) + \eta_S\big) + C\Big) +  E\Big(X_T, Q^\nu_T(y_T) \setminus \bigcup\nolimits_{j=1}^{N_{S,T}} Q^\nu_S(x_j)\Big). \notag
\end{align}
Here, the addend $C$ in the brackets is due to the fact that there may be $x \in X_T \cap Q^\nu_S(x_j)$ with more neighbors in $X_j$  than in  $X_T$. This, however, can only occur for atoms in $x \in Q^\nu_S(x_j)$ such that  $x\in(\partial Q^\nu_S(x_j))_6 \cap (\{y\colon\,\langle  y-x_j, \nu\rangle=0\})_6.$  Since  $E(X_T) <+\infty$, we can apply Lemma \ref{lemma:propertiesofE}(v) and get that their cardinality is controlled by some universal constant $C$.  

 It remains to estimate the energy outside the union of the smaller cubes. We claim that
\begin{align}\label{ineq:Eoutside}
E\Big(X_T, Q^\nu_T (y_T) \setminus \bigcup\nolimits_{j=1}^{N_{S,T}} Q^\nu_S(x_j)\Big)\leq C S.
\end{align}
To see this, note that an atom  $x \in X_T \cap (Q^\nu_T(y_T) \setminus \bigcup_{j=1}^{N_{S,T}} Q^\nu_S(x_j)) $ can contribute to the energy only if  $|\langle  x - y_T, \nu \rangle| \leq 6$.  Since $E(X_T) <+\infty$, applying Lemma \ref{lemma:propertiesofE}(v), we obtain 
\begin{align*}
\#\Big\{x \in X_T \cap \Big( Q^\nu_T(y_T) \setminus \bigcup\nolimits_{j=1}^{N_{S,T}} Q^\nu_S(x_j)\Big)\colon \, |\langle  x - y_T, \nu \rangle| \leq 6 \Big\} \leq C\left(T - S \left\lfloor T/S \right\rfloor\right) \leq CS,
\end{align*}
where $T - S \lfloor T/S \rfloor$ controls the length of the  rightmost dark gray region in  the left part of  Figure \ref{fig: lim-ex}. In view of \eqref{def:energyneighbourhood}, this implies \eqref{ineq:Eoutside}. Combining \eqref{ineq:XT} and  \eqref{ineq:Eoutside} we obtain \eqref{eq: mainEbound}, which concludes the proof. 
\end{proof}

As a final preparation for the proof of Theorem \ref{prop: properties of varphi}, we characterize the translations of  lattices   with touching points. To this end, we introduce the following notation: for given $\theta = \theta^+ -\theta^- \in {\mathcal{G}_{\mathbb{A}}}$, we say $e^{i\theta^+} \tau^+ - e^{i\theta^-} \tau^-$ is a \emph{good translation} and write $e^{i\theta^+} \tau^+ - e^{i\theta^-} \tau^- \in \mathcal{G}_{\mathbb{T}}(\theta)$, whenever $(\tau^+, \tau^-) \in \mathbb{T}^2$ are such that there exist $x \in \mathscr{L}(\theta^+,\tau^+,1)$ and $y \in \mathscr{L}(\theta^-,\tau^-,1)$ with $|x - y| = 1$. (By rotational invariance this does indeed only depend on the difference $\theta = \theta^+ -\theta^-$.) 

\begin{lemma}[Properties of translations]\label{lemma:translationproperties} 
Suppose that $\theta = \theta^+ -\theta^-\in {\mathcal{G}_{\mathbb{A}}}$. Then $\mathcal{G}_{\mathbb{T}}(\theta)$ is contained in a finite union (of arcs) of spheres of radius $1$, namely 
$$ \mathcal{G}_{\mathbb{T}}(\theta) 
   \subset \bigcup\nolimits_{x', y'} \partial B_1(y' - x'), $$ 
where the union is taken over the all $x' \in e^{i\theta^+} \mathscr{L} \cap ({P}_{\theta^+,\theta^-})_5$ and $y' \in e^{i\theta^-} \mathscr{L} \cap {P}_{\theta^+,\theta^-}$, where ${P}_{\theta^+,\theta^-}$ is the fundamental parallelogram defined in \eqref{eq: fundamental parallelogram}. (Recall also notation \eqref{eq: basic set def}). 
\end{lemma}

\begin{proof} 
Consider  $x \in \mathscr{L}(\theta^+,\tau^+,1)$ and $y \in \mathscr{L}(\theta^-,\tau^-,1)$ with $|x - y| = 1$. We find a shifting  vector   $v \in e^{i\theta^+} \mathscr{L} \cap e^{i\theta^-} \mathscr{L}$  such that $y' := y - v  - e^{i\theta^-} \tau^- \in e^{i\theta^-} \mathscr{L} \cap {P}_{\theta^+,\theta^-}$. By defining $x' := x - v - e^{i\theta^+} \tau^+ \in e^{i\theta^+} \mathscr{L}$ we clearly get  
$$ 1 =  |y-x|  =   \big| \big(y'  -x'\big)  - \big(e^{i\theta^+} \tau^+ - e^{i\theta^-} \tau^-\big)\big|.$$
The latter identity along with $|\tau^\pm| \le \sqrt{3} < 2$ (see \eqref{eq: tautautau}) yields  $x' \in  e^{i\theta^+} \mathscr{L} \cap ({P}_{\theta^+,\theta^-})_5$ as well as $e^{i\theta^+} \tau^+  -   e^{i\theta^-} \tau^- \in \partial B_1(y' - x')$. 
\end{proof}

We close this subsection  with the proof of Theorem \ref{prop: properties of varphi}.

\begin{proof}[Proof of Theorem \ref{prop: properties of varphi}]
\emph{Proof of {\rm (i),(ii)}.} The proof of (i) follows from the definition of $\varphi$ and Lemma \ref{lemma:vacuumirrational}(i). For (ii), we use Lemma \ref{lemma:vacuumirrational}(ii) to obtain the inequality 
\begin{align}\label{eq: N9}
 \frac{1}{2} \varphi_{\mathrm{hex}}\big(e^{-i\theta^+} \nu\big)+ \frac{1}{2}\varphi_{\mathrm{hex}}\big(e^{-i\theta^-} \nu\big) \leq \varphi(z^+,z^-,\nu) \le  \varphi_{\mathrm{hex}}\big(e^{-i\theta^+} \nu\big)+ \varphi_{\mathrm{hex}}\big(e^{-i\theta^-} \nu\big)
\end{align}
for all  $(z^+,z^-) \in (\mathcal{Z}\setminus \lbrace \mathbf{0}\rbrace)^2$, $z^+ \neq z^-$.  By Lemma \ref{lemma: touching}, $\varphi = \Phi$ (see Lemma \ref{lemma: calculation} and Proposition \ref{proposition:existence}), Lemma \ref{lemma: calculation}, and the definition of $\mathcal{G}_{\mathbb{T}}(\theta^+-\theta^-)$, the inequality in \eqref{eq: N9} can be strict only if $\theta^+-\theta^- \in {\mathcal{G}_{\mathbb{A}}}$ and $e^{i\theta^+}\tau^+-e^{i\theta^-}\tau^- \in {\mathcal{G}_{\mathbb{T}}}(\theta^+-\theta^-)$. Clearly, ${\mathcal{G}_{\mathbb{A}}} \subset \mathbb{A}$ is countable, see \eqref{eq: good angles}, and ${\mathcal{G}_{\mathbb{T}}}(\theta^+-\theta^-) \subset \mathbb{R}^2$ is contained in a finite union of spheres by Lemma \ref{lemma:translationproperties}.

\noindent \emph{Proof of {\rm (iii)}.}  Let $\nu_1,\nu_2 \in \mathbb{S}^1$, $\lambda \in (0,1)$. Our goal is to prove
\begin{align}\label{eq:convexity}
\varphi(z^+,z^-,\lambda\nu_1 +(1-\lambda)\nu_2) \leq  \lambda\varphi(z^+,z^-,\nu_1)+(1-\lambda)\varphi(z^+,z^-,\nu_2).
\end{align}
Assume that $\lambda\nu_1 +(1-\lambda)\nu_2\neq 0$ (otherwise the statement is trivial) and define $\nu = \frac{\lambda\nu_1 +(1-\lambda)\nu_2}{|\lambda\nu_1 +(1-\lambda)\nu_2|} \in \mathbb{S}^1$. By the  positive  $1$-homogeneity of $\varphi$, \eqref{eq:convexity} is equivalent to
\begin{align}\label{eq:convexitynu}
\varphi(z^+,z^-,\nu) \leq \lambda_1\varphi(z^+,z^-,\nu_1)+\lambda_2\varphi(z^+,z^-,\nu_2),
\end{align}
where $\lambda_1 =  \frac{\lambda}{|\lambda\nu_1 +(1-\lambda)\nu_2|}, \lambda_2 = \frac{1-\lambda}{|\lambda\nu_1 +(1-\lambda)\nu_2|} >0$.
In the following, we will prove \eqref{eq:convexitynu}.

\noindent \emph{Step 1: Convexity via construction.} We construct competitors for the problem $\varphi(z^+,z^-,\nu)$, and refer to Figure \ref{fig: lim-ex} for an illustration. Fix $n \in \mathbb{N}$ such that $\lambda_1,\lambda_2 \leq n/2$.  Let $1 \ll S \ll T$. As before,  we  assume that $S \in \lbrace T_j\rbrace_j$, where $\lbrace T_j\rbrace_j$ is the sequence identified in  Proposition  \ref{proposition:xindependence}. For simplicity, if $S = T_j$, we will write $\eta_S$ instead of $\eta_{T_j}$ for the null sequence given by  Proposition~\ref{proposition:xindependence}.

Define  $ N_j(S,T) := \left\lfloor \lambda_j(T-(10n+5)S)/(nS) \right\rfloor$ for $j \in \lbrace 1,2\rbrace$.  In the following, the indices $i$, $j$, and $k$ are always chosen from $j \in \lbrace 1,2\rbrace$,  $i\in\lbrace 0,\ldots,N_j(S,T)\rbrace$,  and $k\in \lbrace 0,\ldots,n-1\rbrace$ without further notice.  As usual, the orthonormal vectors to $\nu,\nu_1,\nu_2$ obtained by clockwise rotation about $\pi/2$ are denoted by $\nu^\bot,\nu_1^\bot,\nu_2^\bot$, respectively. From $\nu = \lambda_1 \nu_1 + \lambda_2 \nu_2$ and the definition of $N_j(S,T)$ we get 
\begin{align}\label{eq: w relevance}
 N_1(S,T) \nu_1^\bot + N_2(S,T) \nu_2^\bot 
   = M \nu^\bot - w, 
   \end{align}
where $M = (T-(10n+5)S)/(nS)$ and $w = \alpha_1 \nu_1^\bot + \alpha_2 \nu_2^\bot$ for suitable $0 \le \alpha_1, \alpha_2 < 1$, in particular, $|w| \le 2$. 
We set
\begin{align*}
x_i^{1,k}  =\big(-T/2+  5S + S(M+10)k  \big) \, \nu^\perp +  i\,S\nu_1^\perp, \ \ \ \
x_i^{2,k}  =x_{N_1(S,T)}^{1,k} +  5S\nu^\perp+  i\,S\nu_2^\perp,
\end{align*}
and let $X_i^{j,k} \subset \mathbb{R}^2$ be defined as a minimizer of the problem
\begin{align}\label{eq:EnergyQijk}
\begin{split}
\min \left\{E_1\big(X,Q^{\nu_j}_S(x_i^{j,k})\big) \colon  \,  X = \mathscr{L}(z^\pm) \text{ on } \partial_1^\pm Q^{\nu_j}_S(x_i^{j,k}) \right\}.
\end{split}
\end{align}
We recall notation \eqref{eq: basic set def}--\eqref{def:line segment} and define
\begin{align*}
U = \big([- \tfrac{T}{2}\nu^\perp; x_0^{1,0}]\big)_{\kappa} \cup \bigcup_{k=0}^{n-1} \big([x_{N_1(S,T)}^{1,k}; x_0^{2,k}]\big)_{\kappa} \cup \bigcup_{k=0}^{n-2} \big([x_{N_2(S,T)}^{2,k}; x_0^{1,k+1}]\big)_{\kappa} \cup  \big([ x_{N_2(S,T)}^{2,n-1}; \tfrac{T}{2}\nu^\perp]\big)_{\kappa},
\end{align*}
where $\kappa>1$ is chosen later.
Note that $U$ consists of $2n+1$ tubular neighborhoods of segments whose  maximal  length is bounded by $CS$.  (Apart from the segment $[ x_{N_2(S,T)}^{2,n-1}; \tfrac{T}{2}\nu^\perp]$, this follows directly from the choice of the points $x_i^{j,k}$ and \eqref{eq: w relevance}. For $[ x_{N_2(S,T)}^{2,n-1}; \tfrac{T}{2}\nu^\perp]$, it  follows from $x_{N_2(S,T)}^{2,n-1} = (-T/2 + S(M+10)n)\, \nu^\perp - S w = (T/2 - 5S)\, \nu^\perp - Sw$, where $|w| \le 2$.)  We also observe that  $Q^\nu_T \setminus (\bigcup_{i,j,k} Q^{\nu_j}_S(x_i^{j,k}) \cup U)$ consists of two connected components. The connected component intersecting $\partial^+_1 Q^\nu_T$ is denoted  by $A^+$ and the other one is denoted by $A^-$. Note that the cubes $Q^{\nu_j}_S(x_i^{j,k})$ do not intersect $\partial^\pm_1 Q^\nu_T$. We introduce the sets  $A_i^{j,k} = Q^{\nu_j}_{10}( x_i^{j,k} + (S/2) \nu^\perp_j )\cup Q^{\nu_j}_{10}( x_i^{j,k} - (S/2) \nu^\perp_j )$ and let $X_T$ be defined by 
\begin{align}\label{def:XTconvexity}
X_T = \begin{cases} X_i^{j,k} &\text{in } Q^{\nu_j}_S(x_i^{j,k})\setminus A_i^{j,k}, \\
\emptyset &\text{in }  \left(U \setminus\left( \bigcup_{i,j,k} Q^{\nu_j}_S(x_i^{j,k}) \cup \partial^-_1 Q^{\nu}_T\cup \partial^+_1 Q^\nu_T\right)\right)\cup \bigcup_{i,j,k}A_i^{j,k} ,\\
\mathscr{L}(z^\pm) &\text{in } A^\pm \cup \partial^\pm_1 Q^\nu_T.
\end{cases}
\end{align}
For an illustration of the sets and the configuration $X_T$ we refer to Figure \ref{fig: lim-ex}. Clearly, we have $X_T = \mathscr{L}(z^\pm)$ on $\partial^\pm_1 Q^\nu_T$.

\noindent \emph{Step 2: Energy estimate on $X_T$.} We now estimate the energy of $X_T$. First, due to the boundary conditions  $X_i^{j,k} = \mathscr{L}(z^\pm)$ on  $\partial_1^\pm Q^{\nu_j}_S(x_i^{j,k})$, one can check that for $\kappa$ big enough there holds $|x-y| \ge 1$ for all $x,y \in X_T$, $x\neq y$ and therefore $E_1(X_T)<+\infty$. We now prove the following two sub-estimates
\begin{align}\label{ineq:estimate1conv}
E_1\Big(X_T, \big( A^+ \cup A^- \cup \partial^+_1 Q^\nu_T\cup\partial^-_1 Q^\nu_T \big)  \cap Q^\nu_T \Big)\leq CnS
\end{align}
and
\begin{align}\label{ineq:estimate2conv}
E_1\left(X_T,\bigcup\nolimits_{i,j,k} Q^\nu_S(x_i^{j,k}) \cup \big(U \setminus (\partial^-_1 Q^{\nu}_T\cup \partial^+_1 Q^\nu_T)   \big)\right) \leq & \sum\nolimits_{j=1}^2\frac{\lambda_jT}{S} \left(S\big({\varphi}(z^+,z^-, \nu_j) + \eta_S\big) +C\right),
\end{align}
where $\lbrace \eta_S \rbrace_S$ denotes a sequence with $\eta_S \to 0$ as $S \to +\infty$. 

\noindent \emph{Proof of \eqref{ineq:estimate1conv}:} For $x \in X_T\cap (A^+ \cup A^- \cup \partial^+_1 Q^\nu_T\cup\partial^-_1 Q^\nu_T )  \cap Q^\nu_T$ such that $\mathrm{dist}(x, U) >1$, there holds $\#\mathcal{N}(x)=6$. This follows from the boundary conditions of $X_i^{j,k}$ on every cube $Q^{\nu_j}_S(x_i^{j,k})$ and the fact that $X_T=\mathscr{L}(z^\pm)$ in $A^\pm\cup \partial^\pm_1 Q^\nu_T$. Therefore, in order to obtain \eqref{ineq:estimate1conv}, it suffices to estimate the cardinality of the atoms $x\in X_T$ lying in $(U)_1$. As $U$ consists of $2n+1$ tubular neighborhoods of segments whose length is bounded by $CS$, we get $\mathcal{L}^2((U)_2) \le CnS$.   Therefore, employing Lemma \ref{lemma:propertiesofE}(v), we obtain $\#(X_T \cap (U)_1 )\leq CnS$.  By \eqref{def:energyneighbourhood} this implies \eqref{ineq:estimate1conv}.

\noindent \emph{Proof of \eqref{ineq:estimate2conv}:} In view of \eqref{def:XTconvexity}, in order to obtain \eqref{ineq:estimate2conv},  it suffices to estimate the energy contribution of atoms in $\bigcup_{i,j,k}  (Q^{\nu_j}_S(x_i^{j,k})\setminus A_i^{j,k})$. For each  $i,j,k$, there holds $X_T= \mathscr{L}(z^\pm)$ on 
 $$\big(\partial Q^{\nu_j}_S(x_i^{j,k})\big)_5 \setminus \big\{x\colon  \pm \langle x-x_i^{j,k}, \nu_j\rangle  \leq   C \kappa  \big\}$$
 with a constant $C > 0$ only depending on $\nu_1, \nu_2$ and $\nu$. This shows that the cardinality of $X_T \cap Q^{\nu_j}_S(x_i^{j,k}) \cap ((A_i^{j,k})_1 \cup (U)_1)$, which contains all atoms $x \in X_T \cap Q^{\nu_j}_S(x_i^{j,k})$ for which possibly $\#(\mathcal{N}(x)\cap X_T) < \#(\mathcal{N}(x)\cap X_i^{j,k})$, is uniformly controlled due to Lemma \ref{lemma:propertiesofE}(v). We thus obtain $E\big(X_T,Q^{\nu_j}_S(x_{i}^{j,k})\big) \leq E\big(X_i^{j,k},Q^{\nu_j}_S(x_{i}^{j,k})\big) + C$ by \eqref{def:energyneighbourhood}.  Thus, using \eqref{eq:EnergyQijk},  Proposition \ref{proposition:existence}, and Proposition \ref{proposition:xindependence} we get 
\begin{align}\label{ineq:EXijk}
E\big(X_T,Q^{\nu_j}_S(x_{i}^{j,k})\big) \leq E\big(X_i^{j,k},Q^{\nu_j}_S(x_{i}^{j,k})\big) + C \le S\big({\varphi}(z^+,z^-,\nu_j) + \eta_S\big) +C.
\end{align} 
For $j \in \lbrace 1,2\rbrace$, we find 
\begin{align*}
\# \big\{(i,k) \colon \,   i=0,\ldots,N_j(S,T),  \, k=0,\ldots,n-1 \big\}  =  n\bigg(\bigg\lfloor \frac{\lambda_j(T-(10n+5)S) }{nS} \bigg\rfloor + 1 \bigg)  \leq \frac{\lambda_j T}{S}.
\end{align*}
This along with \eqref{ineq:EXijk} yields   \eqref{ineq:estimate2conv}. 

\noindent \emph{Step 3: Conclusion.} Noting that
\begin{align*}
\min\left\{E_1(X,Q^{\nu}_T) \colon  \, X = \mathscr{L}(z^\pm) \text{ on } \partial_1^\pm Q^{\nu}_T \right\}\leq E_1(X_T,Q^\nu_T),
\end{align*}
 and using \eqref{ineq:estimate1conv}--\eqref{ineq:estimate2conv}  as well as Lemma \ref{lemma:propertiesofE}(iv),  we have
\begin{align*}
\min\left\{E_1(X,Q^{\nu}_T) \colon X = \mathscr{L}(z^\pm) \text{ on } \partial_1^\pm Q^{\nu}_T \right\}  & \leq  \lambda_1 T\big({\varphi}(z^+,z^-,\nu_1) + \eta_S\big) +C\lambda_1 T /S \notag \\ & \ \ \ + \lambda_2T\big({\varphi}(z^+,z^-,\nu_2) + \eta_S\big) + C\lambda_2T/S  + CnS.
\end{align*}
Dividing by $T$, letting first $T\to+\infty$, and then $S\to +\infty$, we obtain \eqref{eq:convexitynu} by Proposition \ref{proposition:existence}, where we also  use  $\eta_S \to 0$. This concludes the proof of (iii).

\noindent \emph{Proof of $\mathrm{(iv)}$.} Let $z^\pm=(\theta^\pm,\tau^\pm,1)$, $\nu \in \mathbb{S}^1$, and $\theta \in \mathbb{A}$. Our goal is to prove
\begin{align}\label{ineq:firstframeinvariance}
\varphi\big((\theta^++\theta,\tau^+,1),(\theta^-+\theta,\tau^-,1),e^{i\theta}\nu\big)=\varphi\big((\theta^+,\tau^+,1),(\theta^-,\tau^-,1),\nu\big).
\end{align} 
Due to Proposition \ref{proposition:existence}, for every $T>0$ we can choose $X_T \subset \mathbb{R}^2$, such that $X_T =\mathscr{L}((\theta^\pm,\tau^\pm,1))$ on $\partial_1^\pm Q^\nu_T$ and such that
\begin{align}\label{eq:convergencerotinvariance}
\lim_{T\to+\infty} \frac{1}{T} E_1\big(X_T,Q^\nu_T\big) = \varphi\big((\theta^+,\tau^+,1),(\theta^-,\tau^-,1),\nu\big).
\end{align}
We set $X_T^\theta= e^{i\theta}X_T$. Then $X_T^\theta = \mathscr{L}((\theta^\pm+\theta,\tau^\pm,1))$ on $\partial_1^\pm Q^{\nu_\theta}_T$, where $\nu_\theta= e^{i\theta}\nu$.  Applying  Proposition \ref{proposition:existence}, Lemma \ref{lemma:propertiesofE}(i), and \eqref{eq:convergencerotinvariance}, we obtain 
\begin{align*}
\varphi((\theta^++\theta,\tau^+,1),(\theta^-+\theta,\tau^-,1),e^{i\theta}\nu)&\leq \liminf_{T\to+\infty} \frac{1}{T} E_1(X_T^\theta,Q^{\nu_\theta}_T) \\&= \lim_{T\to +\infty} \frac{1}{T} E_1(X_T,Q^{\nu}_T) = \varphi((\theta^+,\tau^+,1),(\theta^-,\tau^-,1),\nu).
\end{align*}
This implies  one inequality in \eqref{ineq:firstframeinvariance}. The other inequality follows by repeating the argument for $(\tilde{\theta}^\pm,\tau^\pm,1)= (\theta^\pm+\theta,\tau^\pm,1)$, $\tilde{\nu} = e^{i\theta}\nu$, and $\tilde{\theta}=-\theta$. This concludes the proof of (iv). 

\noindent \emph{Proof of $\mathrm{(v)}$.}  Let $z^\pm=(\theta^\pm,\tau^\pm,1)$, $\nu \in \mathbb{S}^1$, and $\tau \in \mathbb{T}$. Our goal is to prove
\begin{align}\label{ineq:firstframeinvariance-new}
\varphi\Big( \big(\theta^+,\tau^++e^{-i\theta^+}\tau ,1\big), \big(\theta^-,\tau^-+e^{-i\theta^-}\tau,1\big),\nu\Big)=\varphi\big((\theta^+,\tau^+,1),(\theta^-\tau^-,1),\nu\big).
\end{align}
Due to Proposition \ref{proposition:existence}, for every $T>0$ we can choose $X_T \subset \mathbb{R}^2$ , such that $X_T =\mathscr{L}((\theta^\pm,\tau^\pm,1))$ on $\partial_1^\pm Q^\nu_T$ and such that  \eqref{eq:convergencerotinvariance} holds. We set $X_T^\tau= X_T+\tau$. Then $X_T^\tau = \mathscr{L}((\theta^\pm,\tau^\pm+e^{-i\theta^{\pm}}\tau,1))$ on $\partial_1^\pm Q^{\nu}_T(\tau)$.  Applying  Proposition \ref{proposition:existence}, Lemma \ref{lemma:propertiesofE}(i), and   \eqref{eq:convergencerotinvariance}, we get
\begin{align*}
\varphi\big(\big(\theta^+,\tau^++e^{-i\theta^+}\tau,1\big),\big(\theta^-,\tau^-+e^{-i\theta^-}\tau,1\big),\nu\big)&\leq \liminf_{T\to+\infty} \frac{1}{T} E_1(X_T^\tau,Q^{\nu}_T(\tau)) = \lim_{T\to +\infty} \frac{1}{T} E_1(X_T,Q^{\nu}_T)  \\& = \varphi\big((\theta^+,\tau^+,1),(\theta^-,\tau^-,1),\nu\big).
\end{align*}
This yields one inequality of \eqref{ineq:firstframeinvariance-new}. The other inequality follows by repeating the argument for    $({\theta}^\pm,\tilde{\tau}^\pm,1) = (\theta^\pm,\tau^\pm+e^{-i\theta^\pm}\tau,1)$ and $\tilde{\tau}=-\tau$. This concludes the proof of (v).
\end{proof}

\section*{Acknowledgements} 

MF and LK acknowledge support by the DFG projects FR 4083/1-1, FR 4083/3-1 and by the Deutsche Forschungsgemeinschaft (DFG, German Research Foundation) under Germany's Excellence Strategy EXC 2044 -390685587, Mathematics M\"unster: Dynamics--Geometry--Structure.

\end{document}